\documentclass{article}
\usepackage[utf8]{inputenc}
\usepackage{amsmath}
\usepackage{amsfonts,mathrsfs}
\usepackage{amssymb,amsthm,xcolor}
\usepackage[pdftex,plainpages=false,hypertexnames=false,pdfpagelabels]{hyperref}
\usepackage{xcolor}
\usepackage{tikz}
\usetikzlibrary{calc}
\usepackage{tocvsec2}
\usepackage{amscd}
\usepackage{caption}
\definecolor{dark-red}{rgb}{0.7,0.25,0.25}
\definecolor{dark-blue}{rgb}{0.15,0.15,0.55}
\definecolor{medium-blue}{rgb}{0,0,.8}
\definecolor{DarkGreen}{RGB}{0,150,0}
\definecolor{rho}{named}{red}
\hypersetup{
   colorlinks, linkcolor={purple},
   citecolor={medium-blue}, urlcolor={medium-blue}
}
\usepackage{fullpage}

\newcommand{\cB}{\mathcal{B}}
\newcommand{\F}{\mathcal{F}}
\newcommand{\cM}{\mathcal{M}}
\newcommand{\cH}{\mathcal{H}}
\newcommand{\scG}{\mathscr{G}}
\newcommand{\K}{\mathcal{K}}
\newcommand{\cK}{\mathcal{K}}
\newcommand{\cC}{\mathcal{C}}
\newcommand{\C}{\mathbb{C}}
\newcommand{\cA}{\mathcal{A}}

\newcommand{\cI}{\mathcal{I}}
\newcommand{\cU}{\mathcal{U}}

\newcommand{\cPU}{\mathcal{PU}}
\newcommand{\bbP}{\mathbb{P}}
\newcommand{\g}{\mathfrak{g}}
\renewcommand{\Re}{\operatorname{Re}}
\newcommand{\Aut}{\operatorname{Aut}}

\newcommand{\bbA}{\mathbb{A}}
\newcommand{\R}{\mathbb{R}}
\newcommand{\Z}{\mathbb{Z}}

\newcommand{\fru}{\mathfrak{u}}
\newcommand{\cl}{\operatorname{cl}}
\newcommand{\im}{\operatorname{im}}
\newcommand{\abs}[1]{\left|#1\right|}
\newcommand{\norm}[1]{\left\|#1\right\|}
\newcommand{\ip}[1]{\left\langle#1\right\rangle}
\newcommand{\D}{\mathbb{D}}
\newcommand{\cD}{\mathcal{D}}
\newcommand{\cDR}{\mathcal{DR}}
\newcommand{\cDA}{\mathcal{DA}}
\newcommand{\cDP}{\mathcal{DP}}
\newcommand{\cO}{\mathcal{O}}

\newcommand{\cR}{\mathcal{R}}
\newcommand{\Diff}{\operatorname{Diff}}
\newcommand{\CAR}{\operatorname{CAR}}
\newcommand{\Span}{\operatorname{span}}
\newcommand{\End}{\operatorname{End}}
\newcommand{\Vect}{\operatorname{Vect}}
\newcommand{\Vir}{\operatorname{Vir}}
\newcommand{\Mob}{\operatorname{M\ddot{o}b}}
\newcommand{\Id}{1}
\newcommand{\grotimes}{\hat \otimes}
\newcommand{\supp}{\operatorname{supp}}
\newcommand{\tr}{\operatorname{tr}}

\newcommand{\interior}[1]{\mathring{#1}}
\newcommand{\hOmega}{\hat \Omega}

\newcommand{\admisall}{\widetilde{\mathscr{A}}}
\newcommand{\admis}{\mathscr{A}}

\newtheorem{thmalpha}{Theorem}

\newtheorem{Theorem}{Theorem}[section]
\newtheorem*{Theorem*}{Theorem}
\newtheorem{Lemma}[Theorem]{Lemma}
\newtheorem{Proposition}[Theorem]{Proposition}
\newtheorem{Corollary}[Theorem]{Corollary}
\theoremstyle{definition}
\newtheorem{Remark}[Theorem]{Remark}
\newtheorem{Definition}[Theorem]{Definition}
\newtheorem*{Definition*}{Definition}
\newtheorem{Example}[Theorem]{Example}

\numberwithin{equation}{section}
\numberwithin{figure}{section}

\newcommand{\Addresses}{{% additional braces for segregating \footnotesize
  \bigskip
  \footnotesize

  J.~Tener, \textsc{Mathematical Sciences Institute, The Australian National University,
    Canberra}\par\nopagebreak
  \textit{E-mail address}: \texttt{james.tener@anu.edu.au}
}}

\begin{document}

\title{Geometric realization of algebraic conformal field theories}
\author{James E. Tener}
\date{}
%\date{\today}

\maketitle

\abstract{
We explore new connections between the fields and local observables in two dimensional chiral conformal field theory.
We show that in a broad class of examples, the von Neumann algebras of local observables (a conformal net) can be obtained from the fields (a unitary vertex operator algebra) via a continuous geometric interpolation procedure involving Graeme Segal's functorial definition of conformal field theory, and that the conformal net may be thought of as a boundary value of the Segal CFT.
In particular, we construct conformal nets from these unitary vertex operator algebras by showing that `geometrically mollified' versions of the fields yield \emph{bounded}, \emph{local} observables on the Hilbert space completion of the vertex algebra.
These are the first results which unite the three major definitions of chiral conformal field theory. 
This work is inspired by Henriques' picture of conformal nets arising from degenerate Riemann surfaces.
}

\bigskip\bigskip

\tableofcontents

\newpage

\settocdepth{section}
\section{Introduction} 

There are three major mathematical formulations of two dimensional chiral conformal field theory (CFT).
On the algebraic side, we have the notion of a \emph{vertex operator algebra} (VOA), which axiomatizes the fields of a chiral CFT.
In the language of functional analysis and operator algebras, we have \emph{conformal nets}, which axiomatize the algebras of local observables (in the sense of Haag-Kastler algebraic quantum field theory).
These two notions have been more extensively developed than the third formulation, Graeme Segal's geometric definition in terms of functors from the two dimensional complex bordism category (Segal CFT).

It is widely believed that the three approaches are essentially equivalent, after imposing some technical conditions, and perhaps restricting the Segal formulation to bordisms with genus zero.
Since all three definitions are supposed to capture the same physical notion of 2d chiral conformal field theory, each has a version of the major examples (e.g. minimal models, WZW models) and constructions (e.g. coset construction, orbifold construction), and it would be very satisfying to have a robust theory which identifies the three manifestations of these.

More importantly, each of these formulations has important and interesting connections within mathematics, for example the connection between conformal nets and Jones' theory of subfactors, or the connection between vertex operator algebras and `monstrous moonshine.'
There are many examples of important results in conformal field theory which can be established in one of the frameworks but not the others\footnote{
For example, the rationality of orbifolds and cosets is an open problem in the theory of VOAs which has been solved in the context of conformal nets, whereas the rationality of many important examples has been established for VOAs, but not for conformal nets},
and it is very desirable to develop the connection between different formulations of CFT to the point that one may answer open questions about one version using a result from another.
One striking example of the value of this approach is Wassermann's computation of the fusion rules for the $SU(N)_k$ conformal nets using smeared primary fields \cite{Wa98}, which provided a natural construction of subfactors with index $4 \cos^2 \frac{\pi}{n}$.

Recently, Carpi, Kawahigashi, Longo and Weiner initiated a general theory relating vertex operator algebras and conformal nets \cite{CKLW18}.
They give a construction which produces a conformal net from a (simple, unitary) vetex operator algebra satisfying regularity conditions, which they show are satisfied by essentially every known vertex operator algebra.
Moreover, they show how to recover the vertex operator algebra from the conformal net that it produces.

In this paper, we will present an alternative, geometric perspective on the relationship between vertex operator algebras and conformal nets, based on a geometric picture of conformal nets introduced by Andr\'e Henriques \cite{Henriques14}.
We will show that, in a broad class of examples, Segal's functorial definition of conformal field theory allows one to continuously interpolate between unitary vertex operator algebras and conformal nets.
To our knowledge, these are the first results which unite the three definitions of conformal field theory.

We will now outline Henriques' geometric picture of conformal nets in more detail.
The dictionary between vertex operator algebras and the geometric picture of Segal CFT has long been understood by mathematicians and physicists.
In Segal's picture, there is a a Hilbert space $\cH$ assigned to the circle $S^1$, and to every two dimensional complex bordism $\Sigma$ there is a one dimensional space of trace class linear maps 
$$
E(\Sigma): \bigotimes_{\pi_0(\partial \Sigma^{in})} \cH \to \bigotimes_{\pi_0(\partial \Sigma^{out})} \cH.
$$
In particular, one map $T:\cH \otimes \cH \to \cH$ assigned to a disk with two disks removed corresponds to the the state-field correspondence $a \mapsto Y(a,w)$ of a vertex operator algebra. 
More precisely, we have
$$
\begin{tikzpicture}[baseline={([yshift=-.5ex]current bounding box.center)}]
% BIG DISK
	\filldraw[fill=red!10!blue!20!gray!30!white] (0,0) circle (1cm);
%	\node at (0:1cm) {\scriptsize{\textbullet}};
%	\node at (0:1.15cm) {1};
	%\draw (0,0) circle (1cm);
% CENTERED INNER DISK
	\filldraw[fill=white] (0,0) circle (0.3cm);
%	\node at (0,0) {\scriptsize{\textbullet}};
%	\node at (0:.15cm) {0};
%	\node at (0:.3cm) {\scriptsize{\textbullet}};
%	\node at (0:.45cm) {r};
	%\draw (0,0) circle (0.3cm);
% OFF CENTER INNER  DISK
	\filldraw[fill=white] (140:.65cm) circle (.25cm); 
%	\node at (140:.65cm) {\scriptsize{\textbullet}};
%	\node at (150:.6cm) {w};
%	\node at ([shift=(0:0.25cm)]140:.65cm) {\scriptsize{\textbullet}};
%	\node at ([shift=(20:0.6cm)]153:.7cm) {w+s};
	%\node at (140:.65cm) {\scriptsize{$\xi$}};
\end{tikzpicture} 
\qquad \longleftrightarrow \qquad
\begin{array}{ll}
T:\cH \otimes \cH \to \cH, \smallskip \\  T(a \otimes b) = Y(s^{L_0} a, w)r^{L_0} b
\end{array}
.
$$
where the parameters $s, w$ and $r$ are determined by the geometry of the surface, and $L_0$ is the energy operator.

The symmetry group of a chiral conformal field theory is the group of orientation preserving diffeomorphisms of the unit circle, $\Diff_+(S^1)$.
It is common practice to think of these diffeomorphisms as `thin' bordisms, i.e. as degenerate annuli with zero thickness.
Henriques' idea in \cite{Henriques14} is to also consider degenerate annuli which are thin along only part of the  boundary, such as the ones in Figure \ref{figIntroDegenerateAnnuli}.
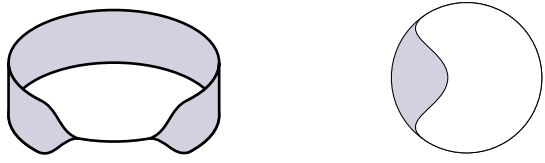
\begin{figure}[!ht]
$$
\tikz[scale=.7,,baseline={([yshift=-.5ex]current bounding box.center)}]{
\coordinate (a) at ($(0,0)+(-45:2cm and 1cm)$);
\coordinate (b) at ($(0,0)+(225:2cm and 1cm)$);
\coordinate (c) at ($(0,1)+(-45:2cm and 1cm)$);
\coordinate (d) at ($(0,1)+(225:2cm and 1cm)$);
\coordinate (e) at ($(0,.5)+(-70:2cm and 1cm)$);
\coordinate (f) at ($(0,.5)+(-110:2cm and 1cm)$);
\fill[red!10!blue!20!gray!30!white] (-2,0) arc (180:0:2cm and 1cm) -- +(0,1) arc (0:180:2cm and 1cm) -- cycle;
\draw[line width=1](0,0)+(-45:2cm and 1cm) arc (-45:180+45:2cm and 1cm)
(0,1)+(-45:2cm and 1cm) arc (-45:180+45:2cm and 1cm)
(0,.5)+(-70:2cm and 1cm) arc (-70:-110:2cm and 1cm);
\filldraw[line width=1, fill=red!10!blue!20!gray!30!white]
(2,0) arc (0:-45:2cm and 1cm) to[out=195, in=0] (e) to[out=20, in=203] (c) arc (-45:0:2cm and 1cm) -- cycle
(-2,0) arc (180:225:2cm and 1cm) to[out=-20, in=180] (f) to[out=160, in=-23] (d) arc (225:180:2cm and 1cm) -- cycle;
}
\qquad \qquad \qquad
\begin{tikzpicture}[baseline={([yshift=-.5ex]current bounding box.center)}]
	\coordinate (a) at (120:1cm);
	\coordinate (b) at (240:1cm);
	\coordinate (c) at (180:.25cm);
% BIG DISK
	\fill[fill=red!10!blue!20!gray!30!white] (0,0) circle (1cm);
	\draw (0,0) circle (1cm);
% CURVED BOUNDARY REGION
	\fill[fill=white] (a)  .. controls ++(210:.6cm) and ++(90:.4cm) .. (c) .. controls ++(270:.4cm) and ++(150:.6cm) .. (b) -- ([shift=(240:1cm)]0,0) arc (240:480:1cm);
	\draw ([shift=(240:1cm)]0,0) arc (240:480:1cm);
	\draw (a) .. controls ++(210:.6cm) and ++(90:.4cm) .. (c);
	\draw (b) .. controls ++(150:.6cm) and ++(270:.4cm) .. (c);
% INNER DISK
%	\filldraw[fill=white] (180:.65cm) circle (.25cm); 
% COORDINATE LABELS
%	\node at (a) {(a)};
%	\node at (b) {(b)};
%	\node at (c) {(c)};
\end{tikzpicture}
$$
\captionsetup{justification=centering,width=0.8\linewidth}
\caption{A pair of degenerate annuli, one (from \cite{Henriques14}) depicted in three space, and another embedded in the complex plane.}
\label{figIntroDegenerateAnnuli}
\end{figure}
Given a Segal CFT, one might hope that it assigns bounded linear maps $\cH \to \cH$ to such degenerate annuli, although these maps will no longer be trace class.
We should be able to obtain the linear maps for degenerate annuli as limits of maps assigned to ordinary annuli:
\begin{equation}\label{eqnIntroSurfaceLimit}
E\left(
\begin{tikzpicture}[baseline={([yshift=-.5ex]current bounding box.center)}]
	\coordinate (a) at (120:1cm);
	\coordinate (b) at (240:1cm);
	\coordinate (c) at (180:.25cm);
% BIG DISK
	\fill[fill=red!10!blue!20!gray!30!white] (0,0) circle (1cm);
	\draw (0,0) circle (1cm);
% CURVED BOUNDARY REGION
	\fill[fill=white] (a)  .. controls ++(210:.6cm) and ++(90:.4cm) .. (c) .. controls ++(270:.4cm) and ++(150:.6cm) .. (b) -- ([shift=(240:1cm)]0,0) arc (240:480:1cm);
	\draw ([shift=(240:1cm)]0,0) arc (240:480:1cm);
	\draw (a) .. controls ++(210:.6cm) and ++(90:.4cm) .. (c);
	\draw (b) .. controls ++(150:.6cm) and ++(270:.4cm) .. (c);
% POINT LABELS
	\node at (0:1cm) {\scriptsize{\textbullet}};
	\node at (0:0.8cm) {1};
%	\node at (0:1.2cm) {\scriptsize{\textbullet}};
%	\node at (0:1.4cm) {R};
% INNER DISK
%	\filldraw[fill=white] (180:.65cm) circle (.25cm); 
% COORDINATE LABELS
%	\node at (a) {(a)};
%	\node at (b) {(b)};
%	\node at (c) {(c)};
\end{tikzpicture}\right)
 = \lim_{R \downarrow 1} E\left(
\begin{tikzpicture}[baseline={([yshift=-.5ex]current bounding box.center)}]
	\coordinate (a) at (120:1cm);
	\coordinate (b) at (240:1cm);
	\coordinate (c) at (180:.25cm);
% BIG DISK
	\fill[fill=red!10!blue!20!gray!30!white] (0,0) circle (1.2cm);
	\draw (0,0) circle (1.2cm);
% CURVED BOUNDARY REGION
	\fill[fill=white] (a)  .. controls ++(210:.6cm) and ++(90:.4cm) .. (c) .. controls ++(270:.4cm) and ++(150:.6cm) .. (b) -- ([shift=(240:1cm)]0,0) arc (240:480:1cm);
	\draw ([shift=(240:1cm)]0,0) arc (240:480:1cm);
	\draw (a) .. controls ++(210:.6cm) and ++(90:.4cm) .. (c);
	\draw (b) .. controls ++(150:.6cm) and ++(270:.4cm) .. (c);
% POINT LABELS
	\node at (0:1cm) {\scriptsize{\textbullet}};
	\node at (0:0.8cm) {1};
	\node at (0:1.2cm) {\scriptsize{\textbullet}};
	\node at (0:1.4cm) {R};
% INNER DISK
%	\filldraw[fill=white] (180:.65cm) circle (.25cm); 
% COORDINATE LABELS
%	\node at (a) {(a)};
%	\node at (b) {(b)};
%	\node at (c) {(c)};
\end{tikzpicture}
\right)
\,\, .
\end{equation}

The principal piece of data for a conformal net is a family of von Neumann algebras $\cA(I)$, called the local algebras, indexed by intervals $I \subset S^1$.
In Henriques' geometric perspective, the local operators of a conformal net correspond to degenerate annuli with states inserted in the thick part of the annulus. That is, $\cA(I)$ is generated by degenerate surfaces which look like:
\begin{equation}\label{eqnVagueDegeneratePants}
\begin{tikzpicture}[baseline={([yshift=-.5ex]current bounding box.center)}]
	\coordinate (a) at (120:1cm);
	\coordinate (b) at (240:1cm);
	\coordinate (c) at (180:.25cm);
% BIG DISK
	\fill[fill=red!10!blue!20!gray!30!white] (0,0) circle (1cm);
	\draw (0,0) circle (1cm);
% CURVED BOUNDARY REGION
	\fill[fill=white] (a)  .. controls ++(210:.6cm) and ++(90:.4cm) .. (c) .. controls ++(270:.4cm) and ++(150:.6cm) .. (b) -- ([shift=(240:1cm)]0,0) arc (240:480:1cm);
	\draw ([shift=(240:1cm)]0,0) arc (240:480:1cm);
	\draw (a) .. controls ++(210:.6cm) and ++(90:.4cm) .. (c);
	\draw (b) .. controls ++(150:.6cm) and ++(270:.4cm) .. (c);
% INNER DISK
	\filldraw[fill=white] (180:.65cm) circle (.25cm); 
	\node at (180:.65cm) {$a$};
% POINT LABELS
%	\node at (0:1cm) {\scriptsize{\textbullet}};
%	\node at (0:0.8cm) {1};
%	\node at (0:1.2cm) {\scriptsize{\textbullet}};
%	\node at (0:1.4cm) {R};
%	\node at (180:.65cm) {\scriptsize{\textbullet}};
%	\node at (193:.7cm) {w};
% INTERVAL LABEL
	\draw (130:1.2cm) -- (130:1.4cm);
	\draw (230:1.2cm) -- (230:1.4cm);
	\draw (130:1.3cm) arc (130:230:1.3cm);
	\node at (180:1.5cm) {\scriptsize{$I$}};
% COORDINATE LABELS
%	\node at (a) {(a)};
%	\node at (b) {(b)};
%	\node at (c) {(c)};
\end{tikzpicture}
\,\,,
\end{equation}
where $a$ runs over all states.
Thus a conformal net can be thought of as a boundary value of a Segal CFT via a limiting procedure like the one in \eqref{eqnIntroSurfaceLimit}.

The content of this paper is that these ideas can be made rigorous in a large family of examples, namely those examples which can be embedded in some number of complex free fermions.
Moreover, we show that Segal CFT can be used to interpolate between vertex operator algebras and conformal nets.
We'll now outline our main results.

In \cite{Ten16}, we gave a construction of the Segal CFT for the free fermion, which assigns to a circle the fermionic Fock Hilbert space $\F$, and to a Riemann surface $X$ equipped with a spin structure, trivialized on the boundary, a one dimensional space of trace class maps 
$$
E(X): \bigotimes_{\pi_0(\partial X^{in})} \F \to \bigotimes_{\pi_0(\partial X^{out})} \F.
$$
The maps $T \in E(X)$ are characterized by certain commutation relations, determined by the Hardy space $H^2(X)$, between $T$ and generators $a(f)$ and $a(g)^*$ of the canonical anticommutation relations algebra $\CAR(L^2(S^1))$.
This construction has many nice properties (discussed in Section \ref{subsecFermionicFockSpace}), the most important of which is the compatibilty between gluing of Riemann surfaces and composition of linear maps.
We also proved in \cite{Ten16} that
\begin{equation}\label{eqnIntroPants}
E\left(
\begin{tikzpicture}[scale=.5,baseline={([yshift=-.5ex]current bounding box.center)}]
% BIG DISK
	\filldraw[fill=red!10!blue!20!gray!30!white] (0,0) circle (1cm);
%	\node at (0:1cm) {\scriptsize{\textbullet}};
%	\node at (0:1.15cm) {1};
	%\draw (0,0) circle (1cm);
% CENTERED INNER DISK
	\filldraw[fill=white] (0,0) circle (0.3cm);
%	\node at (0,0) {\scriptsize{\textbullet}};
%	\node at (0:.15cm) {0};
%	\node at (0:.3cm) {\scriptsize{\textbullet}};
%	\node at (0:.45cm) {r};
	%\draw (0,0) circle (0.3cm);
% OFF CENTER INNER  DISK
	\filldraw[fill=white] (140:.65cm) circle (.25cm); 
%	\node at (140:.65cm) {\scriptsize{\textbullet}};
%	\node at (150:.6cm) {w};
%	\node at ([shift=(0:0.25cm)]140:.65cm) {\scriptsize{\textbullet}};
%	\node at ([shift=(20:0.6cm)]153:.7cm) {w+s};
	%\node at (140:.65cm) {\scriptsize{$\xi$}};
\end{tikzpicture} 
\right)
\,\, = \,\, \Span_\C T
\end{equation}
where $T:\F \otimes \F \to \F$ is given on finite energy vectors $a,b \in \F^0$ by $$T(a \otimes b) = Y(s^{L_0}a, w)r^{L_0}b.$$ Here, $Y$ is the free fermion state-field correspondence and $s,w,r$ are such that the surface in question is $\D \setminus (r\interior{\D} \cup (w + s \interior{\D}))$\footnote{We will use $\D$ for the closed unit disk in $\C$, and $\interior{\D}$ for its interior}.
The spin structure in \eqref{eqnIntroPants} is the one inherited from the depicted embedding into $\C$, and the boundary trivializations of this spin structure are the ones obtained from the Riemann maps $z \mapsto rz$ and $z \mapsto w + sz$ for the regions $r\interior{\D}$ and $w + s \interior{\D}$ removed from the unit disk $\D$ (along with suitable choices of square roots of the derivatives of these maps).

Now consider a family of Riemann surfaces of the form
$$
X_{R,t} = \begin{tikzpicture}[baseline={([yshift=-.5ex]current bounding box.center)}]
	\coordinate (a) at (120:1cm);
	\coordinate (b) at (240:1cm);
	\coordinate (c) at (180:.25cm);
% BIG DISK
	\fill[fill=red!10!blue!20!gray!30!white] (0,0) circle (1.2cm);
	\draw (0,0) circle (1.2cm);
% CURVED BOUNDARY REGION
	\fill[fill=white] (a)  .. controls ++(210:.6cm) and ++(90:.4cm) .. (c) .. controls ++(270:.4cm) and ++(150:.6cm) .. (b) -- ([shift=(240:1cm)]0,0) arc (240:480:1cm);
	\draw ([shift=(240:1cm)]0,0) arc (240:480:1cm);
	\draw (a) .. controls ++(210:.6cm) and ++(90:.4cm) .. (c);
	\draw (b) .. controls ++(150:.6cm) and ++(270:.4cm) .. (c);
% POINT LABELS
	\node at (0:1cm) {\scriptsize{\textbullet}};
	\node at (0:0.8cm) {1};
	\node at (0:1.2cm) {\scriptsize{\textbullet}};
	\node at (0:1.4cm) {R};
% INNER DISK
	\filldraw[fill=white] (180:.65cm) circle (.25cm); 
% COORDINATE LABELS
%	\node at (a) {(a)};
%	\node at (b) {(b)};
%	\node at (c) {(c)};
\end{tikzpicture}
\,\, = \,\, R \D \setminus ( \phi_t(\interior{\D}) \cup (w + s \interior{\D}) ),
$$
where $R > 1$ and $(\phi_t)_{t \ge 0}$ is a one-parameter semigroup of univalent (i.e. holomorphic, injective) self maps of the unit disk $\D$ which fix $0$ and map onto Jordan domains with $C^\infty$ boundary.\footnote{
The requirement that the region removed from $\D$ be of the form $\phi_t(\interior{\D})$ for a semigroup $\phi_t$ is for technical technical reasons, but we expect that this assumption is not essential.}
There is a unique univalent map $\sigma:\interior{\D} \to \C$, called the \emph{Koenigs map} of $\phi_t$, which satisfies Schr\"oder's equation $\sigma(\phi_t(z)) = \phi_t^\prime(0) \sigma(z)$ for all $t \ge 0$ and $z \in \interior{\D}$.
We assume that $\sigma$ extends smoothly to the boundary $S^1$ of $\interior{\D}$.

Let $L_n$ be the unitary, positive energy representation of the Virasoro algebra for the free fermion, and let
$$
L(\rho) = \sum_{n = -\infty}^\infty \hat \rho_n L_n
$$
be the smeared field corresponding to the function $\rho(z) = \frac{\sigma(z)}{z \sigma^\prime(z)}$, where $\hat \rho_n$ are the Fourier coefficients of the restriction of $\rho$ to $S^1$.\footnote{
The function $\rho$ is closely related to the generating vector field of the semigroup $\phi_t$ of Berkson and Porta \cite{BerksonPorta78}.} 

One can verify that the space assigned by the (non-degenerate) free fermion Segal CFT, $E(X_{R,t})$, is spanned by the map $T_{R,t}$ given on finite energy vectors $a,b$ by 
$$
T_{R,t}(a \otimes b) = R^{-L_0}Y(s^{L_0}a,w)e^{-tL(\rho)}b,
$$
when $X_{R,t}$ is given the \emph{standard spin structure} inherited from $\C$, and \emph{standard boundary trivializations} induced by the Riemann maps $\phi_t$ and $z \mapsto sz + w$ (along with appropriate choices of square roots of their derivatives), after perhaps reparametrizing the semigroup $\phi_t \mapsto \phi_{\alpha t}$.
The following theorem, characterizing the value of the Segal CFT on the degenerate boundary limit $\lim_{R \downarrow 1} X_{R,t}$ (as in \eqref{eqnIntroSurfaceLimit}) is stated more precisely in the body of the paper as Theorem \ref{thmBoundednessAndExistence}.

\begin{thmalpha}\label{thmIntroFermionLimitThm}
Let 
$$
T_{R,t} \in E\left( \begin{tikzpicture}[scale=.4,baseline={([yshift=-.5ex]current bounding box.center)}]
	\coordinate (a) at (120:1cm);
	\coordinate (b) at (240:1cm);
	\coordinate (c) at (180:.25cm);
% BIG DISK
	\fill[fill=red!10!blue!20!gray!30!white] (0,0) circle (1.2cm);
	\draw (0,0) circle (1.2cm);
% CURVED BOUNDARY REGION
	\fill[fill=white] (a)  .. controls ++(210:.6cm) and ++(90:.4cm) .. (c) .. controls ++(270:.4cm) and ++(150:.6cm) .. (b) -- ([shift=(240:1cm)]0,0) arc (240:480:1cm);
	\draw ([shift=(240:1cm)]0,0) arc (240:480:1cm);
	\draw (a) .. controls ++(210:.6cm) and ++(90:.4cm) .. (c);
	\draw (b) .. controls ++(150:.6cm) and ++(270:.4cm) .. (c);
% POINT LABELS
%	\node at (0:1cm) {\scriptsize{\textbullet}};
%	\node at (0:0.8cm) {1};
%	\node at (0:1.2cm) {\scriptsize{\textbullet}};
%	\node at (0:1.4cm) {R};
% INNER DISK
	\filldraw[fill=white] (180:.65cm) circle (.25cm); 
% COORDINATE LABELS
%	\node at (a) {(a)};
%	\node at (b) {(b)};
%	\node at (c) {(c)};
\end{tikzpicture}
\right)
$$ 
be as above, and fix $t > 0$. 
%Suppose $s > 0$ and $w \in \interior{\D}$, and that $w + s\D \subset \interior{\D} \setminus \phi_t(\D)$.
Then $\lim_{R \downarrow 1} T_{R,t}$ converges to a bounded operator $T_t:\F \otimes \F \to \F$ in the strong operator topology, given on finite energy vectors $a,b$ by
$$
T_t(a \otimes b) = Y(s^{L_0}a,w)e^{-tL(\rho)}b,
$$
where $Y$ is the free fermion state-field correspondence.
Moreover, $T_t$ can be characterized in terms of commutation relations with generators for the CAR algebra determined by the Hardy space of a degenerate Riemann surface $X_t = \D \setminus ((w + s \interior{\D}) \cup \phi_t(\interior{\D}))$, depicted:
$$
X_t=\begin{tikzpicture}[scale=.4,baseline={([yshift=-.5ex]current bounding box.center)}]
	\coordinate (a) at (120:1cm);
	\coordinate (b) at (240:1cm);
	\coordinate (c) at (180:.25cm);
% BIG DISK
	\fill[fill=red!10!blue!20!gray!30!white] (0,0) circle (1cm);
	\draw (0,0) circle (1cm);
% CURVED BOUNDARY REGION
	\fill[fill=white] (a)  .. controls ++(210:.6cm) and ++(90:.4cm) .. (c) .. controls ++(270:.4cm) and ++(150:.6cm) .. (b) -- ([shift=(240:1cm)]0,0) arc (240:480:1cm);
	\draw ([shift=(240:1cm)]0,0) arc (240:480:1cm);
	\draw (a) .. controls ++(210:.6cm) and ++(90:.4cm) .. (c);
	\draw (b) .. controls ++(150:.6cm) and ++(270:.4cm) .. (c);
% POINT LABELS
%	\node at (0:1cm) {\scriptsize{\textbullet}};
%	\node at (0:0.8cm) {1};
%	\node at (0:1.2cm) {\scriptsize{\textbullet}};
%	\node at (0:1.4cm) {R};
% INNER DISK
	\filldraw[fill=white] (180:.65cm) circle (.25cm); 
% COORDINATE LABELS
%	\node at (a) {(a)};
%	\node at (b) {(b)};
%	\node at (c) {(c)};
\end{tikzpicture}
\,\, .
$$
\end{thmalpha}
See Section \ref{secSegalCFTForDegenerateSurfaces} for a precise definition of the Hardy space of a degenerate Riemann surface and of the commutation relations which characterize $T_t$.
The most difficult part of the proof of Theorem \ref{thmIntroFermionLimitThm} is to show that the limit operator $T_t$ is bounded.
Our approach is to show that $T_t^*$ is an example of what we call an \emph{implementing operator} (defined in Section \ref{secImplementingOperators}), which in this case means that it arises as the second quantization of a bounded, not necessarily contractive, map $L^2(S^1) \to L^2(S^1) \oplus L^2(S^1)$.
We then prove that $T_t^*$ is bounded by combining a careful study of the boundedness of implementing operators in general (Theorem \ref{thmAdmissibleBoundedness}) with the `quantum energy inequality' of Fewster and Hollands \cite{FewsterHollands} for smeared Virasoro fields on the circle.
The boundedness of the maps $T_t$ is closely related to the concept of `local energy bounds' for fields, which will be appear in the forthcoming paper \cite{CarpiWeinerLocal}.

Theorem \ref{thmIntroFermionLimitThm} characterizes the value of the free fermion Segal CFT on degenerate Riemann surfaces with standard boundary trivializations, but one can check that changing the trivializations by (spin) diffeomorphisms $\gamma$ of $S^1$ corresponds to composition with a certain unitary representation $U(\gamma)$.
Given a fixed semigroup $(\phi_t)_{t \ge 0}$ as above, and a fixed choice of $t > 0$, we will be interested in pairs $(\gamma_1, \gamma_2)$ of spin diffeomorphisms which satisfy $\gamma_1(z) = \phi_t(\gamma_2(z))$ for all $z$ lying in some interval $I \subset S^1$; let $\mathscr{P}_{I}$ be the collection of all such pairs.

\begin{thmalpha}\label{thmIntroSubalgebraThm}
Let $(\phi_t)_{t \ge 0}$ be a one-parameter group of univalent maps $\phi_t: \D \to \D$ as above, and let $\rho$ be its generator.
Fix $t > 0$, and assume that $\phi_t(S^1) \cap S^1$ contains an interval.
Let $\F^0$ be the finite energy vectors of the free fermion vertex operator algebra, regarded as a subspace of its Hilbert space completion $\F$.
Let $V \subset (\F^0)^{\otimes N}$ be a unitary vertex operator subalgebra, and let $Y:V \to \End(V)[[x^{\pm 1}]]$ be its state-field correspondence.
For $a \in V$, let $T_{t;a} = Y(s^{L_0}a,w)e^{-tL(\rho)}$.
Then $T_{t;a}$ is bounded, and
$$
\cA_V(I) := \big\{ \,U(\gamma_1)^*T_{t,a}U(\gamma_2)\,,\, (U(\gamma_1)^*T_{t;a}U(\gamma_2))^* \,\, \big| \,\, a \in V, \,(\gamma_1,\gamma_2) \in \mathscr{P}_I\big\}^{\prime\prime}
$$
defines a conformal net on the Hilbert space completion of $V$, with conformal symmetry $U$ given by the positive energy representation of $\Diff_+(S^1)$ induced by the conformal vector of $V$.
\end{thmalpha}
Theorem \ref{thmFermiSubalgebrasAreGood} gives a more detailed statement of Theorem \ref{thmIntroSubalgebraThm} which also addresses vertex operator \emph{super}algebras $V$ (which produce \emph{Fermi} conformal nets).
While the results of \cite{CKLW18} are not stated for superalgebras and Fermi conformal nets, they should still hold in that case with minimal modification.
Assuming the `super version' of these results, one can show that the conformal nets constructed in Theorem \ref{thmIntroSubalgebraThm} are isomorphic to the ones constructed in \cite{CKLW18} (see Remark \ref{rmkCKLWComparison}).
Note that while we will cite several results from \cite{CKLW18} on the structure of unitary vertex operator (super)algebras, our construction of conformal nets is entirely independent.
We call our construction a `geometric realization', as the generators of local algebras arise as limits of a Segal CFT, which may be depicted as degenerate Riemann surfaces with states inserted.
We outline our construction in Figure \ref{figFieldsToAlgebras}.
\begin{figure}[!ht]
$$
\begin{array}{ccccc}
\begin{tikzpicture}[baseline={([yshift=-.5ex]current bounding box.center)}]
% BIG DISK
	\filldraw[fill=red!10!blue!20!gray!30!white] (0,0) circle (1cm);
%	\node at (0:1cm) {\scriptsize{\textbullet}};
%	\node at (0:1.15cm) {1};
	%\draw (0,0) circle (1cm);
% CENTERED INNER DISK
	\filldraw[fill=white] (0,0) circle (0.3cm);
%	\node at (0,0) {\scriptsize{\textbullet}};
%	\node at (0:.15cm) {0};
%	\node at (0:.3cm) {\scriptsize{\textbullet}};
%	\node at (0:.45cm) {r};
	%\draw (0,0) circle (0.3cm);
% OFF CENTER INNER  DISK
	\filldraw[fill=white] (160:.65cm) circle (.25cm); 
	\node at (160:.65cm) {$a$};
%	\node at (150:.6cm) {w};
%	\node at ([shift=(0:0.25cm)]140:.65cm) {\scriptsize{\textbullet}};
%	\node at ([shift=(20:0.6cm)]153:.7cm) {w+s};
	%\node at (140:.65cm) {\scriptsize{$\xi$}};
\end{tikzpicture} 
&
\quad
\leadsto
\quad
&
\begin{tikzpicture}[scale=.9,baseline={([yshift=-.5ex]current bounding box.center)}]
	\coordinate (a) at (120:1cm);
	\coordinate (b) at (240:1cm);
	\coordinate (c) at (180:.25cm);
% BIG DISK
	\fill[fill=red!10!blue!20!gray!30!white] (0,0) circle (1.2cm);
	\draw (0,0) circle (1.2cm);
% CURVED BOUNDARY REGION
	\fill[fill=white] (a)  .. controls ++(210:.6cm) and ++(90:.4cm) .. (c) .. controls ++(270:.4cm) and ++(150:.6cm) .. (b) -- ([shift=(240:1cm)]0,0) arc (240:480:1cm);
	\draw ([shift=(240:1cm)]0,0) arc (240:480:1cm);
	\draw (a) .. controls ++(210:.6cm) and ++(90:.4cm) .. (c);
	\draw (b) .. controls ++(150:.6cm) and ++(270:.4cm) .. (c);
% POINT LABELS
%	\node at (0:1cm) {\scriptsize{\textbullet}};
%	\node at (0:0.8cm) {1};
%	\node at (0:1.2cm) {\scriptsize{\textbullet}};
%	\node at (0:1.4cm) {R};
% INNER DISK
	\filldraw[fill=white] (180:.65cm) circle (.25cm); 
	\node at (180:.65) {$a$};
% COORDINATE LABELS
%	\node at (a) {(a)};
%	\node at (b) {(b)};
%	\node at (c) {(c)};
\end{tikzpicture}
&
\qquad
\leadsto
&
\begin{tikzpicture}[baseline={([yshift=-.5ex]current bounding box.center)}]
	\coordinate (a) at (120:1cm);
	\coordinate (b) at (240:1cm);
	\coordinate (c) at (180:.25cm);
% BIG DISK
	\fill[fill=red!10!blue!20!gray!30!white] (0,0) circle (1cm);
	\draw (0,0) circle (1cm);
% CURVED BOUNDARY REGION
	\fill[fill=white] (a)  .. controls ++(210:.6cm) and ++(90:.4cm) .. (c) .. controls ++(270:.4cm) and ++(150:.6cm) .. (b) -- ([shift=(240:1cm)]0,0) arc (240:480:1cm);
	\draw ([shift=(240:1cm)]0,0) arc (240:480:1cm);
	\draw (a) .. controls ++(210:.6cm) and ++(90:.4cm) .. (c);
	\draw (b) .. controls ++(150:.6cm) and ++(270:.4cm) .. (c);
% POINT LABELS
%	\node at (0:1cm) {\scriptsize{\textbullet}};
%	\node at (0:0.8cm) {1};
%	\node at (0:1.2cm) {\scriptsize{\textbullet}};
%	\node at (0:1.4cm) {R};
% INNER DISK
	\filldraw[fill=white] (180:.65cm) circle (.25cm); 
	\node at (180:.65) {$a$};
% COORDINATE LABELS
%	\node at (a) {(a)};
%	\node at (b) {(b)};
%	\node at (c) {(c)};
\end{tikzpicture}
\smallskip
\\
\mbox{A field } Y(a,w)
\bigskip
&&
\begin{minipage}{4cm}
\centering 
A path of operators from the Segal CFT
\end{minipage}
&&
\begin{minipage}{5cm}
\centering
A bounded local operator
$ U(\gamma_1)^*Y(s^{L_0}a,w)e^{-tL(\rho)} U(\gamma_2)$
\end{minipage}
\end{array}
$$
\captionsetup{justification=centering,width=0.8\linewidth}
\caption{Moving continuously from the fields of a 2d chiral CFT to the local algebras via Segal CFT.}
\label{figFieldsToAlgebras}
\end{figure}
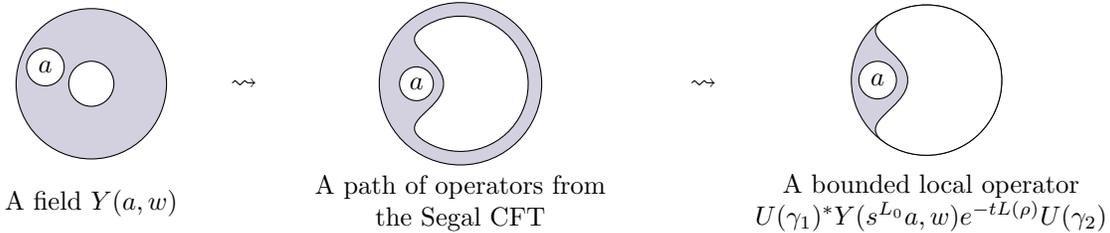

In contrast to the construction in \cite{CKLW18}, where the local algebras are generated by unbounded operators, our construction gives a linear map identifying states in the vertex operator algebra with \emph{bounded} generators of local algebras.
This approach has considerable upside, as bounded operators arising from fields are considerably easier to work with from a technical standpoint, as demonstrated in Wassermann's computation of the fusion rules for $SU(N)_k$ in \cite{Wa98}.
Moreover, geometric ideas have already proven to be valuable in the study of conformal nets, such as in the recent article  \cite{MorinelliTanimotoWeiner18} of Morinelli, Tanimoto, and Weiner which proved the long-held conjecture that conformal nets automatically satisfy the split property.

Of course, one is led to ask which unitary vertex operator algebras Theorem \ref{thmIntroSubalgebraThm} applies to; that is, which appear as subalgebras of $N$ free fermions, for some $N$.
We provide many examples of such VOAs in Section \ref{secOtherSuperalgebras}, and these include the free boson, lattice VOAs corresponding to sublattices of $\Z^N$, the discrete series of (super) Virasoro VOAs (and many other Virasoro VOAs), and affine VOAs corresponding to simple Lie algebras $\g$, at levels $k \Delta_\g$ ($k = 1, 2, \ldots$), where $\Delta_\g \in \Z_{> 0}$ is a constant depending on $\g$.

We see no reason that Theorem \ref{thmIntroSubalgebraThm} should be limited to subalgebras of free fermions; 
embeddings into free fermions are simply a technical tool useful for establishing analytic properties of fields.
Such embeddings were used for the same reason in Wassermann's paper \cite{Wa98}, as well as in a forthcoming paper of Carpi, Weiner and Xu relating representations of conformal nets and representations of vertex operator algebras \cite{CarpiWeinerXu}.
In the future, we hope to extend our results to a larger class of unitary VOAs (e.g. \emph{all} affine models, \emph{all} Virasoro models, \emph{all} lattice models).

\subsection{Acknowledgments}

I am grateful to Andr\'e Henriques for numerous helpful discussions on the topics of this paper, and in particular for explaining his geometric perspective on conformal nets.
I would also like to thank Marcel Bischoff, Sebastiano Carpi, Vaughan Jones, Yasuyuki Kawahigashi, Roberto Longo, David Penneys, Peter Teichner, Antony Wassermann, Mih\'aly Weiner, and Feng Xu for enlightening conversations.
I am very grateful for the hospitality and support of the Max Planck Institute for Mathematics, Bonn, without which this project would not have been possible.
This work was also supported in part by NSF grant DMS 0856316.

\settocdepth{subsection}
\section{Preliminaries}\label{secPreliminaries}

\subsection{Fermionic second quantization and the free fermion Segal CFT}\label{subsecFermions}

We will briefly outline fermionic second quantization and the free fermion Segal CFT; for a more detailed overview, the reader may consult \cite{Ten16}.

\subsubsection{Fermionic fock space}\label{subsecFermionicFockSpace}

Let $H$ and $K$ be complex Hilbert spaces.
We will write $\cB(H, K)$ for the space of bounded linear maps from $H$ to $K$, abbreviated $\cB(H)$ when $H=K$, and $\norm{x}$ for the operator norm of $x \in \cB(H,K)$.
For $p \ge 1$ and $x \in \cB(H,K)$, let $\norm{x}_p = \tr((x^*x)^{p/2})^{1/p}$, and let
$$
\cB_p(H,K) = \{ x \in \cB(H,K) : \norm{x}_p < \infty\}.
$$
We will be primarily interested in the cases $p=1$ and $p=2$.
In these cases, elements of $\cB_p(H,K)$ are called \emph{trace class} and \emph{Hilbert-Schmidt} maps, respectively.
We will write $\cB_p(H)$ for $\cB_p(H,H)$, which is a two-sided ideal of $\cB(H)$.
Note that $\cB_p(H) \subseteq \cB_q(H)$ when $p \le q$.

Given a complex Hilbert space $H$, $\CAR(H)$ is the universal unital $C^*$-algebra with generators $a(f)$ for $f \in H$ which are linear in $f$ and satisfy the relations
\begin{align*}
a(f)a(g) + a(g)a(f) &= 0,\\
a(f)a(g)^* + a(g)^* a(f) &= \ip{f,g} \Id.
\end{align*}

There is an irreducible, faithful representation of $\CAR(H)$ on the Hilbert space
$$
\Lambda H = \bigoplus_{k =0}^\infty \Lambda^k H,
$$
densely defined by $a(f)\omega = f \wedge \omega$. 
These operators are bounded with $\norm{a(f)} = \norm{f}$. 
The subspace $\Lambda^0 H$ is spanned by a distinguished unit vector $\Omega$ which satisfies $a(f)^*\Omega = 0$ for all $f \in H$. 

The exterior Hilbert space $\Lambda H$ is naturally a super Hilbert space, with $\Z/2$-grading given by
\begin{equation*}\label{eqFockGrading}
(\Lambda H)^j = \bigoplus_{k =0}^\infty \Lambda^{2k+j} H
\end{equation*}
for $j \in \{0,1\}$.
We will write $\Gamma_{\F_{H,p}}$ for the grading operator which acts by $(-1)^j$ on $\big(\Lambda H\big)^j$.

There is a family of irreducible, faithful representations of $\CAR(H)$ indexed by projections $p \in \cB(H)$ given as follows. 
Let $H_p = (pH)^* \oplus (1-p)H$, and let $\F_{H,p} = \Lambda H_p$. 
The representation $\pi_p:\CAR(H) \to \cB(\Lambda H_p)$ is given by
$$
\pi_p(a(f)) = a((pf)^*)^* + a((1-p)f).
$$
We will generally write $a(f)$ instead of $\pi_p(a(f))$ when the space that $\CAR(H)$ acts on is clear.

The distinguished unit vector $\Omega_p \in \Lambda^0 H_p$ is characterized, up to scalar multiples, by the equations
\begin{align*}
a(f)\Omega_p &= 0 & \mbox{ for } & f \in pH,\\
a(g)^*\Omega_p &= 0 & \mbox{ for } & g \in (1-p)H .
\end{align*}

The following result is often called the Shale-Stinespring condition, or the  Segal equivalence criterion.

\begin{Theorem}\label{thmShaleStinespring}
Let $H$ be a Hilbert space and let $p$ and $q$ be projections on $H$.
Then there exists a unit vector $\hOmega_q \in \F_{H,p}$ which satisfies
\begin{align*}
a(f)\hOmega_q &= 0 & \mbox{ for } & f \in qH,\\
a(g)^*\hOmega_q &= 0 & \mbox{ for } & g \in (1-q)H,
\end{align*}
if and only if $p-q$ is Hilbert-Schmidt.
If these conditions are satisfied, the vector $\hOmega_q$ will be unique up to scalar multiple.
\end{Theorem}
If the conditions of Theorem \ref{thmShaleStinespring} hold, then there is a unitary isomorphism $\F_{H,q} \to \F_{H,p}$ of representations of $\CAR(H)$ determined by $\Omega_q \mapsto \hOmega_q$.

If $I = \{i_1, \ldots, i_n\}$ is a finite ordered set with $i_1 < i_2 < \ldots < i_n$, and $\{h_i\}_{i \in I} \subset H$ is a family of vectors from $H$ indexed by $I$, then we write
\begin{equation}\label{eqnFermionProductNotationEarly}
a(h_I) := a(h_{i_1}) \cdots a(h_{i_n}) \in \CAR(H).
\end{equation}

We will generally consider the case when $H$ is separable and $pH$ and $(1-p)H$ are both infinite dimensional.
We will then choose an orthonormal basis $\{e_i\}_{i \in \Z}$ with $e_i \in pH$ when $i \ge 0$ and $e_i \in (1-p)H$ when $i < 0$.
Then $\F_{H,p}$ has an orthonormal basis indexed by finite subsets $I \subset \Z_{\ge 0}$ and $J \subset \Z_{<0}$ given by
$$
a(e_J)a(e_I)^*\Omega_p.
$$

A key property of the Fock space construction is that it takes direct sums to tensor products.
\begin{Proposition}\label{propFockSumToTensor}
There is a natural even unitary isomorphism
$
\mathcal{F}_{H \oplus K, p \oplus q} \cong \mathcal{F}_{H,p} \otimes \mathcal{F}_{K,q}
$
characterized by 
$$
a(h_J)a(h_I)^*a(k_{J^\prime})a(k_{I^\prime})^*\Omega_{p\oplus q} \mapsto a(h_J)a(h_I)^*\Omega_p \otimes a(k_{J^\prime})a(k_{I^\prime})^*\Omega_q,
$$
where $h_J$ and $h_I$ are ordered families of vectors from $H$ and $k_{J^\prime}$ and $k_{I^\prime}$ are ordered families of vectors from $K$.
The induced action of $\CAR(H \oplus K)$ on $\mathcal{F}_{H,p} \otimes \mathcal{F}_{K,q}$ is
\begin{equation}\label{eqTensorCARReps}
a(h+k) \mapsto a(h) \otimes 1 + \Gamma_{\F_{H,p}} \otimes a(k)
\end{equation}
where $\Gamma_{\F_{H,p}}$ is the grading operator on $\F_{H,p}$.
\end{Proposition}

Let $\cU(H)$ be the group of unitary operators on $H$, and let 
$$
\cU_{res}(H,p) = \{ u \in \cU(H) : [p, u] \in \cB_2(H) \}
$$
where $[a,b]$ is the commutator $ab-ba$.
There is a strongly continuous projective unitary representation of $\cU_{res}(H,p)$ on $\F_{H,p}$, called the \emph{basic representation}, such that the image $U$ of $u \in \cU_{res}(H,p)$ is characterized by $Ua(f)U^* = a(uf)$.
The image of the vacuum under $U$ is given by
\begin{equation}\label{eqnBasicRepOnVacuum}
U\Omega_p = \hat \Omega_q,
\end{equation}
where $q= upu^*$ and $\hat \Omega_q$ is as in Theorem \ref{thmShaleStinespring}.
Note that $\hat \Omega_q$ only depends on $q$, and not on $u$.
For more details on the basic representation, one may consult \cite[\S 3]{Wa98}, \cite[\S 10]{PrSe86}, or \cite[\S2.1]{Ten16}.

In the following, fix $H = L^2(S^1)$ (with normalized arclength measure on $S^1$), and $p \in \cB(H)$ to be the projection of $H$ onto the classical Hardy space
$$
H^2(\D) = \overline{\Span \{ z^n : n \ge 0\}}.
$$
In this case we will just write $\F$ for $\F_{H,p}$ and $\Omega$ for the vaccum vector $\Omega_p$.
We will refer to $\F$ as \emph{fermionic Fock space}.

Let $\Diff_+(S^1)$ be the group of orientation preserving diffeomorphisms of the unit circle $S^1 \subset \C$, and let $C^\infty(S^1)^{\times}$ be the group of non-vanishing smooth functions $S^1 \to \C$ under pointwise multiplication.
Then $\Diff_+(S^1)$ acts on $C^\infty(S^1)^{\times}$ by automorphisms via $\gamma \cdot f = f \circ \gamma^{-1}$, and we can form the semidirect product $C^\infty(S^1)^\times \rtimes \Diff_+(S^1)$. 
Let $\Diff_+^{NS}(S^1)$ be the double cover of $\Diff_+(S^1)$ given as a subgroup of $C^\infty(S^1)^\times \rtimes \Diff_+(S^1)$ by
$$
\Diff_+^{NS}(S^1) = \{ (\psi, \gamma) \in C^\infty(S^1)^\times \rtimes \Diff_+(S^1) \, : \, \psi^2 = (\gamma^{-1})^\prime \}.
$$
Here, and throughout, if $f \in C^\infty(S^1)$, then we write $f^\prime$ for the complex derivative
$$
f^\prime(z) = \frac{d}{dz} f(z) = \frac{1}{iz} \frac{d}{d\theta} f(e^{i \theta}) \Big|_{z=e^{i \theta}}.
$$

Let $u_{NS}:\Diff_+^{NS}(S^1) \to \cU(H)$ be the unitary representation given by $u_{NS}(\psi, \gamma)f = \psi \cdot (f \circ \gamma^{-1})$.
Then $u_{NS}(\psi, \gamma) \in U_{res}(H, p)$, and by composing with the basic representation we get a projective unitary representation $U_{NS}:\Diff_+^{NS}(S^1) \to \cPU(H)$, which is strongly continuous when $\Diff_+^{NS}(S^1)$ is given, for example, the $C^\infty$ topology.
The representing operators $U_{NS}(\psi,\gamma)$ are even for all $(\psi,\gamma) \in \Diff_+^{NS}(S^1)$.
See e.g. \cite[\S2.1]{Ten16} for an expanded discussion of this representation.

Let $r_\theta \in \Diff_+(S^1)$ be the map $r_\theta z = e^{i \theta}z$.
By Stone's theorem, there is a self-adjoint operator $L_0$, in our case unbounded, such that
\begin{equation}\label{eqnIntroduceL0}
U_{NS}(e^{-i \theta/2}, r_{\theta}) = e^{i \theta L_0}.
\end{equation}
If we write $e_j$ for the function $z^j \in H$, then $\{e_j\}_{j \in \Z}$ is an orthonormal basis for $H$, with $e_j \in pH$ when $j \ge 0$ and $e_j \in (1-p)H$ when $j < 0$.
The corresponding basis $a(e_J)a(e_I)^*\Omega$ for $\F$ diagonalizes $L_0$, and one has
$$
L_0a(e_J)a(e_I)^*\Omega = \Big( \sum_{i \in I} (i + \tfrac12) - \sum_{j \in J} (j + \tfrac12) \Big) a(e_J)a(e_I)^*\Omega
$$
where $J \subset \Z_{< 0}$ and $I \subset \Z_{\ge 0}$ are finite subsets.
Note that the eigenvalues of $L_0$ are $\tfrac12 \Z_{\ge 0}$, and each eigenspace is finite dimensional.
We denote by $\F^0$ the algebraic span of the eigenvectors of $L_0$, and write $\F_{\le N} \subset \F^0$ for the finite dimensional subspace spanned by eigenvectors of $L_0$ with eigenvalue at most $N$.

\subsubsection{The free fermion Segal CFT}\label{subsubsecFreeFermionSegalCFT}

Let $\Sigma$ be a Riemann surface.
A \emph{spin structure} on $\Sigma$ is a holomorphic line bundle $L \to \Sigma$, and a holomorphic isomorphism $L \otimes L \to K_\Sigma$, where $K_\Sigma$ is the holomorphic cotangent bundle.
We will refer to a Riemann surface $\Sigma$ equipped with a spin structure as a \emph{spin Riemann surface.}

If $L_1$ and $L_2$ are spin structures on $\Sigma_1$ and $\Sigma_2$, respectively, then an isomorphism of spin structures $L_1 \to L_2$ is a holomorphic isomorphism of bundles $B:L_1 \to L_2$ such that the diagram
$$
\begin{CD}
L_1 \otimes L_1 @>{B \otimes B}>> L_2 \otimes L_2\\
@V\Phi_1 VV @VV\Phi_2 V\\
K_{\Sigma_1} @<{B|_{\Sigma_1}}^*<< K_{\Sigma_2}
\end{CD}
$$
commutes, where $B|_\Sigma^*$ is the pullback of holomorphic $1$-forms induced by the biholomorphic map $B|_{\Sigma_1}:\Sigma_1 \to \Sigma_2$.

Up to isomorphism there is a unique spin structure on $\C$.
It is given by the trivial bundle $L = \C \times \C$, and the isomorphism $\Phi:L \otimes L \to \K_\C$ is given on sections by 
\begin{equation}\label{eqnStandardCSpin}
\Phi_*(f \otimes g) = fg \, dz.
\end{equation}

If $0 < r < 1$, then the annulus $\bbA_r = \{ r \le \abs{z} \le 1\}$ has two non-isomorphic spin structures, called the Neveu-Schwarz (NS) and Ramond (R) spin structures, again given by the trivial bundle $L = \bbA_r \times \C$. For $\sigma \in \{NS,R\}$, the isomorphisms $\Phi_\sigma: L \otimes L \to K_{\bbA_r}$ are given by
$$
{\Phi_\sigma}_*(f \otimes g) = \left\{\begin{array}{cc} -ifg \, dz & \sigma=NS \\ -i fg z^{-1}dz & \sigma =R \end{array} \right.
$$
We denote these spin annuli by $(\bbA_r, \sigma)$.

If $Y$ is a closed, smooth $1$-manifold, then a {\it spin structure} on $Y$ is a smooth, complex line bundle $L$ and an isomorphism of complex line bundles $\phi:L \otimes L \to T^*Y_\C$, where $T^*Y_\C = T^*Y \otimes_\R \C$.

An isomorphism of spin structures $(Y_1, L_1) \to (Y_2, L_2)$ is a smooth bundle map $\beta:L_1 \to L_2$ such that
\begin{equation}\label{eqnCircleSpinIso}
\begin{CD}
L_1 \otimes L_1 @>{\beta \otimes \beta}>> L_2 \otimes L_2\\
@V\phi_1 VV @VV\phi_2 V\\
{T^*Y_1}_\C @<{\beta|_{Y_1}}^*<< {T^*Y_2}_\C
\end{CD}
\end{equation}
where ${\beta|_{Y_1}}^*$ is the isomorphism of cotangent bundles induced by the diffeomorphism $\beta|_{Y_1}:Y_1 \to Y_2$.

If $(\Sigma, L)$ is a compact spin Riemann surface with non-empty boundary, then $\partial \Sigma$ becomes a spin 1-manifold by identifying $T^*\partial\Sigma_\C \cong K_\Sigma|_{\partial \Sigma}$ in such a way that the real subspace $T^*\partial\Sigma$ corresponds to the dual of tangent vectors parallel to the boundary.

There are two non-isomorphic spin structures on the unit circle $S^1 \subset \C$, called the Neveu-Schwarz (NS) and Ramond (R) spin structures, obtained by restricting $(\bbA_r, \sigma)$ to $S^1$, where $\sigma \in \{NS, R\}$.
We denote these spin circles by $(S^1, \sigma)$.

The group $\Aut_+(S^1, NS)$ of orientation preserving automorphisms of $(S^1, NS)$ can naturally be identified with $\Diff^{NS}_+(S^1)$ as follows.
Given $(\psi, \gamma) \in \Diff^{NS}_+(S^1)$, there is a unique $\alpha_{NS}(\psi,\gamma) \in \Aut_+(S^1, NS)$ such that $\alpha_{NS}(\psi, \gamma)$ acts on sections $f \in C^\infty(S^1)$ by $\alpha_{NS}(\psi,\gamma)_*f = u_{NS}(\psi, \gamma)f$.

Similarly, there is an isomorphism $\alpha_{R}$ of $\Aut_+(S^1, R)$ with the double cover $\Diff_+^R(S^1)$ of $\Diff_+(S^1)$ given by
$$
\Diff_+^R(S^1) = \{(\psi, \gamma) \in C^\infty(S^1)^\times \rtimes \Diff_+(S^1) : \psi^2 = \abs{(\gamma^{-1})^\prime}\}.
$$
This isomorphism is characterized by 
$$
\alpha_{R}(\psi, \gamma)_*f = u_{R}(\psi, \gamma)f := \psi \cdot (f \circ \gamma^{-1})
$$ 
for all $f \in C^\infty(S^1)$.
Like with the Neveu-Schwarz representation, we have $u_R(\psi,\gamma) \in \cU_{res}$, and so we have a projective unitary representation $U_R:\Diff_+^R(S^1) \to \cPU(\F)$.

\begin{Definition}
A \emph{spin Riemann surface with parametrized boundary} is a collection of:
\begin{itemize}
\item A compact Riemann surface with boundary $\Sigma$ with spin structure $(L, \Phi)$. 
We write $\pi_0(\partial \Sigma)$ for the set of connected components of $\partial \Sigma$, and for $j \in \pi_0(\partial \Sigma)$ we let $\sigma(j) \in \{NS, R\}$ denote the isomorphism class of the restriction $L|_j$.
\item A trivialization of $L|_{\partial \Sigma}$ given by an isomorphism of spin structures
$$
\beta = (\beta_j)_{j \in \pi_0(\partial \Sigma)}: \bigsqcup_{j \in \pi_0(\partial \Sigma)} (S^1, \sigma(j)) \to L|_{\partial \Sigma}.
$$
\end{itemize}
Let $\cR$ be the collection of all such $(\Sigma, L, \Phi, \beta)$ with the property that $\Sigma$ has no closed components.
If $X_i = (\Sigma_i, L_i, \Phi_i, \beta_i) \in \cR$ for $i \in \{1,2\}$, then we say that $X_1$ and $X_2$ are isomorphic if there exists an isomorphism of spin structures $B:(\Sigma_1, L_1, \Phi_1) \to (\Sigma_2, L_2, \Phi_2)$ such that $\beta_2 = B \circ \beta_1$.
\end{Definition}

The complex structure of a Riemann surface induces an orientation, and if $X = (\Sigma, L, \Phi, \beta) \in \cR$ this allows us to partition the connected components of $\partial \Sigma$ into incoming and outgoing components, as follows.
We say that a connected component $j$ of $\partial \Sigma$ is incoming if the diffeomorphism $\beta_j|_{S^1}$ is orientation reversing, and we say that $j$ is outgoing if $\beta_j|_{S^1}$ is orientation preserving.
We write $\partial \Sigma^0$ for the subset of $\partial \Sigma$ consisting of incoming components, and $\partial \Sigma^1$ for the subset consisting of outgoing components.

The free fermion Segal CFT assigns to every $X \in \cR$ a one-dimensional space of trace class maps
\begin{equation}\label{eqnEXDef}
E(X) \subset \cB_1 \Big( \bigotimes_{j \in \pi_0(\partial \Sigma^0)} \F, \bigotimes_{j \in \pi_0(\partial \Sigma^1)} \F \,\,\,\big).
\end{equation}
The unordered tensor products in \eqref{eqnEXDef} are unordered tensor products of super Hilbert spaces, meaing that we have a family of maps, one for every ordering of the tensor products, compatible with the braiding of super Hilbert spaces.

The following theorem summarizes some of the main properties of the assignment $X \mapsto E(X)$. 
For a more detailed treatment, see \cite{Ten16}.

\begin{Theorem}\label{thmFermionSegalCFT}
Let $X = (\Sigma, L, \Phi, \beta) \in \cR$. 
The maps $E(X)$ assigned by the free fermion Segal CFT satisfy the following properties:
\begin{itemize}
\item  (Existence and invariance) $E(X)$ is one-dimensional, and its elements are homogeneous and trace class. 
If $X$ and $\tilde X$ are isomorphic, then $E(X) = E(\tilde X)$.
\item (Non-degeneracy) If every connected component of $\Sigma$ has an outgoing boundary component, then non-zero elements of $E(X)$ are injective. If every connected component of $\Sigma$ has an incoming boundary component, then non-zero elements of $E(X)$ have dense image.
\item (Monoidal) If $Y \in \cR$, then $E(X \sqcup Y) = E(X) \grotimes E(Y)$, where $X \sqcup Y$ is the disjoint union and $\grotimes$ is the graded tensor product of maps of super Hilbert spaces. 
\item (Gluing) If $\hat X \in \cR$ is obtained by sewing two boundary components of $X$ along the parametrizations, then the partial supertrace induces an isomorphism $\tr^s:E(X) \to E(\hat X)$. 
In particular, if $X_0,X_1 \in \cR$ and $X$ is obtained by sewing the outgoing boundary of $X_1$ to the incoming boundary of $X_0$, then $E(X) = E(X_0)E(X_1)$.
\item (Reparametrization) If $(\psi_j,\gamma_j) \in \prod_{j \in \pi_0(\partial \Sigma)} \Diff^{\sigma(j)}_+(S^1)$, and $\tilde X = (\Sigma, L, \Phi, \tilde \beta) \in \cR$ is obtained by setting $\tilde \beta_j = \beta_j \circ \alpha_{\sigma(j)}(\psi_j,\gamma_j)^{-1}$, then 
$$
E(\tilde X) = \left( \bigotimes_{j \in \pi_0(\partial \Sigma^1)} U_{\sigma(j)}(\psi_j,\gamma_j) \right)E(X)\left( \bigotimes_{j \in \pi_0(\partial \Sigma^0)} U_{\sigma(j)}(\psi_j,\gamma_j)^* \right).
$$
\item (Unitarity) $E(\overline{X}) = E(X)^*$, where $\overline{X}$ is the complex conjugate spin Riemann surface, and $E(X)^*$ denotes taking the adjoint elementwise.
\end{itemize}
\end{Theorem}
We have omitted precise explanations of the notions of graded tensor products of maps, of sewing spin Riemann surfaces along boundary parametrizations, and of the conjugate $\overline{X}$; they are discussed in \cite{Ten16} in Section 2.1, Section 2.2, and Section 3.2, respectively.

We will now briefly describe the construction of the spaces $E(X)$, as it is similar to the construction of localized vertex operators in Section \ref{secSegalCFTForDegenerateSurfaces}.

As before, let $H = L^2(S^1)$ and let $p \in \cB(H)$ be the projection onto the classical Hardy space 
$$
H^2(\D) = \overline{ \Span \{ z^n : n \ge 0\}}.
$$
Given $X = (\Sigma, L, \Phi, \beta) \in \cR$, we define the boundary Hilbert space
$$
H_{\partial \Sigma} = \bigoplus_{j \in \pi_0(\partial \Sigma)} H.
$$
We write $\cO(\Sigma; L)$ for the space of sections of $L$ which are holomorphic on the interior $\interior{\Sigma}$ and which extend to smooth functions on $\partial \Sigma$.
Given a section $F \in \cO(\Sigma; L)$, the pullback $\beta^*F$ yields a smooth function on $\bigsqcup_{j \in \pi_0(\partial \Sigma)} S^1$, and thus an element of $H_{\partial \Sigma}$.
Define the Hardy space $H^2(X) \subset H_{\partial \Sigma}$ by
$$
H^2(X) = \overline{ \{\beta^*F : F \in \cO(\Sigma; L)\}}.
$$ 

Now decompose $H_{\partial \Sigma} = H^1_{\partial \Sigma}  \oplus H^0_{\partial \Sigma}$, where
$$
H^i_{\partial \Sigma} = \bigoplus_{j \in \pi_0(\partial \Sigma^i)} H
$$
and define the boundary projections $p_i \in \cB(H_{\partial \Sigma^i})$ by
$$
p_i = \bigoplus_{j \in \pi_0(\partial \Sigma^i)} p.
$$
There is a natural unitary isomorphism
\begin{equation}\label{eqnBoundaryFockTensorIso}
\F_{H^i_{\partial \Sigma}, p_i} \to \bigotimes_{j \in \pi_0(\partial \Sigma^i)} \F
\end{equation}
given by Proposition \ref{propFockSumToTensor}.
That is, for every ordering of the tensor factors of the right-hand side of \eqref{eqnBoundaryFockTensorIso}, we have an isomorphism with $\F_{H^i_{\partial \Sigma, p_i}}$, and these isomorphisms are compatible with the braiding of super Hilbert spaces.

Making this identification, we define $E(X) \subset \cB_1(\F_{H^0_{\partial \Sigma, p_0}}, \F_{H^1_{\partial \Sigma, p_1}})$ to be the space of trace class maps $T$ which satisfy the $H^2(X)$ commutation relations:
$$
a(f^1)T = (-1)^{p(T)} T a(f^0), \qquad a(g^1)^*T = - (-1)^{p(T)}Ta(g^0)^*
$$
for every $(f^1,f^0) \in H^2(X) \subset H^1_{\partial \Sigma}  \oplus H^0_{\partial \Sigma}$, and every $(g^1, g^0) \in H^2(X)^\perp$, where $p(T) \in \{0,1\}$ is the parity of $T$, and the equation is understood by extending linearly if $T$ is not homogeneous.

This description of $E(X)$ is useful for two reasons. 
First, it is good for computing with.
For every holomorphic function on $\Sigma$ one obtains a relation satisfied by elements of $E(X)$, and in practice these relations are easy to work with.
In \cite[\S 5.2]{Ten16}, we used this description to give a short proof that when $X$ is a disk with two disks removed, $E(X)$ is spanned by maps related to vertex operators.

The second advantage of the description of $E(X)$ in terms of commutation relations from $H^2(X)$ is that it can be generalized to other geometric objects $X$ which have a Hardy space.
In Section \ref{secSegalCFTForDegenerateSurfaces}, we will consider what happens when $X$ is is a `degenerate Riemann surface' where the incoming and outgoing boundary of $\Sigma$ are allowed to overlap.

\begin{Example}(\cite[Prop. 5.1]{Ten16})
When $X$ is given by the closed unit disk $\D$, the spin structure it inherits from $\C$, and the identity parametrization on the boundary, then $E(X)$ is spanned by the vacuum vector $\Omega \in \F$
\end{Example}

\begin{Example} (\cite[Prop. 5.2]{Ten16})
When $X$ is given by the spin annulus $(\bbA_r, NS)$ with boundary parametrizations given by the identity and the map $z \mapsto rz$, then $E(X)$ is spanned by $r^{L_0}$.
Note that both the boundary parametrization and $r^{L_0}$ depend on a choice of square root of $r$.
Similarly, when $NS$ is replaced by $R$, $E(X)$ is spanned by $r^{L_0^R}$, where $L_0^R$ is the generator of the one-parameter group $U_R(1, r_\theta)$.
\end{Example}

\begin{Example}(\cite[Thm. 5.4]{Ten16}
Let $w \in \D$ and $r_1,r_2 \in (0,1)$, and assume they satisfy $s + r < \abs{w} < 1 - s$. 
Define the pair of pants
$$
\bbP_{w,s,r} = \D \setminus \left( (s \interior{\D} + w) \cup r \interior{\D}\right),
$$
where $\interior{\D}$ is the open unit disk.

Give $\bbP_{w,s,r}$ the spin structure inherited from $\C$, and parametrize the boundary components via the identity map on $S^1$, and the maps $z \mapsto r z$ and $z \mapsto sz + w$.
Let
$$
Y:\F^0 \to \End(\F^0)[[x^{\pm 1}]]
$$
be the free fermion vertex operator algebra state-field correspondence (see Example \ref{exFFUVOSA}).
Then $E(\bbP_{w,s,r})$ is spanned by the map $T: \F \otimes \F \to \F$ given on $\xi \otimes \eta \in \F^0 \otimes \F^0$ by
$$
T(\xi \otimes \eta) = Y(s^{L_0}\xi, w)r^{L_0}\eta
$$
when the inputs are ordered so that the one corresponding to the boundary component $sS^1 + w$ comes first.
Note that both the boundary parametrizations and $T$ depend on choices of square roots of $s$ and $r$.
Leaving the boundary parametrizations implicit, we can depict this result as follows:
$$
E\left( \,\,
\begin{tikzpicture}[scale=1.4,baseline={([yshift=-.5ex]current bounding box.center)}]
% BIG DISK
	\filldraw[fill=red!10!blue!20!gray!30!white] (0,0) circle (1cm);
	\node at (0:1cm) {\scriptsize{\textbullet}};
	\node at (0:1.15cm) {1};
	%\draw (0,0) circle (1cm);
% CENTERED INNER DISK
	\filldraw[fill=white] (0,0) circle (0.3cm);
	\node at (0,0) {\scriptsize{\textbullet}};
	\node at (0:.15cm) {0};
	\node at (0:.3cm) {\scriptsize{\textbullet}};
	\node at (0:.45cm) {r};
	%\draw (0,0) circle (0.3cm);
% OFF CENTER INNER  DISK
	\filldraw[fill=white] (140:.65cm) circle (.25cm); 
	\node at (140:.65cm) {\scriptsize{\textbullet}};
	\node at (150:.6cm) {w};
	\node at ([shift=(0:0.25cm)]140:.65cm) {\scriptsize{\textbullet}};
	\node at ([shift=(20:0.65cm)]153:.7cm) {w+s};
	%\node at (140:.65cm) {\scriptsize{$\xi$}};
\end{tikzpicture}
\right)
= \C T.
$$
We will frequently leave the parameters $w,r,s$ implicit as well, and depict the state insertions $T(\xi \otimes \eta)$ as follows.
$$
\begin{tikzpicture}[scale=1.4,baseline={([yshift=-.5ex]current bounding box.center)}]
% BIG DISK
	\filldraw[fill=red!10!blue!20!gray!30!white] (0,0) circle (1cm);
	%\draw (0,0) circle (1cm);
% CENTERED INNER DISK
	\filldraw[fill=white] (0,0) circle (0.3cm);
	%\draw (0,0) circle (0.3cm);
% OFF CENTER INNER  DISK
	\filldraw[fill=white] (140:.65cm) circle (.25cm); 
	\node at (140:.65cm) {\scriptsize{$\xi$}};
	\node at (0,0) {\scriptsize{$\eta$}};
\end{tikzpicture}
\,\, = \,\, T(\xi \otimes \eta) \,\,=\,\, Y(s^{L_0}\xi,w)r^{L_0}\eta \,\,.
$$
\end{Example}

\subsection{Unitary vertex operator superalgebras}

We will give only a brief overview of unitary vertex operator superalgebras.
A detailed treatment of unitary vertex operator algebras in the spirit of this paper may be found in \cite[\S4-5]{CKLW18}.
Our treatment is adapted from this reference, as well as from \cite{AiLin17}.

\begin{Definition}
A  vertex operator superalgebra is given by:
\begin{enumerate}
\item a $\Z/2\Z$-graded vector space $V = V^0 \oplus V^1$. Elements of $V^0 \cup V^1$ are called homogeneous vectors, and elements of $V^0$ (resp. $V^1$) are called even (resp. odd) vectors. If $a \in V^i$, we denote the parity $p(a) = i$.
\item even vectors $\Omega,\nu \in V^0$ called the \emph{vacuum vector} and the \emph{conformal vector}, respectively.
\item a state-field correspondence $Y:V  \to \End(V)[[x^{\pm 1}]]$, usually denoted 
\begin{equation}\label{eqnStateFieldForm}
Y(a,x) = \sum_{n \in \Z} a_{(n)} x^{-n-1}.
\end{equation}
Here $\End(V)[[x^{\pm 1}]]$ is the vector space of formal series of the form \eqref{eqnStateFieldForm}.
The coefficients $a_{(n)}$ of this formal series are called the \emph{modes} of $a$.
\end{enumerate}
This data must satisfy:
\begin{enumerate}
\item For every $a \in V$, if $a$ is even (resp. odd) then $a_{(n)}$ is even (resp. odd) for all $n \in \Z$.
\item For every $a,b \in V$, we have $a_{(n)}b = 0$ for $n$ sufficiently large. 
\item For every $a \in V$, we have $a_{(n)}\Omega = 0$ for $n \ge 0$ and $a_{(-1)}\Omega = a$.
\item $Y(\Omega,x) = 1_V$. That is, $\Omega_{(n)} = \delta_{n,-1} 1_V$.
\item For every $a,b \in V$, there exists a positive integer $N$ such that we have an identity of formal series $(x-y)^N [Y(a,x), Y(b,y)]_{\pm} = 0$. Here, the super commutator $[ \, \cdot \, , \, \cdot \,]_{\pm}$ is given by
$$
Y(a,x)Y(b,y) - (-1)^{p(a)p(b)}Y(b,y)Y(a,x)
$$
when $a$ and $b$ are homogeneous, and extended linearly otherwise.
\item If we write $Y(\nu,x) = \sum_{n \in \Z} L_n x^{-n-2}$, then the $L_n$ give a representation of the Virasoro algebra. 
That is, 
$$
[L_m,L_n] = (m-n)L_{m+n} + \tfrac{c}{12} (m^3-m)\delta_{m,-n}1_V
$$
for a number $c \in \C$, called the \emph{central charge}.
\item If we write $V_\alpha = \ker (L_0 - \alpha 1_V)$, then we have a decomposition of $V$ as an algebraic direct sum
$$
V^0 = \bigoplus_{\alpha \in \Z_{\ge 0}} V_\alpha, \qquad V^1 = \bigoplus_{\alpha \in \tfrac12 + \Z_{\ge 0}} V_\alpha
$$
with $\dim V_\alpha < \infty$.
\item For every $a \in V$ we have $[L_{-1}, Y(a,x)] = \frac{d}{dx} Y(a,x)$.
\end{enumerate}
\end{Definition}
We will often abuse terminology by referring to $V$ as a vertex operator superalgebra, instead of referring to the quadruple $(V, Y, \Omega, \nu)$.
If $V^1 = \{0\}$, then $V$ is called a \emph{vertex operator algebra.}

If $a \in V_\alpha$, then we say that $a$ is homogeneous of conformal weight $\alpha =: \Delta_a$.
It follows from the definition that if $a$ is homogeneous, then $a_{(n)}V_{\beta} \subset V_{\beta - n - 1 + \Delta_a}$.

Under this definition, the Borcherds identity (or Jacobi identity) and the Borcherds commutator formula are consequences:
\begin{Theorem}\label{thmBorcherds}
Let $V$ be a vertex operator superalgebra.
Then the Borcherds identity 
\begin{align*}
\sum_{j = 0}^\infty \binom{m}{j} \big(a_{(n+j)}b\big)_{(m+k-j)}&c = \sum_{j=0}^\infty (-1)^j \binom{n}{j} a_{(m+n-j)}b_{(k+j)}c \\ 
&-(-1)^{p(a)p(b)} \sum_{j=0}^\infty(-1)^{j+n} \binom{n}{j} b_{(n+k-j)}a_{(m+j)}c
\end{align*}
holds for every  $a,b,c \in V$ and every $m,k,n \in \Z$.
In particular, for every $a,b,c \in V$ we have the Borcherds product formula
$$
\big(a_{(n)}b\big)_{(k)} c= \sum_{j=0}^\infty (-1)^j \binom{n}{j} \left(a_{(n-j)}b_{(k+j)} - (-1)^{p(a)p(b)+n}b_{(n+k-j)}a_{(j)}\right)c
$$
for all $n,k \in \Z$ by specializing to $m=0$, and the Borcherds commutator formula
$$
a_{(m)}b_{(k)}c - (-1)^{p(a)p(b)} b_{(k)}a_{(m)}c = \sum_{j =0}^\infty \binom{m}{j} \big(a_{(j)}b\big)_{(m+k-j)}c
$$
for all $m,k \in \Z$ by specializing to $n=0$.
As formal series, we have
$$
a_{(m)}Y(b,x) - (-1)^{p(a)p(b)} Y(b,x)a_{(m)} = \sum_{j=0}^\infty \binom{m}{j} Y(a_{(j)}b,x)x^{m-j}.
$$
\end{Theorem}
See \cite[\S4.8]{Kac98} for an extended discussion of the Borcherds identity.

If $W = (W \cap V^0) \oplus (W \cap V^1)$ is a $\Z/2\Z$-graded subspace of $V$, then it is called a vertex subalgebra\footnote{The term `subsuperalgebra' might be more precise, but it is a bit clumsy} 
if $\Omega \in W$ and $ a_{(n)}b \in W$ for all $a,b \in W$ and $n \in \Z$. 
If $\nu \in W$, then $W$ is called a conformal subalgebra of $V$.
The even vectors $V^0$ always form a conformal subalgebra of $V$.

A vertex subalgebra $W$ is called an \emph{ideal} if we have $a_{(n)}b \in W$ for every $a \in V$ and $b \in W$.
A  vertex operator superalgebra $V$ is called \emph{simple} if its only ideals are $\{0\}$ and $V$.

A homomorphism (resp. antilinear homomorphism) from a vertex operator superalgebra $V$ to a vertex operator superalgebra is a complex linear (resp. antilinear) map $\phi:V \to W$ which satisfies $\phi(\Omega_V) = \Omega_W$, $\phi(\nu_V) = \nu_W$, and $\phi(a_{(n)}b) = \phi(a)_{(n)}\phi(b)$ for all $a,b \in V$.
We also have the obvious notion of (antilinear) isomorphism and automorphism.
The grading operator $\Gamma = (-1)^{2L_0}$ is always an automorphism of a vertex operator superalgebra.

\begin{Definition}\label{defUnitaryVOSA}
A \emph{unitary vertex operator superalgebra} is a vertex operator superalgebra $V$ along with an inner product $\ip{\,\cdot \, , \, \cdot \,}$ on $V$ and an antilinear automorphism $\theta$ of $V$, called the \emph{PCT operator}, satisfying:
\begin{enumerate}
\item (Normalization) $\ip{\Omega,\Omega} = 1$
\item (Invariance) $\ip{a, Y(\theta b, x)c} = \ip{Y(e^{x L_1} (-1)^{L_0 + 2L_0^2} x^{-2L_0}b, x^{-1})a, c}$ for all $a,b,c \in V$.
\end{enumerate}
Note that $x$ is treated as a formal, real variable in the statement of the invariance property. 
An isomorphism $\phi:V \to W$ of unitary vertex operator superalgebras is called unitary if $\ip{\phi a, \phi b} = \ip{a, b}$ for all $a,b \in V$.
If $V^1 = \{0\}$ then we refer to $V$ as a \emph{unitary vertex operator algebra}.
\end{Definition}
We will often abuse terminology by simply referring to $V$ as a unitary vertex operator superalgebra, with the additional data left implicit.

\begin{Remark}
We could alter Definition \ref{defUnitaryVOSA} by replacing $(-1)^{L_0 + 2L_0^2}$ by $(-1)^{L_0 - 2L_0^2}$.
If we call the two definitions $(+)$ and $(-)$ unitary vertex operator superalgebras, then there is a bijection between the $(+)$ and $(-)$ notions given by
$$
(V,Y,\Omega,\nu,\ip{ \,\cdot \,, \,\cdot \,}, \theta) \longleftrightarrow (V,Y,\Omega,\nu,\ip{ \,\cdot \,, \,\cdot \,}, \Gamma\theta).
$$
See \cite[\S2]{Yamauchi2014} for a more detailed discussion.
\end{Remark}

The following basic properties of the PCT operator are straightforward generalizations of \cite[Prop. 5.1]{CKLW18}.
\begin{Proposition}
Let $(V,Y,\Omega,\nu,\ip{ \,\cdot \,, \,\cdot \,}, \theta)$ be a unitary vertex operator superalgebra.
Then $\theta$ is the unique antilinear automorphism satisfying the invariance property of Definition \ref{defUnitaryVOSA}.
Moreover, we have
\begin{enumerate}
\item $\theta(V^i) = V^i$
\item $\theta^2 = 1_V$,
\item $\ip{\theta a, \theta b} = \ip{b,a}$ for all $a,b \in V$
\item $\ip{L_na, b} = \ip{a, L_{-n}b}$ for all $a,b \in V$ and $n \in \Z$,
\end{enumerate}
\end{Proposition}

\begin{Proposition}\label{propIsoUnitaryIffTheta}
Let $V$ and $W$ be unitary vertex operator superalgebras, and let $\phi:V \to W$ be an isomorphism of vertex operator superalgebras.
Then $\phi$ is unitary if and only if $\phi \circ \theta_V = \theta_W \circ \phi$. 
\end{Proposition}
\begin{proof}
This proposition follows from the super version of the argument at the beginning of the proof of \cite[Thm 5.21]{CKLW18}, using the super version of \cite[Cor. 4.11]{CKLW18}.
\end{proof}

The following is essentially \cite[Prop. 5.3]{CKLW18}, with the same proof.
\begin{Proposition}\label{propUnitarySimple}
Let $V$ be a unitary vertex operator superalgebra.
Then $V$ is simple if and only if $V_0 = \C \Omega$.
\end{Proposition}

If $V$ is a vertex operator superalgebra, $a, b \in V$, and $z \in \C$, we set 
$$
Y(a,z)b = \sum_{n \in \Z} a_{(n)}b z^{-n-1} \in \prod_{j \in \tfrac12 \Z_{\ge 0}} V_j.
$$
We may regard the Hilbert space completion $\cH_V$ of $V$ as the subspace of $\prod V_j$ consisting of vectors $\sum v_j$ with $\sum \norm{v_j}^2 < \infty$.
A useful fact about unitary vertex operator superalgebras is that $Y(a,z)b$ in fact lies in the Hilbert subspace $\cH_V$ when $0 < \abs {z} < 1$, and thus $Y(a,z)$ is a densely defined unbounded operator on $\cH_V$.
\begin{Proposition}\label{propVertexOperatorDenselyDefined}
Let $V$ be a unitary vertex operator superalgebra, let $a,b \in V$, and let $z \in \C$ with $0 < \abs{z} < 1$.
Then the sum defining $Y(a,z)b$ converges absolutely in $\cH_V$, the Hilbert space completion of $V$.
\end{Proposition}
\begin{proof}
We assume without loss of generality that $a$ and $b$ are eigenvectors for $L_0$ with eigenvalues $\Delta_a$ and $\Delta_b$, respectively.
For $c \in V$ an eigenvector of $L_0$, we will re-index the modes of $c$ by writing $c_n = c_{(n+\Delta_c - 1)}$, so that $[L_0, c_n] = -n c_n$.

We will first establish the result under the additional assumption that $L_1a = 0$ (i.e. that $a$ is \emph{quasiprimary}), where $L_n = \nu_{(n+1)}$ is the representation of the Virasoro algebra associated to $V$.
In this case, the invariance property for the inner product says that
$$
\ip{c, (\theta a)_{n}d} = (-1)^{\Delta_a + 2\Delta_a^2} \ip{a_{-n}c, d}
$$
for every $c,d \in V$ and $n \in \Z$.
By standard results about vertex operator superalgebras (see \cite[Prop. 8.10.3]{FLM88}), the series
\begin{align*}
\ip{Y(\theta a,w)Y(a, z)b, b} &=\sum_{n,m \in \Z - \Delta_a} \ip{(\theta a)_n a_m b, b} w^{-n-\Delta_a} z^{-m-\Delta_a}\\
 &= (-1)^{\Delta_a + 2\Delta_a^2} \sum_{n,m \in \Z - \Delta_a} \ip{a_m b, a_{-n} b} w^{-n-\Delta_a} z^{-m-\Delta_a}
\end{align*}
converges absolutely (to a rational function in $z$ and $w$) whenever $\abs{w} > \abs{z} > 0$.
In our case, $\abs{z} < 1$, so we have convergence with $w = \overline{z}^{-1}$.
Using that $\ip{a_m b, a_{-n} b} = 0$ when $m+n \ne 0$, we see that 
$$
\sum_{m \in \Z - \Delta_a} \ip{a_m b, a_m b} \abs{z}^{-2m} < \infty.
$$
But this expression is precisely $\abs{z}^{2 \Delta_a} \norm{Y(a,z)b}^2$, and so $Y(a,z)b \in \cH_V$.

We established the above result under the assumption that $L_1a = 0$.
By \cite[Rem. 4.9d]{Kac98}, $V$ is spanned by  
$$
\{L_{-1}^k a : k \in \Z_{\ge 0}, a \in V \mbox{ with } L_1a = 0, L_0 a = \Delta_a a\}.
$$
Hence it suffices to show that if $a$ is an eigenvector for $L_0$ and $Y(a, z)b \in \cH_V$ for all $0 < \abs{z} < 1$, then $Y(L_{-1}a, z)b \in \cH_V$ for all $0 < \abs{z} < 1$.

Assume we have such an $a$.
Then 
$$
\norm{Y(a,z)b}^2 = \sum_{n \in \Z} \norm{a_{(n)} b}^2 \abs{z}^{2(-n-1)} < \infty
$$
for all $0 < \abs{z} < 1$.
Hence the function $f(z)$ given by the Laurent expansion
$$
f(z) = \sum_{n \in \Z} \norm{a_{(n)} b}^2 z^{2(-n-1)}
$$
is meromorphic on the open unit disk $\interior{\D}$, with its only pole at $z=0$.
Hence $z^{-1}(zf^\prime)^\prime(z)$ is given by the Laurent expansion
$$
z^{-1}(zf^\prime)^\prime(z) = 4\sum_{n \in \Z} (n-1)^2\norm{a_{(n)}b}^2 z^{2(-n-2)},
$$
which must converge absolutely for $0 < \abs{z} < 1$.
But
$$
Y(L_{-1}a,z)b = \frac{d}{dz} Y(L_{-1}a, z)b = \sum_{n \in \Z} (n-1)a_{(n)}b z^{-n-2},
$$
and we have therefore shown that
$$
\norm{Y(L_{-1}a,z)b}^2 = \sum_{n \in \Z} (n-1)^2 \norm{a_{(n)}b}^2 \abs{z}^{2(-n-2)} < \infty.
$$
\end{proof}

We now turn our attention to unitary subalgebras of unitary vertex operator superalgebras.

\begin{Definition}
Let $(V,Y,\Omega,\nu,\ip{\,\cdot\, , \, \cdot \,}, \theta)$ be a unitary vertex operator superalgebra.
Then a subalgebra $W \subset V$ is called a \emph{unitary subalgebra} if $\theta(W) \subset W$ and $L_{1}W \subset W$.
\end{Definition}

The following is essentially \cite[Prop. 5.29]{CKLW18}.
\begin{Proposition}
Let $(V,Y,\Omega,\nu,\ip{\,\cdot\, , \, \cdot \,}, \theta)$ be a simple unitary vertex operator superalgebra and let $W \subset V$ be a unitary subalgebra.
Let $\cH_V$ be the Hilbert space completion of $V$, let $\cH_W$ be the closure of $W$ in $\cH_V$, and let $e_W$ be the projection of $\cH_V$ onto $\cH_W$.
Let $\nu^W = e_W\nu$, and let $Y^W$ and $\theta^W$ be the restrictions of $Y$ and $\theta$ to $W$.
Then $\nu^W \in W$ and $\nu^W$ is a conformal vector for $W$ making $(W, Y^W, \nu^W, \ip{\,\cdot\, , \, \cdot \,}, \theta^W)$ into a simple unitary vertex operator superalgebra.
We have $L_i^{W} = L_i|_{W}$ for $i \in \{-1, 0 ,1\}$, and in particular the $\tfrac12 \Z_{\ge 0}$ grading of $W$ coincides with the one inherited from $V$.
\end{Proposition}
Note that unitary subalgebras of simple unitary vertex operator superalgebras are again simple by Proposition \ref{propUnitarySimple}.

\begin{Definition}
Let $(V,Y,\Omega, \nu)$ be a vertex operator superalgebra and let $W$ be a subalgebra.
The \emph{coset subalgebra} $W^c \subset V$ is given by
$$
W^c = \{a \in V : [Y(a,x), Y(b,y)]_{\pm} = 0 \mbox{ for all } b \in W \}.
$$
\end{Definition}

\begin{Proposition}\label{propUVOSACoset}
Let $(V,Y,\Omega,\nu,\ip{\,\cdot\, , \, \cdot \,}, \theta)$ be a simple unitary vertex operator superalgebra, and let $W \subset V$ be a unitary subalgebra.
Then $W^c$ is a unitary subalgebra and $\nu = \nu^W + \nu^{W^c}$.
\end{Proposition}
\begin{proof}
The proof that $W^c$ is a unitary subalgebra in the super case is the same as the proof given \cite[Ex. 5.27]{CKLW18} in the even case. 
The statement about conformal vectors is proven just as in \cite[Prop. 5.31]{CKLW18}.
\end{proof}

A relatively straightforward construction of unitary vertex operator superalgebras is the tensor product.
\begin{Proposition}\label{propUVOSATensorProduct}
For $i \in \{1,2\}$, let $(V_i, Y^i, \Omega^i, \nu^i, \ip{\, \cdot\, , \, \cdot \,}, \theta_i)$ be unitary vertex operator superalgebras.
For $a^i \in V_i$ homogeneous vectors with parity $p(a^i)$, let $Y(a^1 \otimes a^2, x) = Y^1(a^1, x)\Gamma_{V_1}^{p(a^2)} \otimes Y^2(a^2, x)$.
Then $(V_1 \otimes V_2, Y, \Omega^1 \otimes \Omega^2, \nu^1 \otimes \Omega^2 + \Omega^1 \otimes \nu^2, \ip{\, \cdot \, , \, \cdot \,}, \theta_1 \otimes \theta_2)$ is a unitary vertex operator superalgebra.
\end{Proposition}
\begin{proof}
This is asserted in \cite[Prop. 2.4]{AiLin17}, but we will expand on this a little.
To see that $V_1 \otimes V_2$ is a vertex operator superalgebra, the only non-trivial thing to check is locality.
By Dong's lemma \cite[\S 3.2]{Kac98}, it suffices to check that the generators $A^1(x)=Y^1(a^1, x) \otimes 1_{V_2}$ and $A^2(x)=\Gamma^{p(a^2)} \otimes Y^2(a^2, x)$ are pairwise local.
That the $A^i$ are local with respect to themselves is clear.
Additionally, we have
$$
[A^1(x), A^2(x)]_{\pm} = [Y^1(a^1,x) \otimes 1_{V_2}, 1_{V_1} \otimes Y^2(a^2,x)](\Gamma^{p(a^2)} \otimes 1_{V_2}) = 0.
$$
It is clear that $\theta_1 \otimes \theta_2$ is an antilinear automorphism of $V_1 \otimes V_2$, and the proof of invariance is straightforward, as in \cite[Prop. 2.9]{DongLin14}.
\end{proof}
Note that by Proposition \ref{propUnitarySimple}, the tensor product of simple unitary vertex operator superalgebras is again simple.

The following observation is well-known, but we were unable to find a statement in the literature, and so a proof is included for completeness.
\begin{Proposition}\label{propWandCommutantGenerateTensorProduct}
Let $(V,Y,\Omega,\nu,\ip{\, \cdot \,, \, \cdot \,}, \theta)$ be a simple unitary vertex operator superalgebra, and let $W$ be a unitary subalgebra. 
Let $U = \Span \{ a_{(-1)}b : a \in W, b \in W^c\}$.
Then $U$ is a unitary conformal subalgebra of $V$, unitarily isomorphic to $W \otimes W^c$ via the map $u(a \otimes b) = a_{(-1)}b$.
\end{Proposition}
\begin{proof}
We first check that $u:a \otimes b \mapsto a_{(-1)}b$ gives a vertex superalgebra homomorphism $W \otimes W^c \to V$. 

It is clear that $u(\Omega \otimes \Omega) = \Omega$.
By Proposition \ref{propUVOSACoset}, we have
$$
\nu^V = \nu^W + \nu^{W^c} = u(\nu^W \otimes \Omega) + u(\Omega \otimes \nu^{W^c}) = u(\nu^{W \otimes W^c}).
$$

Let $a,a^\prime \in W$ and $b,b^\prime \in W^c$ be homogeneous vectors.
By the Borcherds product formula (Theorem \ref{thmBorcherds}), we have for $k \in \Z$ and $c \in V$
$$
(a_{(-1)}b)_{(k)} c = \sum_{j \ge 0} a_{(-1 - j)}b_{(k+j)}c + (-1)^{p(a)p(b)}b_{(k-1-j)}a_{(j)}c = \sum_{j \in \Z} a_{(j)}b_{(k-1-j)}c,
$$
with the last sum finite since all modes of $a$ and $b$ supercommute.

On the other hand, we have 
$$
(a \otimes b)_{(k)}(a^\prime \otimes b^\prime) = (-1)^{p(b)p(a^\prime)} \sum_{j \in \Z} a_{(j)} a^\prime \otimes b_{(k-j-1)} b^\prime,
$$
and so
\begin{align*}
u\Big((a \otimes b)_{(k)} (a^\prime \otimes b^\prime)\Big) &= (-1)^{p(b)p(a^\prime)} \sum_{j \in \Z} (a_{(j)} a^\prime)_{(-1)} b_{(k-j-1)}b^\prime\\
&= (-1)^{p(b)p(a^\prime)} \sum_{j \in \Z} (a_{(j)} a^\prime)_{(-1)} b_{(k-j-1)}b^\prime_{(-1)}\Omega\\
&= (-1)^{p(b)p(a)} (-1)^{p(b^\prime)(p(a)+p(a^\prime))}\sum_{j \in \Z}  b_{(k-j-1)}b^\prime_{(-1)}(a_{(j)} a^\prime)_{(-1)}\Omega\\
&=  (-1)^{p(b)p(a)} (-1)^{p(b^\prime)(p(a)+p(a^\prime))}\sum_{j \in \Z} b_{(k-j-1)}b^\prime_{(-1)} a_{(j)} a^\prime_{(-1)}\Omega\\
&= \sum_{j \in \Z} a_{(j)}b_{(k-j-1)} a^\prime_{(-1)}b^\prime.
\end{align*}
Hence 
$$
(u(a \otimes b))_{(k)} u(a^\prime \otimes b^\prime) = (a_{(-1)}b)_{(k)} a_{(-1)}^\prime b^\prime = u\Big((a \otimes b)_{(k)} (a^\prime \otimes b^\prime)\Big),
$$
which establishes that $a \otimes b \mapsto a_{(-1)}b$ is a map of vertex operator superalgebras.

Since $U$ is the image of $u$, it is a conformal vertex subalgebra of $V$. 
Conformal subalgebras are automatically invariant under $L_1$, so to check that $U$ is a unitary subalgebra we just need to check invariance under $\theta$.
However, this is clear because $W$ and $W^c$ are unitary subalgebras, the latter by Proposition \ref{propUVOSACoset}.

Finally, we have
$$
u(\theta \xi \otimes \theta \eta) = (\theta\xi)_{(-1)} \theta \eta = \theta(u(\xi \otimes \eta)).
$$
By Proposition \ref{propIsoUnitaryIffTheta} this implies that $u$ is isometric. 
\end{proof}

\begin{Example}\label{exFFUVOSA}
In this paper, the most important example of a unitary vertex operator superalgebra is the free fermion, given on the space $\F^0$ introduced in Section \ref{subsecFermionicFockSpace}.
This example is discussed in \cite[\S5.1]{Kac98} under the name `charged free fermions.'
It is generated by the fields 
$$
Y(a(1)^*\Omega, x) = \sum_{n \in \Z} a(z^{-n-1})^*x^{-n-1}, \qquad Y(a(z^{-1})\Omega, x) = \sum_{n \in \Z} a(z^n) x^{-n-1}
$$
and has a conformal vector $\nu = \tfrac12\big(a(z^{-2})a(1)^* + a(z)^*a(z^{-1})\big)\Omega$ with central charge $c=1$.
One can verify directly, as in \cite[Eq. (5.1.0)]{Kac98}, that the grading operator $L_0 = \nu_{(1)}$ coincides (after taking closure) with the operator $L_0$ defined in \eqref{subsecFermionicFockSpace}.

We have already given an inner product on $\F^0$, and so to specify a unitary structure we need only supply a PCT operator.
Let $j \in \cB(L^2(S^1))$ be given by $(jf)(z) = -z^{-1}f(z^{-1})$, and let $\theta:\F^0 \to \F^0$ be the antilinear map given by
$$
\theta a(g_1)\cdots a(g_m) a(f_1)^* \cdots a(f_n)^*\Omega = a(j g_1)^* \cdots a(jg_m)^* a(jf_1) \cdots a(jf_n)\Omega
$$
for $f_i,g_j \in L^2(S^1)$.
\end{Example}

\begin{Proposition}
The data from Example \ref{exFFUVOSA} makes $\F^0$ into a unitary vertex operator superalgebra with $c=1$.
\end{Proposition}
\begin{proof}
The discussion in \cite[\S5.1]{Kac98} shows that $\F^0$ is a vertex operator superalgebra with $c=1$, so we only need to verify unitarity.
First, we show that $\theta$ is an antilinear automorphism of $\F^0$.
It is clear that $\theta\Omega = \Omega$ and $\theta \nu = \nu$, and also that $\theta^2 = 1$.
If $b \in \{a(1)^*\Omega, a(z^{-1})\Omega\}$, then by inspection we have $\theta b_k \theta = (\theta b)_k$.
It follows from the Borcherds product formula that this identity extends to all $b \in \F^0$, and thus $\theta$ is an antilinear automorphism.
By \cite[Prop. 2.5]{AiLin17}, it suffices to verify the invariance property
$$
\ip{a, Y(\theta b, x)c} = \ip{Y(e^{x L_1} (-1)^{L_0 + 2L_0^2} x^{-2L_0}b, x^{-1})a, c}
$$
when $b \in \{a(1)^*\Omega, a(z^{-1})\Omega\}$.
Note that both such $b$ have conformal weight $\Delta_b = 1/2$, and thus satisfy $L_1 b = 0$.
Hence we have
\begin{align*}
\ip{a, Y(\theta a(1)^*\Omega, x)c} &= -\ip{a,Y(a(z^{-1})\Omega, x)c}\\
&= - \sum_{n \in \Z} \ip{a, a(z^n)c}x^{-n-1}\\
&= - \sum_{n \in \Z} \ip{a(z^n)^*a,c}x^{-n-1}\\
&= \ip{Y(e^{x L_1} (-1)^{L_0 + 2L_0^2} x^{-2L_0} a(1)^*\Omega, x^{-1})a,c}.
\end{align*}
The proof of invariance when $b = a(z^{-1})\Omega$ is similar.
Finally $\ip{\Omega, \Omega} = 1$, which completes the proof.
\end{proof}

We now make a small digression to summarize the properties of positive energy representations of the Virasoro algebra that we will require; see \cite[\S3.2]{CKLW18} for a detailed overview in the spirit of this paper.

\begin{Definition}
The Virasoro algebra $\Vir$ is the complex Lie algebra spanned by elements $L_n$, $n \in \Z$, and a central element $c$ which satisfy
$$
[L_m, L_n] = (m-n)L_{m+n} + \tfrac{c}{12} (m^3-m)\delta_{m,-n}.
$$
A unitary positive energy representation of $\Vir$ is a representation of $\Vir$ on an inner product space $V$, such that
\begin{enumerate}
\item $\ip{L_n a, b} = \ip{a, L_{-n}b}$ for all $a,b \in V$,
\item $L_0$ is algebraically diagonalizable with non-negative real eigenvalues,
\item the central element $c$ acts by a scalar multiple of the identity.
\end{enumerate}
\end{Definition}
By definition, the modes $L_n = \nu_{(n+1)}$ of the conformal vector of a unitary vertex operator superalgebra give a unitary positive energy representation of the Virasoro algebra.

It is well known (relevant papers include \cite{GoWa85} and \cite{TL99}; see \cite[\S3.2]{CKLW18} for a discussion) that such representations exponentiate to strongly continuous projective unitary positive energy representations of the universal cover of $\Diff_+(S^1)$, $\Diff^{(\infty)}_+(S^1)$.
If the representation arises from a unitary vertex operator superalgebra, we have $e^{4 \pi i L_0} = 1$, and thus this representation factors through the double cover $\Diff_+^{NS}(S^1)$, as in \cite[\S6.3]{CaKaLo08}.

Suppose we have a positive energy representation $L_n$ of $\Vir$ arising from a unitary vertex operator superalgebra $V$, and let $\cH_V$ be the Hilbert space completion of $V$.
Let $L(x) = \sum_{n \in \Z} L_n x^{-n-2}$ be the associated generating function, and for $f \in C^\infty(S^1)$, write
$$
L^0(f) = \sum_{n \in \Z} L_n \hat f_n,
$$
an unbounded operator defined on $V$,
where $f(z) = \sum_{n \in \Z}\hat f_n z^n$ is the Fourier series of $f$.
Let $L(f)$ denote the closure of $L^0(f)$, and let $\cH_V^\infty$ be the smooth vectors for $\overline{1+L_0}$, defined by $\cH_V^\infty = \bigcap_{n=0}^\infty \cD((\overline{1 + L_0 \vspace{12pt}})^n)$.
Then $\cH_V^\infty$ is an invariant core for $L(f)$.
We also have $L(\overline{f}) = L(f)^*$, and if $f$ is real-valued then $L(f)$ is self-adjoint.

The generators of the Virasoro algebra correspond to complex polynomial vector fields on the unit circle $S^1$.
We denote by $\Vect(S^1)$ the space of smooth vector fields, and by $\Vect(S^1)_\C$ its complexification,
whose elements can be represented by $f(z) \frac{d}{dz}$ with $f \in C^\infty(S^1)$.
The real subpace $\Vect(S^1)$ consists of those $f(z) \tfrac{d}{dz}$ for which $-i z^{-1} f(z) \in \R$ for all $z \in S^1$.
For such $f$, there is a corresponding flow $(t,z) \mapsto \gamma_t(z) \in C^\infty(\R \times S^1)$ such that  $\gamma_t \in \Diff_+(S^1)$ is a one-parameter group, denoted
$$
\gamma_t(z) =: \exp\Big(t f(z) \tfrac{d}{dz}\Big).
$$
By definition, the flow satisfies $\tfrac{\partial}{\partial t} \gamma_t(z) = f(\gamma_t(z))$.
There is a unique continuous lift of $\gamma_t$ to $(\psi_t, \gamma_t) \in \Diff_+^{NS}(S^1)$ such that $\psi_0 \equiv 1$.

Let $\pi$ be the representation of $\Vect(S^1)$ via unbounded operators on $\cH_V$ extending the action of the Virasoro algebra. 
That is, the representation given by $\pi(f \tfrac{d}{dz}) = L(z^{-1} f)$.
Let $U_\pi : \Diff_+^{NS} \to \cPU(\cH_V)$ be the associated strongly continuous representation,  which for every $t \in \R$ will satisfy
$$
U_\pi(\psi_t,\gamma_t) = e^{t L(z^{-1} f)}
$$ 
after correcting by a scalar.

Using this description, we wish to prove that the representation $U_{\pi}$ arising from the free fermion unitary vertex operator superalgebra is the Neveu-Schwarz representation $U_{NS}$ introduces in Section \ref{subsecFermions}.
First, we require a preparatory observation.
For the free fermion, we write $\F^\infty$ for the smooth vectors for $\overline{1+L_0}$.
\begin{Proposition}\label{propVirFermionCommRels}
If $g \in C^\infty(S^1)$, then $\F^\infty$ is invariant under $a(g)$ and $a(g)^*$.
If $f,g \in C^\infty(S^1)$ and $\xi \in \F^\infty$, then 
\begin{equation}
\label{eqnSmearedVirasoroAndFermion}
L(f)a(g)\xi = a(g)L(f)\xi -  a(zfg^\prime + \tfrac12 (zf)^\prime g)\xi
\end{equation}
and
\begin{equation}
\label{eqnSmearedVirasoroAndFermionStar}
L(f)a(g)^*\xi = a(g)^*L(f)\xi +  a(z\overline{f}g^\prime + \tfrac12 (z\overline{f})^\prime g)^*\xi.
\end{equation}
\end{Proposition}
\begin{proof}
By the Borcherds commutator formula (Theorem \ref{thmBorcherds}), we have 
$$
L_k a(z^n)\xi = a(z^n)L_k\xi -(n + \tfrac{k+1}{2})a(z^{n+k})\xi
$$
for $\xi \in \F^0$.
This is the desired formula \eqref{eqnSmearedVirasoroAndFermion} when $f = z^k$ and $g = z^n$.
By linearity, \eqref{eqnSmearedVirasoroAndFermion} holds when when $f$ and $g$ are trigonometric polynomials.

Now let $f,g \in C^\infty(S^1)$ and $\xi \in \cD(L(f))$.
By definition, $\F^0$ is a core for $L(f)$, and so we may take a sequence $\xi_n \in \F^0$ with $\xi_n \to \xi$ and $L(f)\xi_n \to L(f)\xi$. 
Let
$$
f_M(z) = \sum_{k=-M}^M \hat f_k z^k, \qquad g_N(z) = \sum_{k = -N}^N \hat g_k z^k
$$
be the truncated Fourier series.
We have $f_M \to f$ and $f_M^\prime \to  f^\prime$ in $L^2$ norm, and similarly $g_N \to g$ and $g_N^\prime \to g^\prime$.
By the above argument, 
\begin{equation}\label{eqnAllTrigVirasoroAndFermion}
L(f_M)a(g_N)\xi_n = a(g_N)L(f_M)\xi_n - a(z f_M g_N^\prime + \tfrac12 (zf_M)^\prime g_N)\xi_n.
\end{equation}
By the definition of $L(f)$, we have $L(f_M)a(g_N)\xi_n \to L(f)a(g_N)\xi_n$ as $M \to \infty$.
On the other hand, we can compute the limit of the right-hand side of \eqref{eqnAllTrigVirasoroAndFermion} the same way, and we obtain
\begin{equation}\label{eqnSecondVirasoroAndFermion}
L(f)a(g_N)\xi_n = a(g_N)L(f)\xi_n - a(z f g_N^\prime + \tfrac12 (zf)^\prime g_N)\xi_n.
\end{equation}
As $n \to \infty$, the right-hand side of \eqref{eqnSecondVirasoroAndFermion} converges, and since $L(f)$ is closed we have $a(g_N)\xi \in \cD(L(f))$ and
$$
L(f)a(g_N)\xi = a(g_N)L(f)\xi - a(z f g_N^\prime + \tfrac12 (zf)^\prime g_N)\xi.
$$ 
Letting $N \to \infty$ and repeating the above argument, we get that $a(g)\xi \in \cD(L(f))$ and that \eqref{eqnSmearedVirasoroAndFermion} holds for all smooth $f$ and $g$, and all $\xi \in \cD(L(f))$ (and in particular all $\xi \in \F^\infty$).

Let $D = \overline{1 + L_0} = 1 + L(1)$.
Then by the above $\cD(D)$ is invariant under $a(g)$ for all $g \in C^\infty(S^1)$, and we have
\begin{equation}\label{eqnDCommRel}
Da(g)\xi = a(g)D\xi - a(zg^\prime + \tfrac12g)\xi
\end{equation}
for all $\xi \in \cD(D)$.

Now suppose that $\xi \in \cD(D^2)$, so that $\xi,D\xi \in \cD(D)$.
Then by \eqref{eqnDCommRel} and the fact that $\cD(D)$ is invariant under $a(h)$ when $h$ is smooth, we have $Da(g)\xi \in \cD(D)$. 
Hence $a(g)\xi \in \cD(D^2)$.
Iterating this argument, we see that $\cD(D^n)$ is invariant under $a(g)$, and thus so is $\F^\infty$.

Now let $\xi \in \cD(L(f))$,  $\eta \in \F^\infty$ and $f,g \in C^\infty(S^1)$.
We have $L(f)^* = L(\overline{f})$, and thus $a(g)\eta \in \cD(L(f)^*)$. 
We can now calculate
\begin{align*}
\ip{a(g)^*L(f)\xi, \eta} &= \ip{\xi, L(\overline{f})a(g)\eta}\\
&= \ip{\xi, a(g)L(\overline{f})\eta} - \ip{\xi,a(z\overline{f} g^\prime + \tfrac12(z\overline{f})^\prime g)\eta}\\
&= \ip{a(g)^*\xi, L(\overline{f})\eta}- \ip{a(z\overline{f} g^\prime + \tfrac12(z\overline{f})^\prime g)^*\xi,\eta}\\
\end{align*}
It follows that $\eta \mapsto \ip{a(g)^*\xi, L(\overline{f})\eta}$ is a bounded antilinear functional, and thus $a(g)^*\xi \in \cD(L(f))$ and
$$
\ip{a(g)^*L(f)\xi,\eta} = \ip{L(f)a(g)^*\xi,\eta} - \ip{a(z\overline{f} g^\prime + \tfrac12(z\overline{f})^\prime g)^*\xi,\eta}.
$$
Thus we have \eqref{eqnSmearedVirasoroAndFermionStar}.
One can now use the same argument as above to show that $a(g)^*\cD(L(f)) \subset \cD(L(f))$ implies that $\F^\infty$ is invariant under $a(g)^*$.
\end{proof}

We now return to our original goal.

\begin{Proposition}\label{propExpIsNS}
Let $(\F^0,Y,\Omega,\nu,\ip{\, \cdot \,, \, \cdot \,}, \theta)$ be the free fermion unitary vertex operator superalgebra, and let $U$ be the projective unitary representation of $\Diff_+^{NS}(S^1)$ on $\F$ associated to the positive energy representation $L_n = \nu_{(n+1)}$ of the Virasoro algebra. 
The $U= U_{NS}$, the Neveu-Schwarz representation.
\end{Proposition}
\begin{proof}
Let $f \in C^\infty(S^1)$ and suppose that $-iz^{-1}f(z) \in \R$ for all $z$, so that $f(z) \tfrac{d}{dz} \in \Vect(S^1)$.
Let $(\psi_t, \gamma_t)$ be the associated one-parameter group in $\Diff^{NS}_+(S^1)$, so that $(t,z) \mapsto \gamma_t(z) \in C^\infty(\R \times S^1)$.
To prove that $U_\pi(\psi_t, \gamma_t) = U_{NS}(\psi_t, \gamma_t)$ up to scalar, it will suffice to prove that 
\begin{equation}\label{eqnExpAgreesWithNS}
e^{t L(z^{-1} f)}a(g) e^{-t L(z^{-1} f)} = a(u_{NS}(\psi_t,\gamma_t)g)
\end{equation}
for all $g \in C^\infty(S^1)$.
Indeed, by \cite[Thm. 1]{Thurston74}, $\Diff_+(S^1)$ is simple, and it follows that one-parameter groups $(\psi_t,\gamma_t)$ generate  $\Diff^{NS}_+(S^1)$ algebraically (observe that the spin involution $(-1, \operatorname{id}) \in \Diff^{NS}_+(S^1)$ lies in the one-parameter subgroup $(e^{-it/2},e^{it}z)$).
Thus once we establish \eqref{eqnExpAgreesWithNS}, we are done.

Since $u_{NS}(\psi_t,\gamma_t)$ is a strongly continuous one-parameter group, there is a skew-adjoint operator $X$ on $L^2(S^1)$ such that
$$
u_{NS}(\psi_t,\gamma_t) = e^{tX}.
$$
For $g \in C^\infty(S^1)$, we have
$$
\frac{\partial}{\partial t} u_{NS}(\psi_t,\gamma_t)g = -u_{NS}(\psi_t, \gamma_t)(\tfrac12 f^\prime g + f g^\prime),
$$
with the difference quotients converging uniformly as functions of $z$.
Hence $C^\infty(S^1) \subset \cD(X)$, and for $g \in C^\infty(S^1)$ we have $Xg = -\tfrac12 f^\prime g - f g^\prime$.
Since $C^\infty(S^1)$ is invariant under $X$, we also have $C^\infty(S^1) \subset \cD(X^n)$ for all $n$.

By Proposition \ref{propVirFermionCommRels}, we have $L(z^{-1}f)a(g)\xi - a(g)L(z^{-1}f)\xi = a(Xg)\xi$ for $g \in C^\infty(S^1)$ and $\xi \in \F^\infty$.
We can then apply the argument of \cite[\S8 Exp. Thm.]{Wa98} to obtain \eqref{eqnExpAgreesWithNS} and complete the proof.
\end{proof}

\subsection{Fermi conformal nets}

In this section we will briefly give the definition of a Fermi conformal net, the $\Z/2\Z$-graded analog of a local conformal net.
For a more detailed introduction, the interested reader may consult \cite{CaKaLo08}.

We first recall some basic terminology.
A super Hilbert space $\cH$ is a Hilbert space, along with a $\Z/2\Z$ grading $\cH = \cH^0 \oplus \cH^1$.
The grading induces a unitary involution $\Gamma = 1_{\cH^0} \oplus -1_{\cH^1}$ called the \emph{grading involution}.
Elements of $\cH^0$ (resp $\cH^1$) are called \emph{even} (resp. \emph{odd}) \emph{homogeneous vectors}, and if $\xi \in \cH^i$ we denote the parity of $\xi$ by $p(\xi) = i$.
The $\Z/2\Z$ grading on $\cH$ induces one on $\cB(\cH)$, correspnding to the involution $x \mapsto \Gamma x \Gamma$.
The supercommutator $[\, \cdot \, , \, \cdot \,]_{\pm}$ on $\cB(\cH)$ is given by $[x,y]_{\pm} = xy - (-1)^{p(x)p(y)}yx$ for homogeneous $x$ and $y$, and by extending linearly otherwise.

An \emph{interval} $I \subset S^1$ is an open, connected, non-empty, non-dense subset.
We denote by $\cI$ the set of all intervals.
If $I \in \cI$, we denote by $I^\prime$ the complementary interval $\interior{I^c}$.

The group $\Diff_+^{NS}(S^1)$ is the subgroup of $C^\infty(S^1)^\times \rtimes \Diff_+(S^1)$ given by
$$
\Diff_+^{NS}(S^1) = \{ (\psi, \gamma) \in C^\infty(S^1)^\times \rtimes \Diff_+(S^1) : \psi^2 = (\gamma^{-1})^\prime\}.
$$
It is a double cover of $\Diff_+(S^1)$.
We denote by $\Mob^{NS}$ the subgroup of $\Diff_+^{NS}(S^1)$ consisting of $(\psi, \gamma)$ for which $\gamma$ extends to a biholomorphic automorphism of the closed unit disk $\D$.
Finally, in a slight abuse of notation if $I \in \cI$ we write
$$
\Diff_{+}(I) = \{ (\psi, \gamma) \in \Diff_+^{NS}(S^1) : \gamma|_{I^\prime} = \operatorname{id} \mbox{ and } \psi|_{I^\prime} \equiv 1\}.
$$

\begin{Definition}
A Fermi conformal net is given by the data:
\begin{enumerate}
\item A super Hilbert space $\cH = \cH^1 \oplus \cH^0$, with corresponding unitary grading involution $\Gamma$.
\item A strongly continuous projective unitary representation $U:\Diff_+^{NS}(S^1) \to \cPU(\cH)$ which restricts to an honest unitary representation of $\Mob^{NS}$.
\item For every $I \in \cI$, a von Neumann algebra $\cA(I) \subset \cB(\cH)$.
\end{enumerate}
The data is required to satisfy:
\begin{enumerate}
\item The local algebras $\cA(I)$ are $\Z/2\Z$ graded. That is, $\Gamma \cA(I)\Gamma = \cA(I)$.
\item If $I,J \in \cI$ and $I \subset J$, then $\cA(I) \subset \cA(J)$.
\item If $I,J \in \cI$ and $I \cap J = \emptyset$, then $[\cA(I), \cA(J)]_{\pm} = \{0\}$.
\item $U(\psi,\gamma)\cA(I)U(\psi,\gamma)^* = \cA(\gamma(I))$ for all $(\psi,\gamma) \in \Diff^{NS}_+(S^1)$, and $U(\psi,\gamma)xU(\psi,\gamma)^* = x$ when $x \in \cA(I)$ and $(\psi,\gamma) \in \Diff_+(I^\prime)$.
\item There is a unique up to scalar unit vector $\Omega \in \cH$, called the \emph{vacuum vector}, which satisfies $U(\psi,\gamma)\Omega = \Omega$ for all $(\psi, \gamma) \in \Mob^{NS}$.
This vacuum vector is required to be cyclic for the von Neumann algebra $\cA(S^1):=\bigvee_{I \in \cI} \cA(I)$, and it must satisfy $\Gamma \Omega = \Omega$.
\item The generator $L_0$ of the one-parameter group $U(e^{-it/2},e^{it}z)$ is positive.
\end{enumerate}
An isomorphism of Fermi conformal nets $(\cA_1,U_1)$ and $(\cA_2, U_2)$ on $\cH_1$ and $\cH_2$, respectively, is a unitary $u:\cH_1 \to \cH_2$ such that $u\cA_1(I)u^* = \cA_2(I)$ for all $I \in \cI$, and $uU_1(\psi,\gamma)u^* = U_2(\psi,\gamma)$ for all $(\psi,\gamma) \in \Diff_+^{NS}(S^1)$.
\end{Definition}
A Fermi conformal net with $\cH = \cH^0$ is called a local conformal net (or sometimes just a conformal net).
If we set $\cA_b(I) = \{x \in \cA(I) : p(x) = 0\}$, then $\cA_b$ is a local conformal net on $\cH^0$.

Fermi conformal nets have many properties analogous to familiar properties of conformal nets.
We list some basic properties here:
\begin{Theorem}[\cite{CaKaLo08}]\label{thmFermiNetProps}
Let $\cA$ be a Fermi conformal net. Then we have:
\begin{enumerate}
\item (Haag duality) $\cA(I^\prime) = \{x \in \cB(\cH) : [x,y]_{\pm} = 0 \mbox{ for all } y \in \cA(I) \}$
\item $\Gamma U(\psi,\gamma) = U(\psi,\gamma) \Gamma$ for all $(\psi,\gamma) \in \Diff_+^{NS}(S^1)$.
\item $U(-1,\operatorname{id}) = e^{2 \pi i L_0} = \Gamma$
\item $\cA(I)$ is a type III factor for every interval $I \in \cI$.
\item (Reeh-Schlieder) $\cH = \overline{\cA(I)\Omega}$ for every $I \in \cI$.
\end{enumerate}
\end{Theorem}

A family of von Neumann subalgebras $\cB(I) \subset \cA(I)$ is called a \emph{covariant subnet} if $\cB(I) \subset \cB(J)$ when $I \subset J$ and $U(\psi,\gamma)\cB(I)U(\psi,\gamma)^* = \cB(\gamma(I))$ for all $(\psi,\gamma) \in \Mob^{NS}$.

\begin{Theorem}\label{thmSubnetsAreNets}
Let $(\cB, U)$ be a covariant subnet of a Fermi conformal net.
Then there is a unique strongly continuous projective unitary representation of $\Diff_+^{NS}(S^1)$ making $\cB$ into a Fermi conformal net on $\cH_B :=\overline{\cB(S^1)\Omega}$.
\end{Theorem}
\begin{proof}
Note that $\Gamma \cB(I) = \cB(I) \Gamma$ since $\Gamma = U(-1,\operatorname{id})$.
Hence $\cH_B$ is a graded subspace of $\cH$, and $\cB(I)$ is a graded algebra.
The only non-trivial thing left to verify is covariance.

The existence of a suitable representation $U_B$ of $\Diff_+^{NS}(S^1)$ can be proven just as is done for local conformal nets in \cite[Thm. 6.2.29]{Weiner05}, and uniqueness can be proven just as in \cite[Thm 6.10]{CKLW18}. 
\end{proof}

If $\cB$ is a covariant subnet of a Fermi conformal net, then the usual argument based on Takesaki's theorem (given in e.g. \cite[Lem. 2]{KawahigashiLongo04}, using the Bisognano-Wichmann property \cite[Thm. 2]{CaKaLo08}) shows the following.
\begin{Proposition}\label{propSubnetTrivial}
Let $\cA$ be a Fermi conformal net on $\cH$, and let $\cB \subset \cA$ be a covariant subnet.
For $x \in \cA(I)$, we have $x \in \cB(I)$ if and only if $x\Omega \in \cH_B$.
In particular, $\cB = \cA$ if and only if $\cH_B = \cH$.
\end{Proposition}

There is a notion of graded tensor product $\cA_1 \grotimes \cA_2$ of a pair of Fermi conformal nets.
If $\cH_1$ and $\cH_2$ are super Hilbert spaces and $x_i \in \cB(H_i)$, define $x_1 \grotimes x_2 = x_1 \Gamma^{p(x_2)} \otimes x_2 \in \cB(\cH_1 \otimes \cH_2)$ for homogeneous $x_2$, and by extending linearly otherwise.
We have $(x_1 \grotimes x_2)(y_1 \grotimes y_2) = (-1)^{p(x_2)p(y_1)} x_1y_1 \grotimes x_2 y_2$.
Note that $\cH_1 \otimes \cH_2$ is a super Hilbert space with grading $\Gamma \otimes \Gamma$.

If $(\cA_1, U_1)$ and $(\cA_2, U_2)$ are Fermi conformal nets,  define $(\cA_1 \grotimes \cA_2)(I) = \{x_1 \grotimes x_2 : x_i \in \cA_i(I)\}^{\prime\prime}$, where the double commutant $S^{\prime \prime}$ is the von Neumann algebra generated by a self-adjoint set $S$.
Then $(\cA_1 \grotimes \cA_2, U_1 \otimes U_2)$ is a Fermi conformal net \cite[\S 2.6]{CaKaLo08}.

\begin{Example}
Let $H = L^2(S^1)$, and let $p \in \cB(H)$ be the projection onto the classical Hardy space $H^2(\D)$.
Let $\F := \F_{H, p}$, and let $U_{NS}$ be the Neveu-Schwarz representation of $\Diff_+^{NS}(S^1)$ on $\F$ (see Section \ref{subsecFermionicFockSpace}).
Then the assignment $\cM(I) = \{a(f), a(f)^* : f \in L^2(S^1), \,\, \supp f \subset I \}^{\prime \prime}$ gives a Fermi conformal net, which we call the free fermion conformal net.
Verification of the axioms of a Fermi conformal net is straightforward, although we point out that the cyclicity of the vacuum is contained in \cite[\S 15]{Wa98}, as are direct proofs of many of the properties of the free fermion conformal net.
\end{Example}

\subsection{Composition operators}\label{subsecCompositionOperators}

Let $\D$ denote the closed unit disk, and let $\interior{\D}$ be its interior.
Let $H^2(\D) = \overline{ \Span \{ z^n : n \ge 0 \} } \subset L^2(S^1)$ be the Hardy space, and recall that we can identify $H^2(\D)$ with the space of holomorphic functions on $\interior{\D}$ with almost everywhere non-tangential $L^2$ boundary values.
Let $\phi: \interior{\D} \to \interior{\D}$ be a holomorphic map, and define the composition operator $C_\phi \in \cB(H^2(\D))$ by $C_\phi f = f \circ \phi$.
For a thorough introduction to composition operators, the reader may consult \cite{Shapiro93}.

We will primarily be interested in the case when $\phi$ is a univalent (i.e. injective) map, with image $U = \phi(\interior{\D})$ a Jordan domain with $C^\infty$ boundary.
Let $\admisall$ denote the semigroup of such $\phi$.

By the smooth Riemann mapping theorem, if $\phi \in \admisall$, then $\phi$ and all of its derivatives extend continuously to $\D$, and $\phi$ induces a diffeomorphism $S^1 \to \partial U$, where $S^1$ is the unit circle.
We will denote the extension of $\phi$ to $\D$, as well as the restriction of this extension to $S^1$, by $\phi$ when there is no danger of confusion.

The deriative $\phi^\prime$ is non-vanishing on $\D$.
Let $\admisall^{NS}$ be the double cover of $\admisall$ consisting of maps $\phi$ equipped with a choice of holomorphic square root $(\phi^\prime)^{1/2}$.
Given $\phi \in \admisall^{NS}$, we define the weighted composition operator $W_\phi \in \cB(H^2(\D))$ by 
$$
(W_{\phi}f)(z) = \phi^\prime(z)^{1/2} f(\phi(z)).
$$
There is a natural structure of a semigroup on $\admisall^{NS}$, and with respect to this composition we have $W_{\phi_1}W_{\phi_2} = W_{\phi_2 \circ \phi_1}$.
The group of invertible elements in $\admisall^{NS}$ is naturally isomorphic to $\Mob^{NS}$, and $W_\phi$ corresponds to the image of $\phi^{-1}|_{S^1} \in \Mob^{NS}$ under the Neveu-Schwarz representation $u_{NS}$.
In particular, if $\phi(\D) = \D$, then $W_{\phi}$ is unitary.

\begin{Proposition}
Let $\phi \in \admisall$, and suppose that $\phi(\interior{\D}) \subsetneq \interior{\D}$ and there exists an $a \in \interior{\D}$ such that $\phi(a) = a$. 
Then there exists a unique univalent map $\sigma:\interior{\D} \to \C$ such that $\sigma(a) = 0$, $\sigma^\prime(a)=1$ and $\sigma(\phi(z)) = \phi^\prime(a) \sigma(z)$ for all $z \in \interior{\D}$.
\end{Proposition}
The map $\sigma$ is called the Koenigs function associated to $\phi$. 
Observe that the map $\phi(\interior{\D}) \to \phi(\interior{\D})$ obtained by conjugating $\phi$ by its Koenigs function is simply scaling by $\phi^\prime(a)$.
Also, when $\sigma \in H^2(\D)$ it is an eigenvector for $C_\phi$.
See \cite[\S6.1]{Shapiro93} for an extended discussion and a proof of the above propositoin.

\begin{Definition}
Let $\admis$ be the subsemigroup of $\admisall$ consisting of univalent maps $\phi: \interior{\D} \to \interior{\D}$ which have the additional properties
\begin{itemize}
\item $\phi(0) = 0$ and $\phi^\prime(0) \in \R_{> 0}$.
\item If $\sigma$ is the Koenigs map associated to $\phi$, then $\sigma(\interior{\D})$ is a Jordan domain with $C^\infty$ boundary.
\end{itemize}
\end{Definition}
Note that if $\phi \in \admis$, then $\phi$ is a Riemann map for $\phi(\interior{\D})$, and so by the smooth Riemann mapping theorem $\phi$ extends smoothly\footnote{
That is, $\phi$ and all of its derivatives extend continuously from $\interior{\D}$ to $\D$, so that the restriction to $S^1$ of the continuous extension of $\phi$ is smooth.} 
to $\D$
, and induces a diffeomorphism between $S^1$ and $\partial \phi(\D)$.

It is easy to produce elements of $\admis$. 
If $U$ is any Jordan domain with $C^\infty$ boundary containing $0$, $\sigma:\interior{\D} \to U$ is a Riemann map with $\sigma(0) = 0$, and $\lambda \in \R_{> 0}$ satisfies $\lambda U \subset U$, then $\phi(z) = \sigma^{-1}(\lambda \sigma(z))$ gives an element of $\admis$.
Indeed, after rescaling, $\sigma$ is the Koenigs map of $\phi$, and every element of $\admis$ arises in this way.
Note that there is a natural embedding $\admis \subset \admisall^{NS}$ by choosing the square root of $\phi^\prime$ so that $\phi^\prime(0)^{1/2} \in \R_{> 0}$.

For $k \in \Z$, let $L^2(S^1)_{\ge k} = z^k H^2(\D) = \overline{ \Span \{ z^n : n \ge k \} }$.
We think of elements of $L^2(S^1)_{\ge k}$ as holomorphic functions on $\interior{\D} \setminus \{0\}$ with (almost everywhere, non-tangential) $L^2$ boundary values.
If $k \ge 0$, then $L^2(S^1) \subseteq H^2(\D)$, and if $\phi \in \admis$ then $L^2(S^1)_{\ge k}$ is invariant under $W_\phi$.
In fact, even when $k < 0$, $W_\phi$ induces a bounded operator on $L^2(S^1)_{\ge k}$ by the usual formula $(W_{\phi}f)(z) = \phi^\prime(z)^{1/2} f(\phi(z)).$
Indeed, if $f \in H^2(\D)$, then
$$
(W_{\phi}z^k f)(z) = z^k \phi^\prime(z)^{1/2}\frac{\phi(z)^k}{z^k} f(\phi(z)) \in L^2(S^1)_{\ge k}.
$$
Hence $W_{\phi}$ induces a linear map $L^2(S^1)_{\ge k} \to L^2(S^1)_{\ge k}$, and this map is bounded since $L^2(S^1)_{\ge k}$ differs from $H^2(\D)$ by a finite dimensional space.

Similarly, we define $L^2(S^1)_{\le k} = \overline{ \Span \{z^n : n \le k\} }$.
Let $c:L^2(S^1) \to L^2(S^1)$ be the map $(cf)(z) = \overline{zf(z)}$. 
That is, $c$ is the antilinear map satisfying $c z^k = z^{-k-1}$.
Hence $cL^2(S^1)_{\ge k} = L^2(S^1)_{\le -k -1}$, and so $\tilde W_\phi := cW_{\phi}c$ is a bounded linear map on each space $L^2(S^1)_{\le k}$.

\begin{Definition}\label{defGoodSemigroups}
Let $\scG$ be the set of families  $(\phi_t)_{t \ge 0} \subset \admis$ which satisfy
\begin{itemize}
\item $\phi_0(z) = z$ for all $z \in \D$.
\item $\phi_t(\phi_s(z)) = \phi_{t+s}(z)$ for all $t,s \in \R_{\ge 0}$ and all $z \in \D$.
\item $(t,z) \mapsto \phi_t(z)$ is a continuous function on $\R_{\ge 0} \times \interior{\D}$.
\item $\phi_t \not \equiv \operatorname{id}$. That is, $\phi_t(z) \ne z$ for some $t > 0$ and $z \in \interior{\D}$ .
\item $\lim_{t \downarrow 0} \frac{\phi_t(z)-z}{t} = -\frac{\sigma(z)}{\sigma^\prime(z)}$ for all $z \in \interior{\D}$, where $\sigma$ is the Koenigs map of $\phi_1$.
\end{itemize}
\end{Definition}
While the final condition of Definition \ref{defGoodSemigroups} may appear strict and unmotivated, we will see below that it is simply a way of choosing an element of the orbit of $\phi_t$ under reparametrization $\phi_t \mapsto \phi_{\alpha t}$.
We will primarily be interested, not in semigroups $\phi_t \in \scG$, but in domains $U$ of the form $U = \phi_t(\interior{\D})$ for some $\phi_t \in \scG$ and some $t > 0$.
Thus we do not lose anything by imposing this final restriction, and it will simplify notation at times.

It is not difficult to produce semigroups $\phi_t \in \scG$.
Let $U$ be a Jordan domain with $C^\infty$ boundary with $0 \in U$ and which is starlike about $0$.
That is, if $z \in U$, then $U$ contains the line segment between $z$ and $0$.
Then if $\sigma:\interior{\D} \to U$ is a Riemann map with $\sigma(0) = 0$ and we set $\phi_t(z) = \sigma^{-1}(e^{-t} \sigma(z))$, we have $\phi_t \in \scG$.

For example, consider the domain $U$ pictured in \eqref{eqnDegenSemigroup} on the left, with the subregion $e^{-t}U$ shaded.
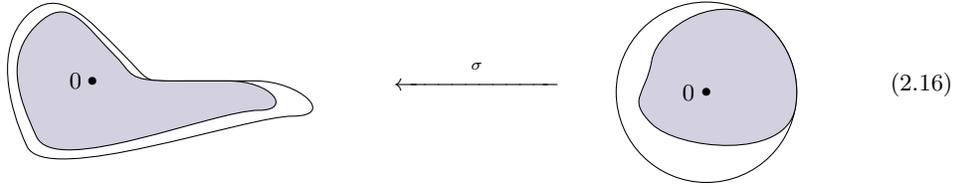
\begin{equation}\label{eqnDegenSemigroup}
\begin{tikzpicture}[scale=1.5,baseline={([yshift=-.5ex]current bounding box.center)}]
	\begin{scope}[xscale=-1,yscale=-1]
	\coordinate (a) at (180:.7cm);
	\coordinate (b) at (225:0.3cm);
%	\coordinate (c) at (270:1cm);
	\coordinate (d) at (-45:.7cm);
	\coordinate (e) at (45:.7cm);
	\coordinate (f) at (170:1.5cm);
	\coordinate (g) at (180:1.2cm);
	\coordinate (aa) at (180:.82cm);
	\coordinate (bb) at (225:0.36cm);
	\coordinate (dd) at (-45:.82cm);
	\coordinate (ee) at (45:.82cm);
	\coordinate (ff) at (170:1.8cm);
	\coordinate (gg) at (180:1.44cm);
%
%	\draw (g) .. controls ++(0:.4cm) and ++(180:.4cm) .. (a);
%	\draw (a) .. controls ++(0:.4cm) and ++(135:.4cm) .. (b);
%	\draw (b) .. controls ++(315:.4cm) and ++(180:.4cm) .. (c);
%	\draw (c) .. controls ++ (0:.4cm) and ++(225:.4cm) .. (d);
%	\draw (d) .. controls ++(45:.4cm) and ++ (295:.4cm) .. (e);
%	\draw (e) .. controls ++(115:.4cm) and ++(0:.4cm) .. (f);
%	\draw (f) .. controls ++(180:.4cm) and ++ (180:.4cm) .. (g);
	\filldraw[fill=white] (gg) .. controls ++(0:.4cm) and ++(180:.4cm) .. (aa) .. controls ++(0:.4cm) and ++(135:.4cm) .. (bb) .. controls ++(315:.4cm) and ++(225:.4cm) .. (dd) .. controls ++(45:.4cm) and ++ (295:.4cm) .. (ee) .. controls ++(115:.4cm) and ++(0:.4cm) .. (ff) .. controls ++(180:.4cm) and ++ (180:.4cm) .. (gg);
	\filldraw[fill=red!10!blue!20!gray!30!white] (g) .. controls ++(0:.4cm) and ++(180:.4cm) .. (a) .. controls ++(0:.4cm) and ++(135:.2cm) .. (b) .. controls ++(315:.4cm) and ++(225:.4cm) .. (d) .. controls ++(45:.4cm) and ++ (295:.4cm) .. (e) .. controls ++(115:.4cm) and ++(0:.4cm) .. (f) .. controls ++(180:.3cm) and ++ (180:.4cm) .. (g);
	\node at (0,0) {\scriptsize{\textbullet}};
	\node at (0:.15cm) {0};
%	\node at (135:1cm){$U$};
%
%	\node at (a) {(a)};
%	\node at (b) {(b)};
%%	\node at (c) {(c)};
%	\node at (d) {(d)};
%	\node at (e) {(e)};
%	\node at (f) {(f)};
%	\node at (g) {(g)};
	\end{scope}
\end{tikzpicture}
\qquad 
\overset{\sigma}{\xleftarrow{\hspace*{2cm}}}
\qquad
\begin{tikzpicture}[scale=1.2,baseline={([yshift=-.5ex]current bounding box.center)}]
	\coordinate (a) at (50:1cm);
	\coordinate (b) at (-10:1cm);
	\coordinate (c) at (150:.7cm);
	\coordinate (d) at (210:.7cm);
% BIG DISK
	\fill[fill=white] (0,0) circle (1cm);
	\draw (0,0) circle (1cm);
% CURVED BOUNDARY REGION
%	\draw ([shift=(240:1cm)]0,0) arc (240:480:1cm);
%	\draw (a) .. controls ++(140:.6cm) and ++(80:.4cm) .. (c);
%	\draw (c) .. controls ++(260:.4cm) and ++(135:.4cm) .. (d);
%	\draw (d) .. controls ++(315:.4cm) and ++(260:.6cm) .. (b);
	\filldraw[fill=red!10!blue!20!gray!30!white] (a) .. controls ++(140:.6cm) and ++(80:.4cm) .. (c) .. controls ++(260:.4cm) and ++(135:.4cm) .. (d) .. controls ++(315:.4cm) and ++(260:.6cm) .. (b) arc (-10:50:1cm);

	\node at (0,0) {\scriptsize{\textbullet}};
	\node at (180:.2cm) {$0$};

% COORDINATE LABELS
%	\node at (a) {(a)};
%	\node at (b) {(b)};
%	\node at (c) {(c)};
%	\node at (d) {(d)};
\end{tikzpicture}
\end{equation}
Observe that $\phi_t(\D)$, the region shaded on the right, intersects $S^1$ in an interval.

In fact, all $\phi_t \in \scG$ arise via the above construction.
If $\phi_t \in \scG$, it is straightforward to check that the $\phi_t$ share a common Koenigs map $\sigma$, and thus $\sigma \circ \phi_t \circ \sigma^{-1}$ acts on $\sigma(\D)$ by rescaling $z \mapsto \lambda(t)z$.
Since this is a semigroup and $\phi_t^\prime(0) > 0$ by assumption, we must have $\lambda(t) = e^{-\alpha t}$ for some $\alpha \in \R_{> 0}$, and so 
\begin{equation}\label{eqSemigroupForm}
\phi_t(z) = \sigma^{-1}(e^{-\alpha t} \sigma(z)).
\end{equation}
From this formula we can see that $(t,z) \mapsto \phi_t(z)$ is smooth.

By \cite[Thm 1.1]{BerksonPorta78}, there is a unique holomorphic function $G: \interior{\D} \to \C$ such that
$$
\frac{\partial}{\partial t} \phi_t(z) = G(\phi_t(z))
$$
for $t \in \R_{> 0}$ and $z \in \interior{\D}$. 
Using \eqref{eqSemigroupForm}, we can compute
$$
G(z) = \lim_{t \downarrow 0} \frac{\partial}{\partial t} \phi_t(z) = -\alpha \frac{\sigma(z)}{\sigma^\prime(z)}.
$$
For each $z \in \interior{\D}$, the the map $t \mapsto \phi_t(z)$ extends smoothly to a neighborhood of $\R_{\ge 0}$, and so we have
$$
\lim_{t \downarrow 0} \frac{\partial}{\partial t} \phi_t(z) = \lim_{t \downarrow 0} \frac{\phi_t(z)-z}{t}.
$$
Since $\phi_t \in \scG$, we must therefore have $\alpha = 1$.
By standard results (see e.g. \cite[Thm. 2.9]{Duren83}), since $\sigma(\interior{\D})$ is starlike, we have 
\begin{equation}\label{eqPositiveRealPart}
\Re \frac{\sigma(z)}{z \sigma^\prime(z)} = -\Re \frac{G(z)}{z} \ge 0.
\end{equation}

We summarize the above discussion in the following proposition.
\begin{Proposition}\label{propSemigroupSummary}
Let $U$ be a Jordan domain with $C^\infty$ boundary.
Suppose that $0 \in U$ and that $U$ is starlike about $0$.
Let $\sigma:\interior{\D} \to U$ be a Riemann map with $\sigma(0) = 0$.
Let $\phi_t(z) = \sigma^{-1}(e^{-t} \sigma(z))$.
Then $\sigma$ and $\phi_t$ extend smoothly to $\D$, $\phi_t$ induces a diffeomorphism between $S^1$ and $\partial \phi_t(\D)$, and $(\phi_t)_{t \ge 0} \in \scG$.
Moreover, every semigroup $(\phi_t)_{t \ge 0} \in \scG$ arises in this way, and after rescaling by common scalar, $\sigma$ is the Koenigs map for every $\phi_t$ with $t > 0$.
The holomorphic map
$$
\rho(z) = \frac{\sigma(z)}{z \sigma^\prime(z)}
$$
satisfies $\Re(\rho(z)) \ge 0$ for all $z \in \D$.
\end{Proposition}

\section{Free fermion Segal CFT for degenerate Riemann surfaces}\label{secSegalCFTForDegenerateSurfaces}

\subsection{Degenerate Riemann surfaces and their Hardy spaces}

The main idea of this paper is to extend the notion of Segal CFT to allowed degenerate Riemann surfaces such as
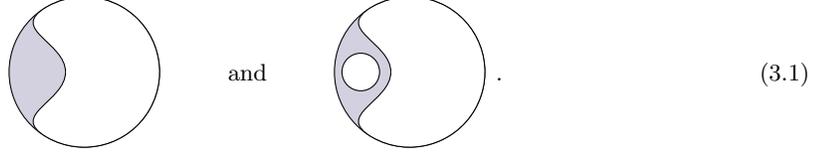
\begin{equation}\label{eqnDegenerateSurfaceCartoons}
\begin{tikzpicture}[baseline={([yshift=-.5ex]current bounding box.center)}]
	\coordinate (a) at (120:1cm);
	\coordinate (b) at (240:1cm);
	\coordinate (c) at (180:.25cm);
% BIG DISK
	\fill[fill=red!10!blue!20!gray!30!white] (0,0) circle (1cm);
	\draw (0,0) circle (1cm);
% CURVED BOUNDARY REGION
	\fill[fill=white] (a)  .. controls ++(210:.6cm) and ++(90:.4cm) .. (c) .. controls ++(270:.4cm) and ++(150:.6cm) .. (b) -- ([shift=(240:1cm)]0,0) arc (240:480:1cm);
	\draw ([shift=(240:1cm)]0,0) arc (240:480:1cm);
	\draw (a) .. controls ++(210:.6cm) and ++(90:.4cm) .. (c);
	\draw (b) .. controls ++(150:.6cm) and ++(270:.4cm) .. (c);
\end{tikzpicture}
\qquad 
\mbox{ and }
\qquad
\,\,\begin{tikzpicture}[baseline={([yshift=-.5ex]current bounding box.center)}]
	\coordinate (a) at (120:1cm);
	\coordinate (b) at (240:1cm);
	\coordinate (c) at (180:.25cm);
% BIG DISK
	\fill[fill=red!10!blue!20!gray!30!white] (0,0) circle (1cm);
	\draw (0,0) circle (1cm);
% CURVED BOUNDARY REGION
	\fill[fill=white] (a)  .. controls ++(210:.6cm) and ++(90:.4cm) .. (c) .. controls ++(270:.4cm) and ++(150:.6cm) .. (b) -- ([shift=(240:1cm)]0,0) arc (240:480:1cm);
	\draw ([shift=(240:1cm)]0,0) arc (240:480:1cm);
	\draw (a) .. controls ++(210:.6cm) and ++(90:.4cm) .. (c);
	\draw (b) .. controls ++(150:.6cm) and ++(270:.4cm) .. (c);
% INNER DISK
	\filldraw[fill=white] (180:.65cm) circle (.25cm); 
\end{tikzpicture}.
\end{equation}
In these pictures, the shaded region should be thought of as consisting of a (shaded) disk with one or two (unshaded) regions removed, in such a way that the boundary of the outer disk partially coincides with the boundary of the removed regions.
We will only consider degenerate Riemann surfaces of a special form, in which the annular region is obtained by removing from the unit disk its image under an element of a one-parameter semigroup of holomorphic self maps of the disk.
These degenerate surfaces are sufficient to produce conformal nets, although in future work we hope to treat more general degenerate surfaces.

We now make precise exactly what data we will use for degenerate Riemann surfaces in this paper.

\begin{Definition}\label{defDegenerateAnnulus}
A \emph{degenerate annulus} is a tuple $(\phi_t, t)$ with $(\phi_t)_{t \in \R_{\ge 0}} \in \scG$ (see Definition \ref{defGoodSemigroups}) and $t \in \R_{> 0}$.
The \emph{underlying space} of $(\phi_t, t)$ is the compact space $\Sigma = \D \setminus \phi_t(\interior{\D})$.
The \emph{incoming} and \emph{outgoing boundaries} are given by $\partial \Sigma^0 = \phi_t(S^1)$ and $\partial \Sigma^1 = S^1$, respectively.
The \emph{boundary} $\partial \Sigma$ is by definition the disjoint union $\partial \Sigma = \partial \Sigma^0 \sqcup \partial \Sigma^1$.
We let $\pi_0(\partial \Sigma)$ and $\pi_0(\partial \Sigma^i)$ be the sets of connected components, i.e.
$$
\pi_0(\partial \Sigma^1) = \{S^1\}, \pi_0(\partial\Sigma^0) = \{\phi_t(S^1)\}, \pi_0(\partial \Sigma) = \{S^1, \phi_t(S^1)\}.
$$
\emph{Boundary parametrizations} for a degenerate annulus are a pair of diffeomorphisms 
$$
\gamma=(\gamma_j)_{j \in \pi_0(\partial \Sigma)} : \bigsqcup_{j \in \pi_0(\partial \Sigma)} S^1 \to \partial \Sigma
$$
which preserve counterclockwise orientations about $0$, along with choices of smooth square roots
$$
\psi_j^2 = \frac{d}{dz} \gamma_j =: \gamma_j^\prime.
$$
The \emph{standard boundary parametrization} for a degenerate annulus $(\phi_t,t)$ is given by the diffeomorphisms $\operatorname{id}$ and $\phi_t|_{S^1}$, along with the choices of square roots $1$ and $\psi_t|_{S^1}$, where $\psi_t$ is the square root of $\phi_t$ with $\psi_t^\prime(0) > 0$.
We denote by $\cDA$ the collection of all degenerate annuli with boundary parametrizations $X=(\phi_t, t, \gamma, \psi)$, and by $\cDA_{st}$ the subcollection of ones that have the standard boundary parametrization.
We will often refer to an element $(\phi_t, t) \in \cDA_{st}$, leaving the boundary parametrizations implicit.
\end{Definition}
One should think of the data $(\phi_t, t)$ as capturing the degenerate surface $\D \setminus \phi_t(\interior{\D})$, as depicted on the left in \eqref{eqnDegenerateSurfaceCartoons}.
We think of this degenerate surface as inheriting a spin structure from the standard spin structure on $\C$, and the boundary parametrizations provide trivializations of the restriction of this spin structure to the boundary.
It would perhaps be more accurate to call elements of $\cDA$ `degenerate \emph{spin} annuli.'
When $\phi_t(\D) \subset \interior{\D}$, so that the `degenerate' surface is actually a genuine Riemann surface with boundary, this philosophy can be made precise (see Proposition \ref{propAgreesOnNondegenerate}).
Of course, Definition \ref{defDegenerateAnnulus} only captures the special class of degenerate annuli that are induced by one-parameter families of univalent maps $(\phi_t)_{t \ge 0}$.

We now move from degenerate annuli to degenerate pairs of pants.
\begin{Definition}\label{defDegeneratePants}
A \emph{degenerate pair of pants} is a degenerate annulus $(\phi_t, t)$, along with choices $w \in \interior{\D}$ and $s \in \R_{>0}$ such that $w + s\D \subset \interior{\D} \setminus \phi_t(\D)$.
The \emph{underlying space} of a degenerate pair of pants is the compact space $\Sigma = \D \setminus \big((w + s \interior{\D}) \cup \phi_t(\interior{\D})\big)$.
The \emph{incoming} and \emph{outgoing boundaries} are given by $\partial \Sigma^0 = \phi_t(S^1) \cup (w + s S^1)$ and $\partial \Sigma^1 = S^1$, respectively.
The \emph{boundary} $\partial \Sigma$ is by definition the disjoint union $\partial \Sigma = \partial \Sigma^0 \sqcup \partial \Sigma^1$.
We let $\pi_0(\partial \Sigma)$ and $\pi_0(\partial \Sigma^i)$ be the sets of connected components, i.e.
$$
\pi_0(\partial \Sigma^1) = \{S^1\}, \pi_0(\partial\Sigma^0) = \{\phi_t(S^1), w + sS^1\}, \pi_0(\partial \Sigma) = \{S^1, \phi_t(S^1), w + sS^1\}.
$$
\emph{Boundary parametrizations} for a degenerate pair of pants are a triple of diffeomorphisms 
$$
\gamma=(\gamma_j)_{j \in \pi_0(\partial \Sigma)} : \bigsqcup_{j \in \pi_0(\partial \Sigma)} S^1 \to \partial \Sigma
$$
which preserve the counterclockwise orientations of $S^1$ and $\phi_t(S^1)$ about $0$, and the counterclockwise orientation of $w + s S^1$ about $w$, along with choices of smooth square roots
$$
\psi_j^2 = \frac{d}{dz} \gamma_j =: \gamma_j^\prime.
$$
The \emph{standard boundary parametrizations} for a degenerate pair of pants is given by the map $w + sz$ and the positive square root of $s$, along with the standard boundary parametrization for the degenerate annulus $(\phi_t, t)$.
We denote by $\cDP$ the collection of all degenerate annuli with boundary parametrizations $X=(\phi_t, t, w,s,\gamma, \psi)$, and by $\cDP_{st}$ the subcollection of ones that have the standard boundary parametrization.
We will often refer to an element $(\phi_t, t, w, s) \in \cDP_{st}$, leaving the boundary parametrizations implicit.
Elements of $\cDA$ or $\cDP$ are called \emph{degenerate Riemann surfaces}, and we set $\cDR = \cDA \sqcup \cDP$ and $\cDR_{st} = \cDA_{st} \sqcup \cDP_{st}$.
We say that $X \in \cDR$ is \emph{non-degenerate} if $\phi_t(\D) \subset \interior{\D}$, or equivalently if $\Sigma$ is a Riemann surface with boundary.
\end{Definition}
\begin{figure}[h!]
$$
(\phi_t, t, w, s)
\qquad
\longleftrightarrow
\qquad
\begin{tikzpicture}[scale=2,baseline={([yshift=-.5ex]current bounding box.center)}]
	\coordinate (a) at (120:1cm);
	\coordinate (b) at (240:1cm);
	\coordinate (c) at (180:.00cm);
% BIG DISK
	\fill[fill=red!10!blue!20!gray!30!white] (0,0) circle (1cm);
	\draw (0,0) circle (1cm);
% CURVED BOUNDARY REGION
	\fill[fill=white] (a)  .. controls ++(210:.6cm) and ++(90:.4cm) .. (c) .. controls ++(270:.4cm) and ++(150:.6cm) .. (b) -- ([shift=(240:1cm)]0,0) arc (240:480:1cm);
	\draw ([shift=(240:1cm)]0,0) arc (240:480:1cm);
	\draw (a) .. controls ++(210:.6cm) and ++(90:.4cm) .. (c);
	\draw (b) .. controls ++(150:.6cm) and ++(270:.4cm) .. (c);
% INNER DISK
	\filldraw[fill=white] (180:.65cm) circle (.20cm); 
	%\node at (180:.65cm) {\scriptsize{$\xi$}};
% POINT LABELS
%	\node at (0:1cm) {\scriptsize{\textbullet}};
%	\node at (0:0.8cm) {1};
%	\node at (0:1.0cm) {\scriptsize{\textbullet}};
%	\node at (0:1.1cm) {1};
	\node at (180:.65cm) {\scriptsize{\textbullet}};
	\node at (188:.7cm) {w};
	\node at (0:0.1) {\scriptsize{\textbullet}};
	\node at (0:.2cm) {0};
	\node at (180:.45) {\scriptsize{\textbullet}};
	\node at (180:.25) {\scriptsize{$w$+$s$}};
	\node at (155:1.1cm) {$S^1$};
	\node at (113:.6cm){$\phi_t(S^1)$};
% COORDINATE LABELS
%	\node at (a) {(a)};
%	\node at (b) {(b)};
%	\node at (c) {(c)};
\end{tikzpicture}
$$
\caption{The geometric interpretation of a degenerate pair of pants}
\end{figure}
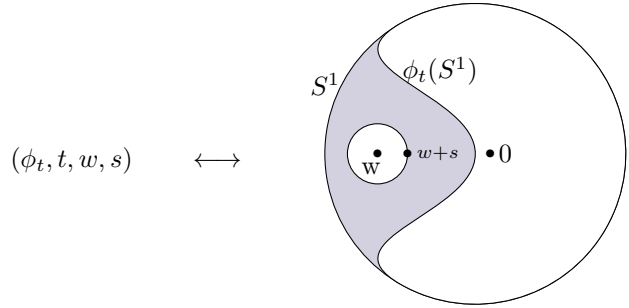

As with degenerate annuli, it would perhaps be more accurate to call elements of $\cDP$ `degenerate \emph{spin} pairs of pants,' but we will generally not do so.
As before, the given definitions of degenerate pair of pants and degenerate Riemann surface obviously only include a special class of a more general notion.

We will now extend the free fermion Segal CFT to take values on $X \in \cDR$.
As with non-degenerate Riemann surfaces, we first need to define a Hardy space $H^2(X)$.
\begin{Definition}\label{defDegenerateHardy}
Let $X \in \cDR$, and let $\Sigma$ be the underlying space of $X$.
The pre-quantized boundary Hilbert spaces are given by $H_{\partial \Sigma}^i = \bigoplus_{j \in \pi_0(\partial\Sigma^i)} L^2(S^1)$, and $H_{\partial \Sigma} = H_{\partial \Sigma}^1 \oplus H_{\partial \Sigma}^0$.
Let $\cO(\Sigma)$ be the space of functions holomorphic on some open set $U$ containing $\Sigma$, and let $(\psi_j,\gamma_j)_{j \in \pi_0(\partial \Sigma)}$ be the boundary parametrization for $X$.
The Hardy space $H^2(X)$ is given by
$$
H^2(X) = \overline{\{  \psi \cdot (F \circ \gamma) : F \in \cO(\Sigma)\}} \subset H_{\partial \Sigma}.
$$
\end{Definition}

As with the free fermion Segal CFT, we want to assign to $X$ the space of linear maps with satisfy the $H^2(X)$ commutation relations.
Let $H = L^2(S^1)$, and let $p \in \cB(H)$ be the projection onto the classical Hardy space $H^2(\D)$, and let 
$$
p_i = \bigoplus_{j \in \pi_0(\partial\Sigma^i)} p \in \cB(H_{\partial \Sigma}^i).
$$
As usual, we will write $\F$ for $\F_{H, p}$, and we set $\F_{\partial \Sigma}^i = \bigotimes_{j \in \pi_0(\partial \Sigma^i)} \F$. 
When $X \in \cDP$,  we identify $\F_{\partial \Sigma}^0$ with $\F_{H_{\partial \Sigma}^0,p_0}$ via Proposition \ref{propFockSumToTensor}, ordering the tensor factors so that the one indexed by $w + s S^1$ comes first.

\begin{Definition}\label{defDegenerateFermionSegalCFT}
Let $X \in \cDR$. 
Then we define $E(X)$ to be space of all bounded linear maps $T \in \cB(\F_{\partial \Sigma}^0, \F_{\partial \Sigma}^1)$ which satisfy the $H^2(X)$ commutation relations.
That is, those $T$ which satisfy
\begin{equation*}
a(f^1)T = (-1)^{p(T)} Ta(f^0), \qquad a(\overline{zf^1})^*T = (-1)^{p(T)} T a(\overline{zf^0})^*
\end{equation*}
for every $(f^1,f^0) \in H^2(X) \subset H_{\partial \Sigma}^1 \oplus H_{\partial \Sigma}^0$.
For non-homogeneous $T$, the commutation relations are extended linearly, or equivalently by requiring that both the even and odd part of $T$ satisfy the commutation relations.
\end{Definition}

\begin{Remark}
To match the definition of the Segal CFT for non-degenerate surfaces, it would have been better to define $E(X)$ as the space of maps satisfying
$$
a(f^1)T = (-1)^{p(T)} Ta(f^0), \qquad a(g^1)^*T = -(-1)^{p(T)} T a(g^0)^*
$$
for every $(f^1,f^0) \in H^2(X)$ and every $(g^1,g^0) \in H^2(X)^\perp$.
In the non-degenerate case, we have $H^2(X)^\perp = \overline{ M_{\pm z}H^2(X)}$, where $M_{\pm z}$ is multiplication by $(-1)^{i} z$ on $H_{\partial \Sigma}^i$, and complex conjugation is taken pointwise, and so the two definitions are equivalent.
In the non-degenerate case, it is easy to show that $\overline{M_{\pm z} H^2(X)} \subset H^2(X)^\perp$, and so the space $E(X)$ described in Definition \ref{defDegenerateFermionSegalCFT} could, \emph{a priori}, be larger than the space defined using $H^2(X)$.
However, we will show that even with the weaker constraint, we have $\dim E(X) = 1$, and so the spaces $E(X)$ are not too large.
\end{Remark}

The goal of Section \ref{secSegalCFTForDegenerateSurfaces} is to show that the spaces $E(X)$ are one-dimensional, and that they admit a nice description in terms of the Virasoro algebra action on $\F$ and free fermion vertex operators.

For simplicity, we would like to work with degenerate Riemann surfaces $X \in \cDR_{st}$ with standard boundary parametrization.
The following proposition allows us to reduce any questions about $E(X)$ for arbitrary $X \in \cDR$ to one about the corresponding element $X_{st} \in \cDR_{st}$.

\begin{Proposition}\label{propDegenerateOperatorReparametrization}
Let $X_1, X_2 \in \cDR$.
Suppose that $X_1$ and $X_2$ are the same except for having differing boundary parametrizations $(\psi^{(1)},  \gamma^{(1)})$ and $(\psi^{(2)},  \gamma^{(2)})$, respectively.
Let $\Sigma$ be the underlying surface of the $X_i$, and for $j \in \pi_0(\partial \Sigma)$ let $\gamma_j = (\gamma_j^{(2)})^{-1} \circ \gamma_j^{(1)} \in \Diff_+(S^1)$ and $\psi_j = \frac{\psi_j^{(2)}}{\psi_j^{(1)} \circ \gamma_j^{-1}}$, so that $(\psi_j, \gamma_j) \in \Diff_+^{NS}(S^1)$.
Then 
\begin{equation}\label{eqnHardyReparamRel}
H^2(X_2) = \left(\bigoplus_{j \in \pi_0(\partial \Sigma)} u_{NS}(\psi_j, \gamma_j) \right)H^2(X_1)
\end{equation}
and 
\begin{equation}\label{eqnSegalReparamRel}
E(X_2) = U_{NS}(\psi_{S^1},\gamma_{S^1}) E(X_1) \left(\bigotimes_{j \in \pi_0(\partial \Sigma^0)} U_{NS}(\psi_j,\gamma_j)^* \right).
\end{equation}
\end{Proposition}
\begin{proof}
Let $U$ be a neighborhood of $\Sigma$, and let $F \in \cO(U)$, so that 
$$
f^{(i)}:=\psi^{(i)} \cdot (F \circ \gamma^{(i)}) \in H^2(X^{(i)}).
$$
Moreover,
$$
u_{NS}(\psi_j,\gamma_j) f^{(1)}_j = f^{(2)}_j,
$$
and consequently
$$
\{  \psi^{(2)} \cdot (F \circ \gamma^{(2)}) : F \in \cO(\Sigma)\} = \Big(\bigoplus_{j \in \pi_0(\partial \Sigma)} u_{NS}(\psi_j, \gamma_j) \Big)\{  \psi^{(1)} \cdot (F \circ \gamma^{(1)}) : F \in \cO(\Sigma)\}.
$$
Taking closures yields the desired relation \eqref{eqnHardyReparamRel} for the Hardy spaces $H^2(X^{(i)})$.

For $(\psi, \gamma) \in \Diff_+^{NS}(S^1)$, we have 
$$
\abs{\psi(z)}^2 = \abs{(\gamma^{-1})^\prime(z)} = \frac{z (\gamma^{-1})^\prime(z)}{\gamma^{-1}(z)} = \frac{z \psi(z)^2}{\gamma^{-1}(z)},
$$
and thus
\begin{equation}\label{eqnccommutes}
\overline{z \psi(z)} = \psi(z)\overline{\gamma^{-1}(z)}
\end{equation}
for all $z \in S^1$.
A direct consequence of \eqref{eqnccommutes} is that $u_{NS}(\psi,\gamma)$ commutes with the antilinear map $c:L^2(S^1) \to L^2(S^1)$ given by $(cf)(z) = \overline{z f(z)}$.
Hence
\begin{equation}\label{eqnHardyPerpReparamRel}
\overline{M_{\pm z} H^2(X_2)} = \left( \bigoplus_{j \in \pi_0(\partial \Sigma)} u_{NS}(\psi_j,\gamma_j) \right) \overline{M_{\pm z} H^2(X_1)},
\end{equation}
where $M_{\pm z}$ is multiplication by $z$ on copies of $L^2(S^1)$ indexed by $j \in \pi_0(\partial \Sigma^1)$ and multiplication by $-z$ on copies corresponding to incoming boundary.
Now the relation \eqref{eqnSegalReparamRel} can be verified directly from \eqref{eqnHardyReparamRel} and \eqref{eqnHardyPerpReparamRel}, just as in \cite[Prop. 4.12]{Ten16}.
\end{proof}

Now to study $E(X)$ when $X \in \cDR_{st}$, we will often want to approximate $X$ by non-degenerate spin Riemann surfaces, as follows.
\begin{Definition}\label{defNondegenerateExtension}
Let $X = (\phi_t, t, w, s) \in \cDR_{st}$ and let $R > 1$.
Then the \emph{non-degenerate extension} of $X$ by $R$ is the spin Riemann surface with parametrized boundary $X_R=(\Sigma_R,L_R,\Phi_R,\beta_R) \in \cR$ given as follows.
Let $\Sigma$ be the underlying space of $X$, and let $\Sigma_R = \Sigma \cup \{1 \le \abs{z} \le R\}$.
Let $(L_R, \Phi_R)$ be the spin structure on $\Sigma_R$ obtained from restricting the standard Neveu-Schwarz spin structure on $\C$.
Define the boundary parametrization $\beta_R : \bigsqcup_{j \in \pi_0(\partial \Sigma_R)} (S^1,NS) \to L|_{\partial \Sigma}$ so that
$$
\beta_{R,j}^*f = \left\{ 
\begin{array}{cl}
R^{1/2}f(R z), &\quad j = RS^1, \,\,f \in C^\infty(RS^1)\\
\psi_t(z)f(\phi_t(z)), &\quad j = \phi_t(S^1), \,\, f \in C^\infty(\phi_t(S^1))
\end{array}
\right.
$$
where $R^{1/2}$ is the positive square root, and $\psi_t$ is the square root of $\phi_t^\prime$ with $\psi_t^\prime(0) > 0$.
If $X \in \cDP_{st}$, then additionally choose $\beta_{R,w+sS^1}$ so that $\beta_{R,w+sS^1}^*f = s^{1/2}f(w + sz)$ for $f \in C^\infty(w + sS^1)$, where $s^{1/2}$ is the positive square root.
\end{Definition}

\begin{figure}[h!]
$$
X \, = \,
\begin{tikzpicture}[scale=1.2,baseline={([yshift=-.5ex]current bounding box.center)}]
	\coordinate (a) at (120:1cm);
	\coordinate (b) at (240:1cm);
	\coordinate (c) at (180:.25cm);
% BIG DISK
	\fill[fill=red!10!blue!20!gray!30!white] (0,0) circle (1cm);
	\draw (0,0) circle (1cm);
% CURVED BOUNDARY REGION
	\fill[fill=white] (a)  .. controls ++(210:.6cm) and ++(90:.4cm) .. (c) .. controls ++(270:.4cm) and ++(150:.6cm) .. (b) -- ([shift=(240:1cm)]0,0) arc (240:480:1cm);
	\draw ([shift=(240:1cm)]0,0) arc (240:480:1cm);
	\draw (a) .. controls ++(210:.6cm) and ++(90:.4cm) .. (c);
	\draw (b) .. controls ++(150:.6cm) and ++(270:.4cm) .. (c);
% INNER DISK
	\filldraw[fill=white] (180:.65cm) circle (.25cm); 
	\node at (0:1cm) {\scriptsize{\textbullet}};
	\node at (0:0.8cm) {1};
%	\node at (0:1.2cm) {\scriptsize{\textbullet}};
%	\node at (0:1.4cm) {R};
	\node at (180:.65cm) {\scriptsize{\textbullet}};
	\node at (193:.7cm) {w};

\end{tikzpicture}
\,\, 
\qquad \leadsto \qquad 
X_R = 
\begin{tikzpicture}[scale=1.2,baseline={([yshift=-.5ex]current bounding box.center)}]
	\coordinate (a) at (120:1cm);
	\coordinate (b) at (240:1cm);
	\coordinate (c) at (180:.25cm);
% BIG DISK
	\fill[fill=red!10!blue!20!gray!30!white] (0,0) circle (1.2cm);
	\draw (0,0) circle (1.2cm);
% CURVED BOUNDARY REGION
	\fill[fill=white] (a)  .. controls ++(210:.6cm) and ++(90:.4cm) .. (c) .. controls ++(270:.4cm) and ++(150:.6cm) .. (b) -- ([shift=(240:1cm)]0,0) arc (240:480:1cm);
	\draw ([shift=(240:1cm)]0,0) arc (240:480:1cm);
	\draw (a) .. controls ++(210:.6cm) and ++(90:.4cm) .. (c);
	\draw (b) .. controls ++(150:.6cm) and ++(270:.4cm) .. (c);
% INNER DISK
	\filldraw[fill=white] (180:.65cm) circle (.25cm); 
	\node at (0:1cm) {\scriptsize{\textbullet}};
	\node at (0:0.8cm) {1};
	\node at (0:1.2cm) {\scriptsize{\textbullet}};
	\node at (0:1.4cm) {R};
	\node at (180:.65cm) {\scriptsize{\textbullet}};
	\node at (193:.7cm) {w};

\end{tikzpicture}.
$$
\caption{The non-degenerate extension $X_R$ of $X$ is a Riemann surface with boundary.}
\end{figure}
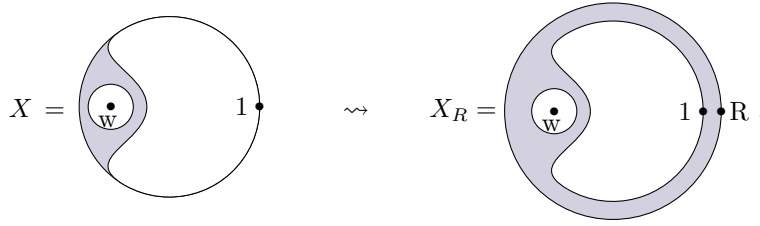

Informally, we think of $X$ as the limit of $X_R$ as $R \downarrow 1$, and in the proof of Theorem \ref{thmFermionPantsBoundedness} we will show that elements $T \in E(X)$ can be obtained as limits (in the strong operator topology) of $T_R \in E(X_R)$.

While it is somewhat involved to show that the spaces $E(X)$ are non-trivial, it is easier to check that they are not too large.

\begin{Proposition}\label{propDimAtMostOne}
Let $X \in \cDR$. Then $\dim E(X) \le 1$, and elements of $E(X)$ are even.
\end{Proposition}
\begin{proof}
By Proposition \ref{propDegenerateOperatorReparametrization}, we may assume without loss of generality that $X \in \cDR_{st}$.
The arguments are essentially identical in the two cases $X \in \cDA_{st}$ and $X \in \cDP_{st}$. 
We will assume that we are in the latter case.

Let $T \in E(X) \subset \cB(\F \otimes \F, \F)$, with the inputs ordered so that the one corresponding to $sS^1 + w$ comes first.
When $n \ge 0$, we have $(z^n, s^{1/2}(sz+w)^n, W_{\phi_t} z^n) \in H^2(X)$, where $W_{\phi_t}$ is the weighted composition operator defined in Section \ref{subsecCompositionOperators}.
Hence
$$
a(z^n)T(\Omega \otimes \Omega) = T(a(s^{1/2}(sz + w)^n)\Omega \otimes \Omega) + T(\Omega \otimes a(W_{\phi_t} z^n)\Omega) = 0
$$
since $(sz+w)^n, W_{\phi_t}z^n \in H^2(\D)$.

For $n \ge 0$ we also have
\begin{equation}\label{eqnTOmegaSatisfiesVacuum1}
(z^{-n-1}, s^{1/2}z^{-1}(sz^{-1} + \overline{w})^n, \overline{z W_{\phi_t} z^n}) \in \overline{z H^2(X)}
\end{equation}
Since $f \mapsto \overline{z f}$ exchanges $H^2(\D)$ and $H^2(\D)^\perp$, we have $z^{-1}(sz^{-1} + \overline{w})^n, \overline{z W_{\phi_t} z^n} \in H^2(\D)^\perp$.
Hence
\begin{equation}\label{eqnTOmegaSatisfiesVacuum2}
a(z^{-n-1})^*T(\Omega \otimes \Omega) = T(a(s^{1/2}z^{-1}(sz^{-1} + \overline{w})^n)^*\Omega \otimes \Omega) + T(\Omega \otimes a(\overline{z W_{\phi_t} z^n})^*) = 0.
\end{equation}
Combining \eqref{eqnTOmegaSatisfiesVacuum1} and \eqref{eqnTOmegaSatisfiesVacuum2}, we see that $T(\Omega \otimes \Omega) = \alpha \Omega$ for some $\alpha \in \C$ by Theorem \ref{thmShaleStinespring}.

Thus to show $\dim E(X) \le 1$, it suffices to show that $T \in E(X)$ is determined by $T(\Omega \otimes \Omega)$.
We will assume that $T(\Omega \otimes \Omega) = 0$, and prove that $T = 0$.
It is clear that if $T \in E(X)$, then both the even and odd parts of $T$ also lie in $T(X)$, and so we may assume without loss of generality that $T$ is homogeneous.

By assumption, $\ker T \ne 0$.
We will show that $\ker T$ is invariant under $\CAR(H_{\partial \Sigma}^0)$, and since this algebra acts irreducibly on $\F_{\partial \Sigma}^1 = \F \otimes \F$, the result will follow.

Let $K$ be the image of $H^2(X)$ under the projection of $H_{\partial \Sigma}$ onto $H_{\partial \Sigma}^0$. 
By Runge's theorem, $\cl K = H_{\partial \Sigma}^0$.
By definition, given $f^0 \in K$, there exists a $f^1 \in H_{\partial \Sigma}^1$ such that $(f^1, f^0) \in H^2(X)$.
Hence $a(f^1)T = (-1)^{p(T)}Ta(f^0)$, from which we can see that $\ker T$ is invariant under $a(f^0)$.
Since $f \mapsto a(f)$ is continuous, we get that $\ker T$ is invariant under $a(f)$ for all $f \in H_{\partial \Sigma}^0$.

Similarly, let $\tilde K$ be the image of $\overline{ z H^2(X)}$ under the projection of $H_{\partial \Sigma}$ onto $H_{\partial \Sigma}^0$, which is dense in $H_{\partial \Sigma}^0$.
For $g^0 \in \tilde K$, we have $a(g^1)^*T = (-1)^{p(T)} Ta(g^0)^*$ for some $g^1 \in H_{\partial \Sigma}^1$.
Hence $\ker T$ is invariant under $a(g^0)^*$ for all $g^0 \in \tilde K$, and thus it is invariant under $a(g)^*$ for all $g \in H_{\partial \Sigma}^0$.

We have shown that $\ker T$ is a non-zero $\CAR(H_{\partial \Sigma}^0)$ subrepresentation of $\F_{\partial \Sigma}^0$, thus $T = 0$, as desired. Moreover, if $T$ is odd and $T \in E(X)$, then we must have $T(\Omega \otimes \Omega) = 0$, and so $T = 0$ by the above argument.
\end{proof}

Using the uniqueness result of Proposition \ref{propDimAtMostOne}, we can show that when $X \in \cDR$ is non-degenerate the definition of $E(X)$ coincides with the space assigned by the free fermion Segal CFT.

\begin{Proposition}\label{propAgreesOnNondegenerate}
Let $X  \in \cDR$, and suppose that $X$ is non-degenerate, so that the underlying space $\Sigma$ is a Riemann surface with boundary.
Let $(\Phi,L)$ be the spin structure on $\Sigma$ obtained by restricting the standard Neveu-Schwarz spin structure on $\C$, given by \eqref{eqnStandardCSpin}.
Let $\beta:\bigsqcup_{j \in \pi_0(\partial \Sigma)} (S^1, NS) \to L|_{\partial \Sigma}$ be the spin isomorphism characterized by $\beta^* f = \psi \cdot (f \circ \gamma)$ for $f \in C^\infty(L|_{\partial \Sigma})$.
Let $\tilde X = (\Sigma, L, \Phi, \beta) \in \cR$, and let $H^2(\tilde X)$ and $E(\tilde X)$ be the Hardy space and Segal CFT for non-degenerate surfaces, as in Section \ref{subsubsecFreeFermionSegalCFT}.
Then $H^2(X) = H^2(\tilde X)$ and $E(X) = E(\tilde X)$.
\end{Proposition}
\begin{proof}
Both $H^2(X)$ and $H^2(\tilde X)$ are given by pullbacks of holomorphic functions on $\Sigma$, with the only difference being that $H^2(X)$ requires that the functions be holomorphic in a nieghborhood of $\Sigma$, and $H^2(\tilde X)$ only requires that they extend smoothly to the boundary.
However, by Runge's theorem we may approximate any element of $H^2(X)$ arbitrarily well by an element of $H^2(\tilde X)$, and since both spaces are closed, they coincide.
Both $E(X)$ and $E(\tilde X)$ consist of maps which satisfy certain commutation relations derived from $H^2(X)$.
Since $H^2(X)^\perp = \overline{M_{\pm z} H^2(X)}$ by \cite[Thm. 6.1]{Ten16}, the commutation relations they're required to satisfy are identical.
However, elements of $E(\tilde X)$ are also required to be trace class, so that $E(\tilde X) \subset E(X)$.
But $\dim E(\tilde X) = 1$ by Theorem \ref{thmFermionSegalCFT}, and by Proposition \ref{propDimAtMostOne} $\dim E(X) \le 1$, so the two spaces must coincide.
\end{proof}

\subsection{Calculation of Segal CFT operators}\label{subsecCalculationSegalCFT}

In Section \ref{subsecCalculationSegalCFT}, we will give an explicit description of the spaces $E(X)$ for $X \in \cDR$, in terms of the free fermion vertex operator superalgebra.
We will briefly recall notation; for a more detailed overview, see Section \ref{secPreliminaries}.
Let $H = L^2(S^1)$, and let $p \in \cB(H)$ be the projection onto the Hardy space $H^2(\D)$.
Let $\F=\F_{H,p}$ be fermionic Fock space, and we write $a(f)$ for the action of $\CAR(H)$ on $\F$.
Let $\F^0$ be the subspace of finite energy vectors, and let $L_n$ be the unitary positive energy representation of $\Vir$ on $\F^0$ coming from the conformal vector $\nu \in \F^0$.
Let $\F^\infty \subset \F$ be the space of smooth vectors for $\overline{1 + L_0}$.
Given a function $f \in C^\infty(S^1)$, we write $L(f)$ for the closure of $\sum_{n \in \Z} \hat f _n L_n$.
If $\rho \in H^2(\D)$ and $\rho$ extends smoothly to $S^1$, then we write $L(\rho)$ for $L(\rho|_{S^1})$.

The main result of this section is the following.
\begin{Theorem}\label{thmBoundednessAndExistence}
Let $X=(\phi_t, t) \in \cDA_{st}$ be a degenerate annulus, and let $Y=(\phi_t, t, w, s) \in \cDP_{st}$ be a degenerate pair of paints obtained by removing a disk from $(\phi_t, t)$.
Let $\sigma$ be the Koenigs map of the $\phi_t$, and let $\rho(z) = \frac{\sigma(z)}{z \sigma^\prime(z)}$.
Let $(\F^0, Y, \Omega, \nu, \ip{\,\cdot\,,\,\cdot\,},\theta)$ be the free fermion vertex operator superalgebra, and let $L_n$ be the positive energy represenation of $\Vir$ associated to $\nu$.
Let $\F$ be the Hilbert space completion of $\F^0$.
Then $e^{-t L(\rho)}$ and $T(\xi \otimes \eta) = Y(s^{L_0}\xi,w)e^{- tL(\rho)}\eta$ define bounded maps on $\F$ and $\F \otimes \F$, respectively. 
Moreover $E(X) = \C e^{-t L(\rho)} $ and $E(Y) = \C T$.
\end{Theorem}

Theorem \ref{thmBoundednessAndExistence} is the union of Proposition \ref{propAnnulusBoundednessAndExistence} and Theorem \ref{thmFermionPantsBoundedness}, both proven in this section.

\begin{Corollary}
Let $X \in \cDR$. Then $\dim E(X) = 1$.
\end{Corollary}
\begin{proof}
In the special case when $X \in \cDR_{st}$, this is just Theorem \ref{thmBoundednessAndExistence}.
The general case follows from the reparametrization formula Proposition \ref{propDegenerateOperatorReparametrization}.
\end{proof}

The first step in proving Theorem \ref{thmBoundednessAndExistence} is to get control of the maps $e^{- t L(\rho)}$.
The key ingredient is a `quantum energy inequality' of Fewster and Hollands \cite[Thm. 4.1]{FewsterHollands}, reformulated for Virasoro fields on the circle (as described in Remark 3 following \cite[Thm. 4.1]{FewsterHollands}, and in a forthcoming article of Carpi and Weiner \cite{CarpiWeinerLocal}).

\begin{Theorem}[\cite{FewsterHollands}, \cite{CarpiWeinerLocal}] \label{thmQEI}
Let $(L_n,V)$ be a unitary positive energy representation of the Virasoro algebra with central charge $c$, and let $\cH$ be the Hilbert space completion of $V$.
Let $f \in C^\infty(S^1, \R)$ be a function with $f \ge 0$, and let $L(f)$ be the associated smeared Virasoro field.
Then there is a number $K_f > 0$, depending only on $f$, such that 
$$
\ip{L(f)\xi,\xi} \ge -c \, K_f \norm{\xi}^2
$$
for all smooth vectors $\xi \in \cH^\infty$.
%We may take 
%$$
%K_f = \frac{1}{2\pi}\int_{-\pi}^\pi \left( \frac{d}{d \theta} \sqrt{f(e^{i \theta})}\right)^2 \, d\theta,
%$$
%where the expression $\frac{d}{d\theta} \sqrt{f(e^{i \theta})}$ is interpreted to be $0$ when $f(e^{i \theta}) = 0$.
\end{Theorem}

Using the estimate from Theorem \ref{thmQEI}, we may apply the Lumer-Phillips theorem to control the norm of $e^{-t L(\rho)}$.

\begin{Proposition}\label{propExponentiatedVirasoro}
Let $V$ be an inner product space equipped with a unitary positive energy representation of the Virasoro algebra $L_n$.
Assume that $V_\alpha := \ker L_0 - \alpha 1_V$ is finite-dimensional for all $\alpha  \in \R_{\ge 0}$.
Let $\rho:\interior{\D} \to \C$ be a holomorphic function which extends smoothly to $\D$.
Let $L(\rho) = \sum_{n \in \Z_{\ge 0}} \hat \rho_n L_n$, where $\hat \rho_n$ are the Fourier coefficients of $\rho|_{S^1}$.
Then for every $\xi \in V$ and $t \in \R$, the sum defining $e^{t L(\rho)} \xi$ converges to an element of $V$, and $e^{t L(\rho)}$ is invertible on $V$.
If $\Re \rho(z) \ge 0$ for all $z \in \D$ then for all $t \ge 0$, $e^{-t L(\rho)}$  extends to a bounded operator on the Hilbert space completion $\cH_V$ of $V$, and $\big(e^{-t L(\rho)}\big)_{t \ge 0}$ is a strongly continuous semigroup.
\end{Proposition}
\begin{proof}
For $\alpha \in \R_{\ge 0}$, let $W_\alpha = \bigoplus_{n \ge 0} V_{\alpha-n}$. 
Then $W_\alpha$ is finite-dimensional and invariant under $L(\rho)$.
Hence $L(\rho)$ induces a bounded operator on $W_\alpha$, and for $\xi \in W_\alpha$ the sum defining $e^{t L(\rho)}\xi$ converges.
Moreover, the operator $e^{t L(\rho)}$ on $W_\alpha$ is invertible.
Since $V = \bigcup_{\alpha \ge 0} W_\alpha$, $e^{t L(\rho)}\xi$ is well-defined for $\xi \in V$, and $e^{t L(\rho)}$ is invertible on $V$.

Now assume $\Re \rho(z) \ge 0$ for all $z \in \D$.
For each $\alpha, t \ge 0$, $e^{-t L(\rho)}$ is a bounded operator on $W_\alpha$.
We need to verify that the norm of the restriction to $W_\alpha$ is uniformly bounded as $\alpha$ varies.
By the Lumer-Phillips theorem \cite[Thm 3.3]{Goldstein85}, if $M \in \R$ has the property that
\begin{equation}\label{eqnQuasidissapative}
\Re \ip{L(\rho)\xi, \xi} \ge M \norm{\xi}^2
\end{equation}
for all $\xi \in W_\alpha$, then for $t \ge 0$ we have
$$
\|e^{ - t L(\rho)}|_{W_\alpha}\| \le e^{-Mt}.
$$
Thus to prove that $e^{-t L(\rho)}$ is  bounded on $\cH_V$, it suffices to show that there exists an $M$ such that \eqref{eqnQuasidissapative} holds for all $\xi \in V$.
Since $\ip{L(\rho)\xi,\xi} = \ip{\xi,L(\overline{\rho})\xi}$, we have
$$
\Re \ip{L(\rho)\xi,\xi} = \ip{L( \Re \rho)\xi, \xi},
$$
and since $\Re \rho(z) \ge 0$, the condition \eqref{eqnQuasidissapative} follows immediately from Theorem \ref{thmQEI} with $M = -c K_{\Re \rho}$.

It is clear that $e^{-t L(\rho)}$ is a semigroup on $V$, and that the function $t \mapsto e^{-t L(\rho)}\xi$ is continuous for $\xi \in V$ and $t \ge 0$. 
Since $\|e^{-t L(\rho)}\|$ is locally bounded, this implies that $e^{-t L(\rho)}$ is a strongly continuous semigroup.
\end{proof}

\begin{Remark}
Since the bound on the spectrum of $\Re L(\rho)$ from Theorem \ref{thmQEI} is independent of the smallest eigenvalue $h$ of $L_0$, we may extend Proposition \ref{propExponentiatedVirasoro} to arbitrary direct sums, allowing us to drop the assumption that the $L_0$ eigenspaces are finite dimensional.
We will not, however, use this fact.
\end{Remark}

Now given that the operators $e^{-tL(\rho)}$ are bounded, we return to the example of the free fermion, and compute commutation relations between $e^{-tL(\rho)}$ and the generators of the CAR algebra.

\begin{Lemma}\label{lemExponentiatedVirasoroFermion}
Let $(\phi_t)_{t \ge 0} \in \scG$ and let $\sigma$ be the Koenigs map associated to $\phi_t$.
Let $\rho(z) = \frac{\sigma(z)}{z \sigma^\prime(z)}$.
Let $H_{\ge k} = \overline{\Span \{z^n : n \ge k\}} \subset H$ and let $H_{\le k} = \overline{\Span \{z^n : n \le k\}}$.
Then for $f \in H_{\ge k}$, $g \in H_{\le k}$, and $t \ge 0$, we have
\begin{equation}\label{eqnExponentiatedVirasoroImplementer}
a(f)e^{-t L(\rho)} = e^{-t L(\rho)} a(W_{\phi_t} f), \qquad a(g)^*e^{-tL(\rho)} = e^{-t L(\rho)} a(\tilde W_{\phi_t}g)^*.
\end{equation}
\end{Lemma}
\begin{proof}
Our argument is similar to \cite[Exp. Thm. \S8]{Wa98}. 
We begin with the first equality of \eqref{eqnExponentiatedVirasoroImplementer}, namely that
$$
a(f)e^{-t L(\rho)}\xi = e^{-t L(\rho)} a(W_{\phi_t} f)\xi
$$
for all $\xi \in \F$.
Fix $\xi \in \F$.

By Proposition \ref{propSemigroupSummary}, we have $\Re \rho(z) \ge 0$ for all $z \in \D$, and thus by Proposition \ref{propExponentiatedVirasoro}, for $t \ge 0$ and $\eta \in \F^0$, the sum defining $e^{-t L(\rho)}\eta$ converges, and the resulting operators are bounded and form a strongly continuous semigroup.

For $n \in \tfrac12\Z_{\ge 0}$, let $\F_n$ be the eigenspace of $L_0$ with eigenvalue $n$.
We may assume without loss of generality that $\xi \in \F_n$ for some $n$.
Since $\norm{a(f)} = \norm{f}$ and $W_{\phi_t}$ is bounded on $H_{\ge k}$, we may assume without loss of generality that $f$ is a Laurent polynomial.
Let $M \in \Z_{>0}$ be a number with $M > n$ and $f \in W:= \Span\{z^{-M}, z^{-M+1}, \ldots, z^{M-1}, z^M\}$.

Let $\F_{\le k} = \bigoplus_{j =0}^{2k} \F_{j/2}$.
Then $\F_{\le n}$ and $\F_{\le n+M}$ are finite-dimensional, and as in the proof of Proposition \ref{propExponentiatedVirasoro}, $L(\rho)$ is a bounded operator on both spaces.
Hence $e^{-t L(\rho)}$ is defined on both spaces for all $t \in \R$, and yields a one-parameter group.

Now let us think of $W_{\phi_t}$ as an operator on $H_{\ge -M}$, with $H_{> M}$ an invariant subspace.
Hence if $q$ is the projection of $H_{\ge -M}$ onto $W = H_{\ge -M} \ominus H_{> M}$, we have $q W_{\phi_t} = q W_{\phi_t} q$.
Thus $q W_{\phi_t}$ is a strongly continuous one-parameter semigroup on the finite dimensional space $W$, and so there exists an $X \in \cB(W)$ such that $q W_{\phi_t} = e^{t X}$.

Observe that since $e^{t L(\rho)}\xi \in \F_{\le n}$, $W_{\phi_t}f \in W$, and $M > n$, we have 
$$
a(W_{\phi_t}f)e^{t L(\rho)}\xi = a(q W_{\phi_t}f)e^{tL(\rho)}\xi = a(e^{tX}f)e^{t L(\rho)}\xi.
$$
Hence the function $\R_{\ge 0} \to \F^0_{\le n+M}$ given by $t \mapsto e^{-t L(\rho)}a(W_{\phi_t}f)e^{t L(\rho)}\xi$ can be smoothly extended to all of $\R$, and more importantly its derivative is given by
\begin{equation}\label{eqAnnulusImplementationDerivative}
\frac{d}{dt} e^{-t L(\rho)}a(W_{\phi_t}f)e^{t L(\rho)}\xi = e^{-t L(\rho)}\Big(a(Xe^{tX} f) - [L(\rho), a(e^{tX}f)]\Big)e^{t L(\rho)}\xi.
\end{equation}

By Proposition \ref{propSemigroupSummary}, we have $\phi_t(z) = \sigma^{-1}(e^{-t}\sigma(z))$ for all $z \in \D$ and $t \ge 0$.
From the formula, we can see that $(t,z) \mapsto \phi_t(z)$ extends to a smooth function in a neighborhood of $\R_{\ge 0} \times \tfrac12 S^1$.
Hence for $z \in \tfrac12 S^1$ and $g \in W$, we can compute
$$
\lim_{t \downarrow 0} \frac{(W_{\phi_t}g)(z)-g(z)}{t} = \left. \frac{d}{dt} (W_{\phi_t}g)(z)\right|_{t=0} = -(z\rho(z)g^\prime(z) + \tfrac12(z\rho)^\prime(z)g(z))
$$
with uniform convergence in $z$ on $\tfrac12 S^1$.
Hence for all $k \in \Z$ we have
\begin{align*}
\lim_{t \downarrow 0} \ip{\frac{W_{\phi_t}g-g}{t},z^k} &= \lim_{t \downarrow 0} \frac{1}{2\pi i}\oint_{\tfrac12 S^1} \frac{(W_{\phi_t}g)(z)-g(z)}{t} z^{-k-1}dz\\  
&= -\frac{1}{2\pi i} \oint_{\tfrac12 S^1} (z\rho(z)g^\prime(z) + \tfrac12(z\rho)^\prime(z)g(z)) z^{-k-1} dz\\
&= -\ip{z\rho g^\prime + \tfrac12 (z \rho)^\prime g,z^k}.
\end{align*}
Since $q W_{\phi_t} = e^{t X}$ on the finite-dimensional space $\im q = W = \Span \{z^{-M}, \ldots, z^M\}$, this implies that
$$
Xg = -q(z\rho g^\prime + \tfrac12 (z \rho)^\prime g).
$$

By Proposition \ref{propVirFermionCommRels}, we have
\begin{align*}
[L(\rho), a(e^{tX}f)]e^{t L(\rho)}\xi = -a((z\rho g^\prime + \tfrac12 (z \rho)^\prime g) e^{t X}f)e^{t L(\rho)}\xi = a(Xe^{t X}f)e^{t L(\rho)}\xi,
\end{align*}
where in the last equality we use that for $h \in H_{\ge -M} \ominus W$, we have $a(h)e^{t L(\rho)}\xi = 0$ since $e^{t L(\rho)}\xi \in \F_{\le n}^0$.
Substituting this result into \eqref{eqAnnulusImplementationDerivative}, we see that $e^{-t L(\rho)} a(W_{\phi_t}f ) e^{t L(\rho)}$ is independent of $t$, and evaluating at $t=0$ we see that
$$
e^{-t L(\rho)} a(W_{\phi_t}f ) e^{t L(\rho)}\xi = a(f)\xi
$$
for all $\xi \in \F_{\le n}$.
Hence $a(f)e^{-t L(\rho)}\xi = e^{-t L(\rho)} a(W_{\phi_t} f)\xi$ for all $\xi \in \F_{\le n}$, which was to be shown.

We now turn to showing that
\begin{equation}\label{eqnStarVirasoroRelMidProof}
a(g)^*e^{-tL(\rho)}\xi = e^{-t L(\rho)} a(\tilde W_{\phi_t}g)^*\xi
\end{equation}
for all $g \in H_{\le k}$ and all $\xi \in \F$.
As above, it suffices to consider $\xi \in \F_{\le n}$ and $g \in W := \{z^{-M-1}, \ldots, z^M\}$, where we choose $M > n$.
Recall that $\tilde W_{\phi_t} = c W_{\phi_t} c$, where $cf = \overline{zf}$.
Note that we have slightly adjusted the definition of $W$ in this case so that $cW = W$.

Using the same ideas as above, let $\tilde q$ be the projection of $H_{\le M}$ onto $W$, so that $\tilde q \tilde W_{\phi_t}$ is a continuous semigroup on $W$. 
We have $\tilde q \tilde W_{\phi_t} = c e^{t X} c = e^{t \tilde X}$, where $\tilde X = c Xc$.
In fact, it is straightforward to compute $\tilde X$ explicitly, and we get
$$
\tilde Xg = \tilde q(z \overline{\rho} g^\prime + \tfrac12 (z \overline{\rho})^\prime g).
$$
Differentiating as above, we get
\begin{equation*}
\frac{d}{dt} e^{-t L(\rho)}a(\tilde W_{\phi_t}f)^*e^{t L(\rho)}\xi = e^{-t L(\rho)}\Big(a(\tilde Xe^{t\tilde X} f)^* - [L(\rho), a(e^{t\tilde X}f)^*]\Big)e^{t L(\rho)}\xi,
\end{equation*}
which vanishes by Proposition \ref{propVirFermionCommRels}.
This establishes \eqref{eqnStarVirasoroRelMidProof}, and completes the proof of the lemma.
\end{proof}

So far, we have collected enough results to establish Theorem \ref{thmBoundednessAndExistence} for $X \in \cDA_{st}$.

\begin{Proposition}\label{propAnnulusBoundednessAndExistence}
Let $X=(\phi_t, t) \in \cDA_{st}$, and let $\rho$ be as in Theorem \ref{thmBoundednessAndExistence}.
Then $E(X) = \C e^{-t L(\rho)}$.
\end{Proposition}
\begin{proof}
By Proposition \ref{propExponentiatedVirasoro}, $e^{-tL(\rho)} \in \cB(\F)$, and $e^{-t L(\rho)}$ is clearly even.
By Runge's theorem, we have $H^2(X) = \overline{ \Span \{(z^k, W_{\phi_t} z^k) : k \in \Z \}}$, and so by Lemma \ref{lemExponentiatedVirasoroFermion} $e^{-t L(\rho)}$ satisfies
$$
a(f^1)e^{-tL(\rho)} = e^{-t L(\rho)}a(f^0)
$$
for all $(f^1, f^0) \in H^2(X)$.

Now let $c:L^2(S^1) \to L^2(S^1)$ be the antilinear unitary $cf = \overline{zf}$.
By definition, we have
$$
\overline{M_{\pm z} H^2(X)} = \overline{ \Span \{ (c z^k, -c W_{\phi_t} z^k) : k \in \Z \}} = \overline{ \Span \{ (z^k, -\tilde W_{\phi_t} z^k) : k \in \Z\}}.
$$
Thus is follows directly from Lemma \ref{lemExponentiatedVirasoroFermion} that
$$
a(g^1)^*e^{-t L(\rho)} = e^{-t L(\rho)} a(g^0)^*
$$
for every $(g^1, g^0) \in \overline{z H^2(X)}$.
Hence $e^{- tL(\rho)} \in E(X)$.
But this finishes the proof, as we established that $\dim E(X) \le 1$ in Proposition \ref{propDimAtMostOne}.
\end{proof}

We now switch from studying degenerate annuli to studying degenerate pairs of pants 
$$
X = (\phi_t,t,w,s) \in \cDP_{st}.
$$
We wish to show that $E(X)$ is spanned by the map $T:\F \otimes \F \to \F$ given by
$$
T(\xi \otimes \eta) = Y(s^{L_0}\xi,w)e^{-t L(\rho)}\eta,
$$
which is defined on $\F^0 \otimes \F^0$ by Proposition \ref{propExponentiatedVirasoro} and Proposition \ref{propVertexOperatorDenselyDefined}.

The strategy for showing this is somewhat indirect, and so we first give a short summary.
Let $R > 1$, and let $X_R$ be the non-degenerate extension of $X$ (Definition \ref{defNondegenerateExtension}).
Then by the gluing property of the (non-degenerate) Segal CFT, there is an element $T_R \in E(X_R)$ satisfying $T_R(\Omega \otimes \Omega) = \Omega$.
First, will verify that $T_R(\xi \otimes \eta) = R^{-L_0}T(\xi \otimes \eta)$ for $\xi,\eta \in \F^0$, and thus that $T_R \to T$ as $R \downarrow 1$, pointwise on the algebraic tensor product $\F^0 \otimes_{alg} \F^0$.
Next, we will show that $T_R^*$ converges pointwise on $\F^0$ to a densely defined map which we call $S$.
We will see that $S$ is an example of what we call an \emph{implementing operator}.
That is, there is a vector $\hOmega \in \F \otimes \F$ and a map $r \in \cB(H,H \oplus H)$ such that
$$
Sa(\xi_1)^* \cdots a(\xi_n)^*a(\eta_1) \cdots a(\eta_m)^*\Omega = a(r \xi_1)^* \cdots a(r \xi_n)^* a(r \eta_1) \cdots a(r \eta_m)\hOmega
$$
whenever $\xi_i \in pH = H^2(\D)$ and $\eta_j \in (1-p)H$.
Here, we have identified $\F \otimes \F \cong \F_{H \oplus H, p \oplus p}$ as in Proposition \ref{propFockSumToTensor}.
In Section \ref{secImplementingOperators}, we develop tools for proving boundedness of implementing operators, and these will tell us that $S$ is bounded.
It then follows that $T_R^* = SR^{-L_0}$, and thus $\norm{T_R^*} \le \norm{S}$.
Hence $\norm{T_R}$ remains bounded as $R \downarrow 1$, from which we can conclude that $T$ is bounded and that $T_R \to T$ in the strong operator topology.
It is then easy to verify that $T$ satisfies the necessary commutation relations to lie in $E(X)$

Our first task is to establish a formula for $T_R$ in terms of vertex operators.
We will need a version of the Borcherds commutator formula for free fermion vertex operators evaluated at a complex number.
\begin{Proposition}\label{propConcreteCommutator}
Suppose $w \in \C$ with $0 < \abs{w} < 1$, and that $s > 0$ satisfies $w + s\D \subset \interior{\D} \setminus \{0\}$.
Then for all $\xi,\eta \in \F^0$ and every $n \in \Z$,
\begin{equation}\label{eqnConcreteCommutator}
a(z^n)Y(s^{L_0}\xi,w)\eta = Y(s^{L_0} a(s^{1/2}(sz+w)^n)\xi,w)\eta + (-1)^{p(\xi)}Y(s^{L_0}\xi,w)a(z^n)\eta 
\end{equation}
and
\begin{equation}\label{eqnConcreteCommutatorStar}
a(z^{-n-1})^* Y(s^{L_0}\xi,w)\eta =  Y(s^{L_0}a(s^{1/2}z^{-1}(sz^{-1}+\overline{w})^n)^*\xi,w)\eta+(-1)^{p(\xi)} Y(s^{L_0}\xi,w)a(z^{-n-1})^*\eta
\end{equation}
where the equations are understood as holding when $\xi$ is homogeneous, and extended linearly otherwise.
\end{Proposition}
\begin{proof}
Observe that all of the terms in \eqref{eqnConcreteCommutator} and \eqref{eqnConcreteCommutatorStar} are well-defined elements of $\F$, with the defining sums converging absolutely,  by Proposition \ref{propVertexOperatorDenselyDefined} and the fact that $a((sz+w)^n)$ and $a(z^{-1}(sz^{-1} + \overline{w}))^*$ map $\F^0$ into itself.
Assume without loss of generality that $\xi$ and $\eta$ are eigenvectors for $L_0$, and that $\eta^\prime $ is as well.
Then by the Borcherds commutator formula (Theorem \ref{thmBorcherds}), we have an identity of formal series
\begin{equation}\label{eqnFormalCommutator}
a(z^n)Y(s^{L_0}\xi,x)\eta = (-1)^{p(\xi)}Y(s^{L_0}\xi,x)a(z^n)\eta + \sum_{k \ge 0} \binom{n}{k} Y(a(z^k)s^{L_0}\xi,x)x^{n-k}\eta,
\end{equation}
where the sum in $k$ is finite.
But the three terms of \eqref{eqnFormalCommutator} all give absolutely convergent series when evaluated at $x=w$, and so we have an equality of elements of $\F$:
\begin{align*}
a(z^n)Y(s^{L_0}\xi,w)\eta &= (-1)^{p(\xi)}Y(s^{L_0}\xi,w)a(z^n)\eta + \sum_{k \ge 0} \binom{n}{k} Y(a(z^k)s^{L_0}\xi,w)w^{n-k}\eta\\
&= (-1)^{p(\xi)}Y(s^{L_0}\xi,w)a(z^n)\eta + \sum_{k \ge 0} w^{n}s^{1/2} \binom{n}{k} Y(s^{L_0}a((sz/w)^k)\xi,w)\eta\\
&= (-1)^{p(\xi)}Y(s^{L_0}\xi,w)a(z^n)\eta + Y(s^{L_0}a(s^{1/2}(sz+w)^n)\xi, w)\eta,
\end{align*}
where we used that $s < \abs{w}$ by assumption, and the finiteness of the sum in $k$.

The proof of relation \eqref{eqnConcreteCommutatorStar} is similar.
By the Borcherds commutator formula, we have an identity of formal series
$$
a(z^{-n-1})^*Y(s^{L_0}\xi,x)\eta = (-1)^{p(\xi)}Y(s^{L_0}\xi,x)a(z^{-n-1})^*\eta + \sum_{k \ge 0} \binom{n}{k} Y(a(z^{-k-1})^*s^{L_0}\xi,x)x^{n-k}\eta.
$$
Evaluating at $x=w$ and arguing as above, we get
\begin{align*}
a(z^{-n-1})^*Y(s^{L_0}\xi,w)\eta &= (-1)^{p(\xi)}Y(s^{L_0}\xi,w)a(z^{-n-1})^*\eta +  \sum_{k \ge 0} \binom{n}{k} Y(a(z^{-k-1})^*s^{L_0}\xi,w)w^{n-k}\eta\\
&= (-1)^{p(\xi)}Y(s^{L_0}\xi,w)a(z^{-n-1})^*\eta +  \sum_{k \ge 0} w^n \binom{n}{k}s^{1/2} Y(s^{L_0} a(z^{-1} (s/(\overline{w}z))^k)^*\xi,w)\eta\\
&= (-1)^{p(\xi)}Y(s^{L_0}\xi,w)a(z^{-n-1})^*\eta +  Y(s^{L_0}a(s^{1/2}z^{-1} (sz^{-1}+\overline{w}))^*\xi,w)\eta.
\end{align*}
\end{proof}

We can now establish the desired formula for $T_R$.
\begin{Proposition}\label{propLimitIsRegularizedVertexOperator}
Let $X = (\phi_t, t, w, s) \in \cDP_{st}$, let $R > 1$ and let $X_R$ be the non-degenerate extension of $X$.
Let $T_R \in E(X_R)$ be the element with $T_R(\Omega \otimes \Omega) = \Omega$.
For $\xi,\eta \in \F^0$, let $T(\xi \otimes \eta) = Y(s^{L_0}\xi,w)e^{-t L(\rho)}\eta$, where $\rho$ is as in Theorem \ref{thmBoundednessAndExistence}.
Then $T_R(\xi \otimes \eta) = R^{-L_0} T(\xi \otimes \eta)$, and $T_R(\xi \otimes \eta) \to T(\xi \otimes \eta)$ as $R \downarrow 1$, for all $\xi,\eta \in \F^0$.
\end{Proposition}
\begin{proof}
By Proposition \ref{propExponentiatedVirasoro}, $e^{-tL(\rho)}$ maps $\F^0$ bijectively onto itself.
Vectors of the form
$$
a(z^{m_1}) \cdots a(z^{m_p}) a(z^{n_q})^* \cdots a(z^{n_1})^*\Omega
$$ 
with $m_j < 0$ and $n_i \ge 0$ form a spanning set for $\F^0$, and since
$$
e^{-t L(\rho)} a(\tilde W_{\phi_t}z^{m_1}) \cdots a(\tilde W_{\phi_t}z^{m_p}) a(W_{\phi_t} z^{n_q}) \cdots a(W_{\phi_t}z^{n_1})^*\Omega = a(z^{m_1}) \cdots a(z^{m_p}) a(z^{n_q})^* \cdots a(z^{n_1})^*\Omega
$$
by Lemma \ref{lemExponentiatedVirasoroFermion},
vectors of the form
\begin{equation}\label{eqnDegeneratePantsLimitVectorForm}
\eta = a(\tilde W_{\phi_t}z^{m_1}) \cdots a(\tilde W_{\phi_t}z^{m_p}) a(W_{\phi_t} z^{n_q})^* \cdots a(W_{\phi_t}z^{n_1})^*\Omega
\end{equation}
span $\F^0$.
Thus it suffices to verify that 
\begin{equation}\label{eqnVertexOperatorLimit}
T_R(\xi \otimes \eta) =R^{-L_0}Y(s^{L_0}\xi,w)e^{-tL(\rho)}\eta
\end{equation}
when $\eta$ is of the form \eqref{eqnDegeneratePantsLimitVectorForm}.
We also assume without loss of generality that $\xi$ is homogeneous.
We now proceed by induction on $p$ and $q$.

When $\eta = \Omega$, by the gluing proprety of the (non-degenerate) free fermion Segal CFT we have $T_R(\xi \otimes \Omega) = \alpha R^{-L_0} Y(s^{L_0}\xi, w)\Omega$ for some $\alpha \in \C^\times$. 
But we normalized $T_R$ so that $T_R(\Omega \otimes \Omega) = \Omega$, and thus $\alpha = 1$. 
Hence \eqref{eqnVertexOperatorLimit} holds when $\eta = \Omega$.

Now assume that \eqref{eqnVertexOperatorLimit} holds for all $\xi \in \F^0$ and for a vector $\eta$ of the form \eqref{eqnDegeneratePantsLimitVectorForm}, and we will show that it holds for $\eta^\prime = a(W_{\phi_t} z^n)\eta$ and $\eta^{\prime\prime} = a(\tilde W_{\phi_t} z^{-n-1})^*\eta$ for all $n \in \Z$.

From the holomorphic function $z^n \in \cO(X_R)$, we have
$$
(R^{n+1/2} z^n, s^{\tfrac12}(sz+w)^n, W_{\phi_t} z^n) \in H^2(X_R),
$$
where we have ordered the boundary components $S^1$, then $s S^1 + w$, then $\phi_t(S^1)$.
Hence by the definition of the operators $E(X_R)$ we have
\begin{equation}\label{eqnVertexOpLimitHardyEqn}
R^{n+1/2} a(z^n) T_R = T_R(a(s^{\tfrac12}(sz+w)^n) \otimes 1) + T_R(\Gamma \otimes a(W_{\phi_t}z^n)),
\end{equation}
where $\Gamma$ is the grading.
Hence
\begin{equation}\label{eqnVertexOpLimitHardyEqn2}
R^{n+1/2} a(z^n) T_R(\xi \otimes \eta) = T_R(a(s^{\tfrac12}(sz+w)^n) \xi \otimes \eta) + (-1)^{p(\xi)}T_R(\xi \otimes \eta^\prime).
\end{equation}

On the other hand, by the inductive hypothesis, Proposition \ref{propConcreteCommutator}, and Lemma \ref{lemExponentiatedVirasoroFermion} we have

\begin{align}\label{eqnCommutatorApplication}
R^{n+1/2}a(z^n)T_R(\xi \otimes \eta) &=R^{n+1/2}a(z^n)R^{-L_0}Y(s^{L_0}\xi,w)e^{-t L(\rho)}\eta\nonumber\\
&= R^{-L_0}a(z^n)Y(s^{L_0}\xi,w)e^{-t L(\rho)}\eta \nonumber\\
&= R^{-L_0}Y(s^{L_0}a(s^{1/2}(sz+w)^n)\xi,w)\eta  + (-1)^{p(\xi)}R^{-L_0}Y(s^{L_0}\xi,w)e^{-t L(\rho)}\eta^\prime \nonumber\\
&= T_R(a(s^{1/2}(sz+w)^n)\xi \otimes \eta) + (-1)^{p(\xi)}R^{-L_0}Y(s^{L_0}\xi,w)e^{-t L(\rho)}\eta^\prime.
\end{align}
Combining \eqref{eqnVertexOpLimitHardyEqn2} and \eqref{eqnCommutatorApplication}, we get 
$$
T_R(\xi \otimes \eta^\prime) = R^{-L_0}Y(s^{L_0}\xi,w)e^{-t L(\rho)} \eta^\prime,
$$
as desired.

Establishing \eqref{eqnVertexOperatorLimit} for $\eta^{\prime\prime} = a(\tilde W_{\phi_t} z^{-n-1})^*\eta$ is similar.
By \cite[Thm. 6.1]{Ten16}, we have $H^2(X_R)^\perp = \overline{M_{\pm z} z H^2(X)}$, where $M_{\pm}$ is multiplication by $z$ on copies of $L^2(S^1)$ indexed by outgoing boundary componenets, and by $-z$ on copies of $L^2(S^1)$ indexed by incoming boundary components. 
Hence
$$
(R^{n+1/2} z^{-n-1}, -s^{1/2} z^{-1}(sz^{-1} + \overline{w})^n, -\overline{ z W_{\phi_t} z^n}) \in H^2(X_R)^\perp.
$$
By definition, $\tilde W_{\phi_t} \overline{zf} = \overline{z W_{\phi_t}  f}$ when $f \in L^2(S^1)_{\ge k}$, so we have 
$$
(R^{n+1/2} z^{-n-1}, -s^{1/2} z^{-1}(sz^{-1} + \overline{w})^n, -\tilde W_{\phi_t} z^{-n-1}) \in H^2(X_R)^\perp.
$$
Hence by the definition of $E(X_R)$ we have
\begin{equation}\label{eqnVertexOpLimitHardyEqnStar}
R^{n+1/2} a(z^{-n-1})^* T_R =  T_R(a(s^{1/2}z^{-1}(sz^{-1} + \overline{w})^n)^* \otimes 1) + T_R(\Gamma \otimes a(\tilde W_{\phi_t} z^{-n-1})^*),
\end{equation}
and thus
\begin{equation}\label{eqnVertexOpLimitHardyEqnStar2}
R^{n+1/2} a(z^{-n-1})^* T_R(\xi \otimes \eta) =  T_R(a(s^{1/2}z^{-1}(sz^{-1} + \overline{w})^n)^*\xi \otimes \eta) + (-1)^{p(\xi)} T_R(\xi \otimes \eta^{\prime\prime}).
\end{equation}
Expanding $(sz^{-1} + \overline{w})^n$ in the domain $\abs{z}^{-1} < \abs{w/s}$, we see that 
$$
a(z^{-1}(sz^{-1} + \overline{w})^n)^*\xi \in \F^0.
$$
As before, we may apply the inductive hypothesis, Proposition \ref{propConcreteCommutator}, and Lemma \ref{lemExponentiatedVirasoroFermion} to establish
$$
R^{n+1/2} a(z^{-n-1})^* T_R(\xi \otimes \eta)  = T_R(a(s^{1/2}z^{-1}(sz^{-1} + \overline{w})^n)^*\xi + (-1)^{p(\xi)} R^{-L_0} Y(s^{L_0}\xi,w)e^{-t L(\rho)}\eta^{\prime\prime},
$$
from which the desired conclusion follows.

Given that $T_R(\xi \otimes \eta) = R^{-L_0} T(\xi \otimes \eta)$, and that $T(\xi \otimes \eta) \in \F$ when $\xi,\eta \in \F^0$ by Proposition \ref{propVertexOperatorDenselyDefined}, it follows immediately that $T_R(\xi \otimes \eta) \to T(\xi \otimes \eta)$ for such $\xi,\eta$.
\end{proof}

Next, we want to understand the limit $\lim_{R \downarrow 1} T_R^*$.

\begin{Proposition}\label{propAdjointsConverge}
Let $X = (\phi_t, t, w, s) \in \cDP_{st}$, let $R > 1$ and let $X_R$ be the non-degenerate extension of $X$.
Let $T_R \in E(X_R)$ be the element with $T(\Omega \otimes \Omega) = \Omega$. 

Then the limit $S\xi := \lim_{R \downarrow 1} T_R^*\xi$ converges for all $\xi \in \F^0$.
The limit operator $S$ satisfies
\begin{equation}\label{eqnAdjointLimitCommRelsStar}
S a(z^n)^*\xi =  (s^{1/2}a((sz + w)^n)^* \otimes 1  + \Gamma \otimes a(W_{\phi_t} z^n)^*)S\xi
\end{equation}
and
\begin{equation}\label{eqnAdjointLimitCommRelsNoStar}
S a(z^{-n-1}) \xi = (s^{1/2} a(z^{-1}(sz^{-1} + \overline{w})^n) \otimes 1 + \Gamma \otimes a( \tilde W_{\phi_t} z^{-n-1}))S\xi
\end{equation}
for all $\xi \in \F^0$ and all $n \in \Z$, where $\Gamma$ is the grading operator.
\end{Proposition}
\begin{proof}
It suffices to establish the result with
$$
\xi = a(z^{m_1}) \cdots a(z^{m_p}) a(z^{n_1})^* \cdots a(z^{n_q})^*\Omega.
$$
We will proceed inductively, first considering when $\xi = \Omega$.

If $R_1 > R_2 > 1$, then we have $T_{R_1} = (R_1/R_2)^{-L_0} T_{R_2}$. 
Hence $T_{R_1}^* = T_{R_2}^* (R_1/R_2)^{-L_0}$, and 
$$
T_{R_1}^*\Omega = T_{R_2}^* (R_1/R_2)^{-L_0}\Omega = T_{R_2}^*\Omega.
$$
Hence $\lim_{R \downarrow 1} T_R^* \Omega$ converges.

We now assume that $\lim_{R \downarrow 1} T_R^* \xi$ converges, and show that the same holds for $a(z^{-n-1})\xi$ and $a(z^n)^*\xi$.
Indeed, applying the adjoint of the commutation relation from \eqref{eqnVertexOpLimitHardyEqn} one has
\begin{equation}\label{eqnAdjointLimitCommRelsStarBody}
T_R^* a(z^n)^*\xi = R^{-n-1/2} (s^{1/2} a((sz + w)^n)^* \otimes 1  + \Gamma \otimes a(W_{\phi_t} z^n)^*)T_R^*\xi.
\end{equation}
It follows that $\lim_{R \downarrow 1} T_R^* a(z^n)^*\xi$ converges, and that \eqref{eqnAdjointLimitCommRelsStar} holds for $\xi$.

Similarly, applying the adjoint of the commutation relation from \eqref{eqnVertexOpLimitHardyEqnStar} one has
\begin{equation}\label{eqnAdjointLimitCommRelsNoStarBody}
T_R^* a(z^{-n-1}) \xi = R^{-n-1/2}(s^{1/2} a(z^{-1}(sz^{-1} + \overline{w})^n) \otimes 1 + \Gamma \otimes a( \tilde W_{\phi_t} z^{-n-1}))T_R^*\xi,
\end{equation}
from which we see that $\lim_{R \downarrow 1} T_R^* a(z^{-n-1})\xi$ converges, and \eqref{eqnAdjointLimitCommRelsNoStar} holds for $\xi$.
\end{proof}

The commutation relations given in Proposition \ref{propAdjointsConverge} almost characterize the densely defined limit operator $S=\lim_{R \downarrow 1} T_R^*$. 
Indeed, since the operators $a(z^n)$ and $a(z^{-n-1})^*$ act cyclically on the vacuum vector $\Omega$, the limit operator is specified by \eqref{eqnAdjointLimitCommRelsStar} and \eqref{eqnAdjointLimitCommRelsNoStar} once we have identified the vector $S\Omega$.
Similarly, the commutation relations from Lemma \ref{lemExponentiatedVirasoroFermion} will allow us to describe $(e^{-t L(\rho)})^*$ once we better understand $(e^{-t L(\rho)})^*\Omega$.

\begin{Proposition}\label{propAdjointsOnVacuum}
Let $(\phi_t,t,w,s) = X \in \cDP_{st}$, let $R >1$, and let $X_R$ be the non-degenerate extension of $X$.
Let $\rho$ be as in Theorem \ref{thmBoundednessAndExistence}.
Then there exist $(\psi,\gamma) \in \Diff_+^{NS}(S^1)$ and $\alpha \in \C^\times$ such that $(e^{-t L(\rho)})^* \Omega = \alpha U_{NS}(\psi,\gamma)\Omega$.
If $T_R \in E(X_R)$ is the element satisfying $T_R(\Omega \otimes \Omega) = \Omega$, then there exists a non-degenerate spin Riemann surface $Y \in \cR$, with no incoming boundary components and two outgoing boundary componenets, such that $S\Omega = \lim_{R \downarrow 1} T_R^* \Omega \in E(Y)$.
Moreover, $S \Omega \ne 0$.
\end{Proposition}
\begin{proof}
We will make free use of the properties of the free fermion Segal CFT (for non-degenerate surfaces) given in Theorem \ref{thmFermionSegalCFT}.

Let $Z_R \in \cR$ be the non-degenerate spin Riemann surface obtained from $X_R$ by filling in the disk centered at $w$.
By the gluing property of the free fermion Segal CFT and the formula for $T_R$ in Proposition \ref{propLimitIsRegularizedVertexOperator}, we have $R^{-L_0}e^{-tL(\rho)} \in E(Z_R)$.
Hence by unitarity $(e^{-tL(\rho)})^*R^{-L_0} \in E(\overline{Z_R})$.
Let $Z$ be the spin Riemann surface, with no incoming boundary and one outgoing boundary component, obtained by gluing a standard disk to the input of $\overline{Z_R}$.
Then $(e^{-t L(\rho)})^*R^{-L_0}\Omega = (e^{-t L(\rho)})^*\Omega \in E(\overline{Z_R})$, and $(e^{-t L(\rho)})^*\Omega \ne 0$ since non-zero elements of $E(\overline{Z_R})$ are injective by the nondegeneracy property of the CFT.
By the smooth Riemann mapping theorem, $\overline{Z_R}$ is spin equivalent to the standard unit disk with some boundary parametrization, and thus by the reparametrization property of the CFT we have $(e^{-tL(\rho)})^*\Omega = \alpha U_{NS}(\psi,\gamma)\Omega$ for some spin diffeomorphism $(\psi,\gamma)$ and some $\alpha \in \C^\times$.

We now handle $S\Omega$.
As we saw in the proof of Proposition \ref{propAdjointsConverge}, $T_R^*\Omega$ is independent of $R$, so we fix $R > 1$ and show that $T_R^*\Omega \in E(Y)$ for some $Y$.
Since $T_R^*$ is injective by the unitarity and non-degeneracy of the CFT, we have $T_R^* \Omega \ne 0$.
The vector $T_R^*\Omega$ has dual vector $\lambda \in (\F \otimes \F)^*$ given by
$$
\lambda(\xi \otimes \eta) = \ip{T_R(\xi \otimes \eta), \Omega}.
$$
The dual vacuum vector $\ip{ \; \cdot \; , \Omega} \in \F^*$ lies in $E(\C \cup \{\infty\} \setminus R\interior{\D})$, and so by the gluing property of the CFT, $\lambda \in E(Y)$, where $Y$ is obtained by gluing $\C \cup \{\infty\} \setminus R\interior{\D}$ onto $X_R$.
Hence by the unitarity property of the CFT, $T_R^*\Omega = \lambda^* \in E(\overline{Y})$.
\end{proof}

We will now show that $\xi \otimes \eta \mapsto Y(s^{L_0}\xi,w)e^{-tL(\rho)}\eta$ defines a bounded operator.
We will require the terminology and results of Section \ref{secImplementingOperators}, which we now summarize.

\begin{Definition*}[Definition \ref{defImplementer}]
Let $H$ and $K$ be separable infinite dimensional Hilbert spaces, and let $p$ and $q$ be projections on $H$ and $K$, respectively, with $pH$ and $(1-p)H$ both infinite dimensional.
Let $\{\xi_i\}_{i \in \Z}$ be an orthonormal basis for $H$ with $\xi_i \in pH$ when $i \ge 0$ and $\xi_i \in (1-p)H$ when $i < 0$.
Such an orthonormal basis is called \emph{compatible with $p$.}
Let $r \in \cB(H,K)$ and let $\hOmega \in \F_{K,q}$.
Then we have an orthonormal basis $a(\xi_I)^*a(\xi_J)\Omega_p$ for $\F_{H,p}$ (see notation \eqref{eqnFermionProductNotationEarly}) indexed by finite subsets $I \subset \Z_{\ge 0}$ and $J \subset \Z_{<0}$, 
and the densely defined map $R:\F_{H,p} \to \F_{K,q}$ given by
$$
Ra(\xi_I)^*a(\xi_J)\Omega = a(r \xi_I)^*a(r\xi_J)\hOmega
$$
is called the \emph{implementing operator} associated to $(r, \hOmega)$.
\end{Definition*}
The results of Propositions \ref{propAdjointsConverge} and \ref{propAdjointsOnVacuum} show that $S$ is an implementing operator, and so we may use the following result to prove that $S$, and consequently $Y(s^{L_0}\xi,w)e^{-tL(\rho)}\eta$, are bounded.

\begin{Theorem*}[Theorem \ref{thmAdmissibleBoundedness}]
Let $H$ and $K$ be separable Hilbert spaces, and let $p$ and $q$ be projections on $H$ and $K$, respectively, with $pH$ and $(1-p)H$ infinite dimensional. 
Let $\{\xi_i\}_{i \in \Z}$ be a basis compatible with $p$. 
Let $r \in \cB(H,K)$, and assume that $qr(1-p)$ is trace class.
Let $q^\prime$ be a projection on $K$ with $q^\prime - q$ trace class, and let $\hOmega_{q^\prime}$ be a non-zero vector satisfying $a(f)\hOmega_{q^\prime} = a(g)^*\hOmega_{q^\prime} = 0$ for all $f \in q^\prime K$ and all $g \in (1-q^\prime)K$.
Then the implementing operator associated to $(r,\hOmega_{q^\prime})$ is bounded if and only if $E(r):=qrp + (1-q)r(1-p)$ can be written as the sum $E(r) = a + x$ with $a,x \in \cB(H,K)$, $\norm{a} \le 1$ and $x$ trace class.
\end{Theorem*}
Maps $r \in \cB(H,K)$ which have the properties that $qr(1-p)$ is trace class, and $E(r)$ can be written as sum $a + x$ as in the theorem, are called \emph{admissible maps}, and we let $\cA(H,K)$ denote the space of admissible maps.

Using this theorem, we can now prove Theorem \ref{thmBoundednessAndExistence} in the case where $X \in \cDP_{st}$.
\begin{Theorem}\label{thmFermionPantsBoundedness}
Let $X = (\phi_t,t,w,s) \in \cDP_{st}$, and let $\rho$ be as in Theorem \ref{thmBoundednessAndExistence}.
Then the map $T:\F \otimes \F \to \F$ given by $T(\xi \otimes \eta) = Y(s^{L_0}\xi, w)e^{-t L(\rho)}\eta$ is bounded, where $Y$ is the free fermion state-field correspondence. 
Moreover, $E(X) = \C T$.
\end{Theorem}
\begin{proof}
Let $R > 1$ and let $X_R = (\Sigma_R, L_R, \Phi_R, \beta_R)$ be the non-degenerate extension of $X$ (Definition \ref{defNondegenerateExtension}).
Let $T_R \in E(X_R)$ be the element with $T_R(\Omega \otimes \Omega) = \Omega$. 
For $\xi \in \F^0$, let $S\xi = \lim_{R \downarrow 1} T_R^*\xi$, as in Proposition \ref{propAdjointsConverge}.

Let $H = L^2(S^1)$ and $pH = H^2(\D)$.
Let $V:\F \otimes \F \to \F_{H \oplus H, p \oplus p}$ be the isomorphism of $\CAR(H \oplus H)$ representations (Proposition \ref{propFockSumToTensor}), and let $S^\prime = VS$.
For $n \in \Z$, let $\xi_n = z^n \in H$, so that $a(\xi_J)a(\xi_I)^*\Omega$ gives an orthonormal basis for $\F$ indexed by pairs of finite subsets $J \subset \Z_{< 0}$ and $I \subset \Z_{\ge 0}$.
Let $W_{sz+w} \in \cB(H^2(\D))$ and $\tilde W_{sz+w} \in \cB(H^2(\D)^\perp)$ be the weighted composition operators associated to the map $z \mapsto sz +w$, corresponding to the positive square root $s^{1/2}$.
Let $W_{\phi_t} \in \cB(H^2(\D))$ and $\tilde W_{\phi_t} \in \cB(H^2(\D)^\perp)$ be the weighted composition operators associated to $\phi_t$ and the square root $\psi_t^2 = \phi_t^\prime$ with $\psi_t(0) > 0$ (as defined in Section \ref{subsecCompositionOperators}).
Let $W_1 = W_{sz + w} \oplus \tilde W_{sz + w} \in \cB(H)$ and $W_2 = W_{\phi_t} \oplus \tilde W_{\phi_t} \in \cB(H)$.
Note that $W_1$ and $W_2$ commute with $p$.

By Proposition \ref{propAdjointsConverge}, $S^\prime$ is the implementing operator defined in terms of the basis $\xi_i$ associated to $(r, \hOmega)$, where $r:H \to H \oplus H$ is given by $rf = (W_1f, W_2f)$ and $\hOmega = S^\prime\Omega$. 
By Proposition \ref{propAdjointsOnVacuum}, $\hOmega \in \F_{H \oplus H, p \oplus p}$ is, up to non-zero scalar, the vector assigned to a non-degenerate Riemann surface by the free fermion Segal CFT.
By \cite[Thm. 6.2]{Ten16}, such vectors are of the form $\hOmega_{q^\prime}$ for a projection $q^\prime$ with the property that $q^\prime - p \oplus p$ is trace class.
Thus we can study the boundedness of $S^\prime$ using Theorem \ref{thmAdmissibleBoundedness}, with $K = H \oplus H$ and $q = p \oplus p$.

By construction, $r = qrp + (1-q)r(1-p)$, and so to show that $S^\prime$ is bounded it suffices to prove that $r \in \cA(H,K)$.
Since $s\D + w \subset \interior{\D}$, $W_1$ is trace class (by \cite[Prop. 5.3]{ShapiroTaylor73}, for example).
Thus it suffices to show that $W_2$ can be written as the sum of a contraction and a trace class operator.

By Lemma \ref{lemExponentiatedVirasoroFermion}, $(e^{-tL(\rho)})^*: \F \to \F$ is the implementing operator associated to $(W_2, (e^{-t L(\rho)})^*\Omega)$.
By Proposition \ref{propAdjointsOnVacuum}, $(e^{-t L(\rho)})^*\Omega = \alpha U_{NS}(\psi,\gamma)\Omega$ for some $(\psi,\gamma) \in \Diff_+^{NS}(S^1)$ and some $\alpha \in \C^\times$.
By \cite[Prop. 6.8.2 and Prop. 6.3.1]{PrSe86}, $[u_{NS}(\psi,\gamma), p]$ is trace class, and thus $\alpha U_{NS}(\psi,\gamma)\Omega = \hat \Omega_{q^{\prime\prime}}$ for some projection $q^{\prime\prime}$ on $H$ with $q^{\prime\prime} - q$ trace class.
Since $e^{-t L(\rho)}$ is bounded by Lemma \ref{lemExponentiatedVirasoroFermion}, by Theorem \ref{thmAdmissibleBoundedness} we have $W_2 \in \cA(H,H)$.
Hence $r \in \cA(H,K)$ as well, and so by Theorem \ref{thmAdmissibleBoundedness}, $S^\prime$ is bounded.
It follows that $S = V^*S^\prime$ is bounded as well.

Now let $R_1 > R_2 > 1$.
We have $T_{R_1} = (R_2/R_1)^{L_0} T_{R_2}$, and thus $T_{R_1}^* = T_{R_2}^*(R_2/R_1)^{L_0}$.
Hence if $\xi \in \F$ is an eigenvector of $L_0$ with eigenvalue $m$, we have
$$
(R_1/R_2)^m T_{R_1}^*\xi = T_{R_2}^*\xi.
$$
Taking the limit of both sides as $R_2 \downarrow 1$, we get 
$$
T_{R_1}^*\xi = R_1^{-m} S\xi = SR_1^{-L_0}\xi.
$$
Since $R_1 > 1$ was arbitrary and eigenvectors $\xi$ for $L_0$ span a dense subspace of $\F$, we have $T_R^* = SR^{-L_0}$ for all $R > 1$.
Hence $\norm{T_R^*} \le \norm{S}$.

But $\norm{T_R^*} = \norm{T_R}$, and so the operators $T_R$ remain uniformly bounded in norm as $R \downarrow 1$.
Since we have already established that $T_R(\xi \otimes \eta) \to T(\xi \otimes \eta)$ for $\xi,\eta \in \F^0$ (Proposition \ref{propLimitIsRegularizedVertexOperator}), the uniform bound in norm is sufficient to guarantee that $T$ is bounded and that $T_R \to T$ in the strong operator topology.

We now show that $T \in E(X)$ by verifying that it satisfies the appropriate commutation relations.
It suffices to verify that 
$$
a(f^1)T = Ta(f^0), \qquad a(\overline{zf^1})^*T = Ta(\overline{zf^0})^*
$$
for $f=(f^1,f^0)$ lying in a dense subspace of $H^2(X)$.

Let $\Sigma$ be the underlying space of $X$, and let $(\psi,\gamma)$ be the standard boundary parametrization.
Then by definition, $H^2(X)$ is the closure of the set of $\psi \cdot (F \circ \gamma)$, where $F$ ranges over functions holomorphic in a neighborhood $U$ of $\Sigma$.
Given $f$ of the form $\psi \cdot (F \circ \gamma)$,
and $R > 1$ sufficiently small, $\beta_R^*F \in H^2(X_R)$. 
Moreover, $f_R := \beta_R^*F \to f$ in $L^2$ norm (in fact, uniformly) as $R \downarrow 1$, and so $a(f^i_R) \to a(f^i)$ in norm.
Hence taking limits in the expression $a(f^1_R)T_R = T_R a(f^0_R)$, we get $a(f^1)T = Ta(f^0)$, which establishes the first half of the commutation relations for $T$.

By \cite[Thm. 6.1]{Ten16}, $H^2(X_R)^\perp = \overline{M_{\pm z} H^2(X_R)}$, and so 
$$
a(\overline{z f^1_R})^* T_R = T_R a(\overline{z f^0_R})^*.
$$
Hence taking limits we get 
$$
a(\overline{z f^1})^* T = T a(\overline{z f^0})^*.
$$
We conclude that $T \in E(X)$, and since $\dim E(X) \le 1$ by Proposition \ref{propDimAtMostOne}, we conclude that $E(X) = \C T$.
\end{proof}

\section{Localized vertex operators and conformal nets}

\subsection{Localized vertex operators for the free fermion}

Recall that we use the term \emph{interval} to mean an open connected subset of $S^1$ which is non-empty and not dense, and that we denote by $I^\prime$ the complementary interval $\interior{I^c}$.

\begin{Definition}
Let $I \subset S^1$ be an interval.
We will write $\cDP(I)$ for the collection of 
$$
X =(\phi_t, t, w, s, \gamma,\psi) \in \cDR
$$
with the property that the boundary parametrizations $(\psi_j,\gamma_j)_{j \in \pi_0(\partial \Sigma)}$ satisfy 
$$
(\psi_{S^1},\gamma_{S^1})|_{I^\prime} = (\psi_{\phi_t(S^1)},\gamma_{\phi_t(S^1)})|_{I^\prime}.
$$
Given $T \in E(X)$ and $\xi \in \F$, let $T_\xi \in \cB(\F)$ be given by $T_\xi(\eta) = T(\xi \otimes \eta)$, where as usual we have ordered the incoming boundary components so that $w + s S^1$ comes first.
Define the set of vertex operators localized in $I$ 
$$
LV(I; \F) = \{ T_\xi : X \in \cDR(I), T \in E(X), \xi \in \F \} \subset \cB(\F).
$$
\end{Definition}
Graphically, we identify $X \in \cDR(I)$ with $T \in E(X)$, and depict them as on the left in Figure \ref{figLocalizedVertexOperator}.
A localized vertex operator is depicted by inserting a state $\xi$ into one of the input disks.
\noindent
\begin{minipage}{\linewidth}
$$
X = 
\begin{tikzpicture}[baseline={([yshift=-.5ex]current bounding box.center)}]
	\coordinate (a) at (120:1cm);
	\coordinate (b) at (240:1cm);
	\coordinate (c) at (180:.25cm);
% BIG DISK
	\fill[fill=red!10!blue!20!gray!30!white] (0,0) circle (1cm);
	\draw (0,0) circle (1cm);
% CURVED BOUNDARY REGION
	\fill[fill=white] (a)  .. controls ++(210:.6cm) and ++(90:.4cm) .. (c) .. controls ++(270:.4cm) and ++(150:.6cm) .. (b) -- ([shift=(240:1cm)]0,0) arc (240:480:1cm);
	\draw ([shift=(240:1cm)]0,0) arc (240:480:1cm);
	\draw (a) .. controls ++(210:.6cm) and ++(90:.4cm) .. (c);
	\draw (b) .. controls ++(150:.6cm) and ++(270:.4cm) .. (c);
% INNER DISK
	\filldraw[fill=white] (180:.65cm) circle (.25cm); 
%	\node at (180:.65cm) {\scriptsize{$\xi$}};
% POINT LABELS
%	\node at (0:1cm) {\scriptsize{\textbullet}};
%	\node at (0:0.8cm) {1};
%	\node at (0:1.2cm) {\scriptsize{\textbullet}};
%	\node at (0:1.4cm) {R};
%	\node at (180:.65cm) {\scriptsize{\textbullet}};
%	\node at (193:.7cm) {w};
% INTERVAL LABEL
	\draw (130:1.2cm) -- (130:1.4cm);
	\draw (230:1.2cm) -- (230:1.4cm);
	\draw (130:1.3cm) arc (130:230:1.3cm);
	\node at (180:1.5cm) {\scriptsize{$I$}};
% COORDINATE LABELS
%	\node at (a) {(a)};
%	\node at (b) {(b)};
%	\node at (c) {(c)};
\end{tikzpicture}
\,\,
, \qquad
T_\xi = 
\begin{tikzpicture}[baseline={([yshift=-.5ex]current bounding box.center)}]
	\coordinate (a) at (120:1cm);
	\coordinate (b) at (240:1cm);
	\coordinate (c) at (180:.25cm);
% BIG DISK
	\fill[fill=red!10!blue!20!gray!30!white] (0,0) circle (1cm);
	\draw (0,0) circle (1cm);
% CURVED BOUNDARY REGION
	\fill[fill=white] (a)  .. controls ++(210:.6cm) and ++(90:.4cm) .. (c) .. controls ++(270:.4cm) and ++(150:.6cm) .. (b) -- ([shift=(240:1cm)]0,0) arc (240:480:1cm);
	\draw ([shift=(240:1cm)]0,0) arc (240:480:1cm);
	\draw (a) .. controls ++(210:.6cm) and ++(90:.4cm) .. (c);
	\draw (b) .. controls ++(150:.6cm) and ++(270:.4cm) .. (c);
% INNER DISK
	\filldraw[fill=white] (180:.65cm) circle (.25cm); 
	\node at (180:.65cm) {\scriptsize{$\xi$}};
% POINT LABELS
%	\node at (0:1cm) {\scriptsize{\textbullet}};
%	\node at (0:0.8cm) {1};
%	\node at (0:1.2cm) {\scriptsize{\textbullet}};
%	\node at (0:1.4cm) {R};
%	\node at (180:.65cm) {\scriptsize{\textbullet}};
%	\node at (193:.7cm) {w};
% INTERVAL LABEL
	\draw (130:1.2cm) -- (130:1.4cm);
	\draw (230:1.2cm) -- (230:1.4cm);
	\draw (130:1.3cm) arc (130:230:1.3cm);
	\node at (180:1.5cm) {\scriptsize{$I$}};
% COORDINATE LABELS
%	\node at (a) {(a)};
%	\node at (b) {(b)};
%	\node at (c) {(c)};
\end{tikzpicture}
$$
\captionsetup{justification=centering,width=0.8\linewidth}
\begin{centering}

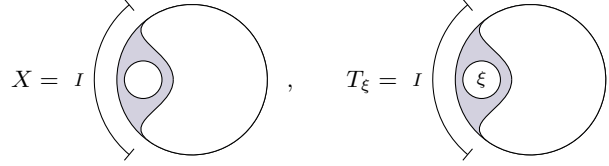
\captionof{figure}{A degenerate Riemann surface $X \in \cDP(I)$, and a localized vertex operator $T_\xi \in LV(I; \F)$}
\label{figLocalizedVertexOperator}
\end{centering}
\end{minipage}
\bigskip
\begin{Remark}
While the boundary parametrizations are not clear in Figure \ref{figLocalizedVertexOperator}, we are implicitly assuming that that $\gamma_{S^1}|_{I^\prime} = \gamma_{\phi_t}(S^1)|_{I^\prime}$ and $\psi_{S^1} |_{I^\prime} = \psi_{\phi_t(S^1)} |_{I^\prime}$, and that $\gamma_{S^1}(I^\prime) = I^\prime$.
More generally, the definition of $\cDP(I)$ allows $(\psi_{S^1}, \gamma_{S^1})|_J$ and $(\psi_{\phi_t(S^1)}, \gamma_{\phi_t(S^1)})|_J$ for an interval $J$ with $\gamma_{S^1}(J) = I^\prime$.
\end{Remark}

We begin with a straightforward observation on the parity of localized vertex operators.

\begin{Proposition}\label{propLocalizedVertexOperatorParity}
Let $X \in \cDP$ and $T \in E(X)$.
Let $\xi \in \F$ and let $T_\xi \in LV(I;\F)$ be the corresponding localized vertex operator.
Then $T_\xi$ is homogeneous if and only if $\xi$ is, and $p(T_\xi) = p(\xi)$.
\end{Proposition}
\begin{proof}
By Proposition \ref{propDegenerateOperatorReparametrization} and the fact that the representation $U_{NS}$ of $\Diff_+^{NS}$ takes values in even operators, it suffices to consider the case $X \in \cDP_{st}.$
In this case, we saw in Theorem \ref{thmFermionPantsBoundedness} that $E(X)$ is spanned by an even map $T$.
Thus $T_\xi$ has the same parity as $\xi$.
\end{proof}

We now have our main result on localized vertex operators for the free fermion.

\begin{Theorem}\label{thmFermionIsGood}
Every $LV(I;\F)$ is non-empty, and $LV(I;\F) \subset \cM(I)$, where $\cM(I)$ is the local algebra of the free fermion conformal net.
Moreover, $\cM(I)$ is generated by $LV(I;\F)$ as a von Neumann algebra.
\end{Theorem}
\begin{proof}
We will just write $LV(I)$ instead of $LV(I;\F)$.
Let $(\hat \psi,\hat \gamma) \in \Diff_+^{NS}(S^1)$ be a spin diffeomorphism with $\hat \gamma(I) = J$.
Then we claim that 
\begin{equation}\label{eqnLocalVertexOperatorsReparametrization}
U_{NS}(\hat \psi,\hat \gamma)LV(I)U_{NS}(\hat \psi,\hat \gamma)^* = LV(J).
\end{equation}
By symmetry, it suffices to show 
\begin{equation}\label{eqnLocalVertexOperatorsReparametrizationSubset}
U_{NS}(\hat \psi,\hat \gamma)LV(I)U_{NS}(\hat \psi,\hat \gamma)^* \subseteq LV(J).
\end{equation}

Suppose $T_\xi \in LV(I)$, corresponding to $\xi \in \F$ and  $T \in E(X)$, where $X = (\phi_t, t,w,s,\gamma,\psi) \in \cDP(I)$.
Let $\Sigma$ be the underlying space of $X$, and for $j \in \pi_0(\partial \Sigma)$ set
$$
\tilde \gamma_j = \left\{\begin{array}{ll}\gamma_j \circ \hat \gamma^{-1} & j \in \{S^1, \phi_t(S^1)\} \\
\gamma_j & j = w + s S^1
\end{array}\right. ,
\qquad
\tilde \psi_j = \left\{\begin{array}{ll}
\hat \psi \cdot (\psi_j \circ \hat \gamma^{-1}) & j \in \{S^1, \phi_t(S^1)\} \\
\psi_j & j = w + s S^1.
\end{array}\right.
$$
Set $\tilde X = (\phi_t, t, w, s, \tilde \gamma, \tilde \psi)$, and observe that $\tilde X \in \cDP(J)$.

By Proposition \ref{propDegenerateOperatorReparametrization}, $E(\tilde X)$ is spanned by $\tilde T := U_{NS}(\hat \psi, \hat \gamma) T (1 \otimes U_{NS}(\hat \psi, \hat \gamma)^*)$.
Hence 
$$
U_{NS}(\hat \psi, \hat \gamma) T_\xi U_{NS}(\hat \psi, \hat \gamma)^* \in LV(J).
$$
We have proven \eqref{eqnLocalVertexOperatorsReparametrizationSubset} and thus \eqref{eqnLocalVertexOperatorsReparametrization}.

It follows that $LV(I)$ is non-empty if and only if $LV(J)$ is.
Moreover, since 
$$
\cM(J) = U_{NS}(\psi,\gamma) \cM(I) U_{NS}(\psi,\gamma)^*
$$
by diffeomorphism covariance, $LV(I)$ generates $\cM(I)$ if and only if $LV(J)$ generates $\cM(J)$.

To show that some $LV(I)$ is non-empty, pick some one-parameter semigroup $\phi_t$ and a small value of $t$ so that $\phi_t(S^1) \cap S^1$ contains an interval but $\D \setminus \phi_t(\D)$ has non-empty interior, as in \eqref{eqnDegenSemigroup}.
Let $w,s \in \D$ be arbitrary values such that $w + s \D \subseteq \interior{\D} \setminus \phi_t(\D)$.
Let $I^\prime$ be an interval whose closure is compactly contained in $\phi_t(S^1) \cap S^1$.
We may choose a parametrization for $\phi_t(S^1)$ such that $(\psi_{\phi_t(S^1)}, \gamma_{\phi_t(S^1)})|_{I^\prime} \equiv (1, \operatorname{id})$.
Let $(\psi_{S^1},\gamma_{S^1}) = (1,\operatorname{id})$, and choose $(\psi_{w + sS^1},\gamma_{w + sS^1})$ arbitrarily.
Let $X = (\phi_t, t, w, s, \psi,\gamma) \in \cDP$, and by construction $X \in \cDP(I)$.
Then for any $\xi \in \F$ and $T \in E(X)$, we have $T_\xi \in LV(I)$.

Next, we show that $LV(I) \subset \cM(I)$.
Let $X =(\phi_t, t, w, s, \psi,\gamma) \in \cDP(I)$.
We claim that for arbitrary $f \in L^2(S^1)$ with $\supp f \subset I^\prime$, we have $(f,0,f) \in H^2(X)$, where as usual we have orderd the boundary components $S^1$, then $w + s S^1$, then $\phi_t(S^1)$.
Since $H^2(X)$ is closed, it suffices to prove the claim for continuous $f$.

Let $J = \gamma_{S^1}(I^\prime)$.
Since $\gamma_{S^1}|_{I^\prime} = \gamma_{\phi_t(S^1)}|_{I^\prime}$, we must have $J \subset S^1 \cap \phi_t(S^1)$.
Let $h$ be the continuous function on $S^1$ such that $\psi_{S^1} \cdot (h \circ \gamma_{S^1}) = f$.

Let $K = \D \setminus \phi_t(\interior{\D})$, and let $H:K \to \C$ be the continuous function obtained by extending $h$ to be $0$ on $K \setminus J$.
By Mergelyan's theorem \cite[\S 20]{BigRudin}, there exists a sequence of rational functions $H_n$ with poles at $0$ and $\infty$ such that $H_n \to H$ uniformly on $K$.
Let $(f_n, g_n, k_n) = \psi \cdot (H_n \circ \gamma) \in H^2(X)$ be the corresponding boundary values.
By construction, we have $g_n \to 0$ uniformly.
We also have that $f_n = \psi_{S^1} \cdot (H_n \circ \gamma_{S^1})$ converges uniformly to $f = \psi_{S^1} \cdot (h \circ \gamma_{S^1})$.
Moreover, since $\gamma_{S^1}|_{I^\prime} = \gamma_{\phi_t(S^1)}|_{I^\prime}$, we have that $k_n$ converges uniformly to $f$ on $I^\prime$.
By construction, $H_n$ is converging uniformly to $0$ on $\gamma_{\phi_t(S^1)}(I)$, and $f$ vanishes on $I$, so $k_n = \psi_{\phi_t(S^1)} \cdot (H_n \circ \gamma_{\phi_t(S^1)})$ converges uniformly to $f$ on $I$ as well, and hence on all of $S^1$.
Thus $(f,0,f) = \lim (f_n, g_n, k_n) \in H^2(X)$, as claimed.

Now let $T \in E(X)$ and $\xi \in \F$.
Then by the definition of $E(X)$, we have
$$
a(f)T = T(\Gamma \otimes a(f)),
$$
and thus $a(f)T_\xi = (-1)^{p(\xi)} T_\xi a(f)$.
As usual, formulas written involving the parity hold for homogeneous vectors, and are extended linearly otherwise.
Since $p(T) = p(\xi)$ by Proposition \ref{propLocalizedVertexOperatorParity}, the above equation yields
\begin{equation}\label{eqnLVcommutesnostar}
a(f)T_{\xi} = (-1)^{p(T)} T_{\xi} a(f),
\end{equation}

Since $(f,0,f) \in H^2(X)$ for every $f \in L^2(S^1)$ with $\supp f \subset I^\prime$, and $f \mapsto \overline{zf}$ gives a bijection from functions supported in $I^\prime$ to itself, we have $(f, 0, f) \in \overline{z H^2(X)}$ for all such $f$.
Thus by the definition of $E(X)$ we have
\begin{equation}\label{eqnLVcommutesstar}
a(f)^*T_{\xi} = (-1)^{p(T)} T_{\xi} a(f)^*.
\end{equation}
Combining \eqref{eqnLVcommutesnostar} and \eqref{eqnLVcommutesstar}, we have $T_\xi \in \cM(I)$ by Haag duality for the fermion net (Proposition \ref{thmFermiNetProps}).

Now let\footnote{
Here, we have used the notation that for $S \subset \cB(H)$,  $S^\prime$ denotes the \emph{commutant} of $S$, i.e. the algebra of all operators commuting with each element of $S$. If $S$ is closed under taking adjoints, then the von Neumann double commutant theorem says that $S^{\prime \prime}$ is the von Neumann algebra generated by $S$.}
$$
\cA(I) := (LV(I) \cup LV(I)^*)^{\prime \prime} \subseteq \cM(I).
$$ 
By \eqref{eqnLocalVertexOperatorsReparametrization}, $\cA(I)$ is a covariant subnet of $\cM(I)$, and so to show that $\cA(I) = \cM(I)$ if suffices to show that $\cA(I)\Omega$ is dense in $\F$ (by Proposition \ref{propSubnetTrivial}).

Let $X=(\phi_t,w,s,\psi,\gamma) \in \cDP(I)$, and let $T \in E(X)$ be nonzero.
We will show that
$$
W :=\Span \{ T_\xi \Omega : \xi \in \F\}
$$
is dense in $\F$.

Let $J$ be an interval with $\gamma_{S^1}^{-1}(J)$ compactly contained in $S^1 \setminus \phi_t(S^1)$, and let $f$ be a continuous function supported in $J$.
Let $h \in C(S^1)$ be such that $\psi_{S^1} \cdot (h \circ \gamma_{S^1}) = f$,  so that $h$ is supported in $\gamma_{S^1}^{-1}(J)$.
Let $H$ be the continuous function on $S^1 \cup \phi_t(S^1)$ obtained by defining $h$ to be $0$ outside of $\gamma_{S^1}^{-1}(J)$.

By Mergelyan's theorem, there exists a sequence of rational functions $H_n$ with poles at $0$, $\infty$ and $w$ such that $H_n \to h$ uniformly on $S^1 \cup \phi_t(S^1)$.
Let $(f_n,g_n,k_n)  = \psi \cdot (H \circ \gamma) \in H^2(X)$.
By construction $f_n \to f$ and $k_n \to 0$ uniformly.

Let $\xi\ \in \F$.
By the definition of $E(X)$, we have
$$
a(f_n)T(\xi \otimes \Omega) - (-1)^{p(\xi)}T(\xi \otimes a(k_n)\Omega) = T(a(g_n)\xi \otimes \Omega).
$$
Hence 
$$
a(f_n)T_\xi \Omega - (-1)^{p(\xi)}T_\xi a(k_n)\Omega = T_{a(g_n)\xi}\Omega.
$$
The left-hand side converges to $a(f)T_\xi \Omega$ as $n \to \infty$.
On the other hand, the right-hand side lies in $W$.
Hence $\overline{W}$ is invariant under $a(f)$ for every continuous $f$ supported in $J$, and thus for any $f \in L^2(S^1)$ with the same support.

A similar arugment shows that $\overline{W}$ is invariant under $a(f)^*$ where again $f \in L^2(S^1)$ is an arbitrary $L^2$ function supported in $J$.
Thus $\overline{W}$ contains $\overline{\cM(J)\Omega}$, which is all of $\F$ by the Reeh-Schlieder property (see \cite[\S 15]{Wa98}, or \cite[Thm. 1]{CaKaLo08}).
We conclude that $\overline{\cA(J)\Omega} = \F$ and thus $\cA = \cM$.
\end{proof}

\begin{Remark}
Using the notation of Figure \ref{figLocalizedVertexOperator}, the conclusion of  Theorem \ref{thmFermionIsGood} can be depicted
$$
\cM(I) = \left\{\begin{tikzpicture}[baseline={([yshift=-.5ex]current bounding box.center)}]
	\coordinate (a) at (120:1cm);
	\coordinate (b) at (240:1cm);
	\coordinate (c) at (180:.25cm);
% BIG DISK
	\fill[fill=red!10!blue!20!gray!30!white] (0,0) circle (1cm);
	\draw (0,0) circle (1cm);
% CURVED BOUNDARY REGION
	\fill[fill=white] (a)  .. controls ++(210:.6cm) and ++(90:.4cm) .. (c) .. controls ++(270:.4cm) and ++(150:.6cm) .. (b) -- ([shift=(240:1cm)]0,0) arc (240:480:1cm);
	\draw ([shift=(240:1cm)]0,0) arc (240:480:1cm);
	\draw (a) .. controls ++(210:.6cm) and ++(90:.4cm) .. (c);
	\draw (b) .. controls ++(150:.6cm) and ++(270:.4cm) .. (c);
% INNER DISK
	\filldraw[fill=white] (180:.65cm) circle (.25cm); 
	\node at (180:.65cm) {\scriptsize{$\xi$}};
% POINT LABELS
%	\node at (0:1cm) {\scriptsize{\textbullet}};
%	\node at (0:0.8cm) {1};
%	\node at (0:1.2cm) {\scriptsize{\textbullet}};
%	\node at (0:1.4cm) {R};
%	\node at (180:.65cm) {\scriptsize{\textbullet}};
%	\node at (193:.7cm) {w};
% INTERVAL LABEL
	\draw (130:1.2cm) -- (130:1.4cm);
	\draw (230:1.2cm) -- (230:1.4cm);
	\draw (130:1.3cm) arc (130:230:1.3cm);
	\node at (180:1.5cm) {\scriptsize{$I$}};
% COORDINATE LABELS
%	\node at (a) {(a)};
%	\node at (b) {(b)};
%	\node at (c) {(c)};
\end{tikzpicture}
\,\, \Bigg| \,\,
\xi \in \F, \,\,\partial\mbox{-parametrizations } \right\}^{\prime\prime},
$$
which we refer to as a `geometric realization' of the algebraic CFT $\cM$.
\end{Remark}

\subsection{Localized vertex operators for other vertex operator superalgebras}\label{secOtherSuperalgebras}

Let $(V,Y,\Omega,\nu,\ip{\,\cdot\,,\,\cdot\,},\theta)$ be a simple unitary vertex operator superalgebra, with Hilbert space completion $\cH$.
Let $U:\Diff^{NS}_+(S^1) \to \cPU(\cH)$ be the positive energy representation of $\Diff^{NS}_+(S^1)$ coming from the conformal vector $\nu$.
\begin{Definition}
For $X = (\phi_t, t, w ,s, \psi,\gamma) \in \cDP$, we define $E(X;V)$ to be the one-dimensional vector space of (a priori unbounded) linear maps $\cH \otimes \cH \to \cH$ spanned by 
$$
T(\xi \otimes \eta) = U(\hat \psi_{S^1}, \hat \gamma_{S^1})Y(s^{L_0}U(\hat \psi_{w + s S^1}, \hat \gamma_{w + s S^1})^*\xi, w)e^{-t L(\rho)}U(\hat \psi_{\phi_t(S^1)}, \hat \gamma_{\phi_t(S^1)})^*\eta,
$$
where $\hat \gamma_j$ and $\hat \psi_j$ are given in terms of the standard boundary parametrization $(\psi_{st}, \gamma_{st})$ by 
$$
\hat \gamma_j = \gamma_{j}^{-1} \circ \gamma_{j,st} \in \Diff_+(S^1), \quad \mbox{ and } \quad \hat \psi_j = \frac{\psi_j}{\psi_{j,st} \circ \hat \gamma_j^{-1}}.
$$
This definition is characterized by the fact that when $X \in \cDP_{st}$, $E(X)$ is spanned by the map
$$
T(\xi \otimes \eta) = Y(s^{L_0} \xi, w)e^{-tL(\rho)}\eta,
$$
and the spaces satisfy the same diffeomorphism covariance property that the free fermion localized vertex operators enjoyed.

As before, for $\xi \in \cH$, set $T_\xi(\eta) = T(\xi \otimes \eta)$ and set 
$$
LV(I;V) = \{T_\xi \, : \, X \in \cDR(I),\, T \in E(X;V), \,\xi \in \cH\}.
$$
\end{Definition}
By Proposition \ref{propVertexOperatorDenselyDefined}, elements of $E(X;V)$ are densely defined, but we do not have proof that they are bounded in general, or even that they extend to the algebraic tensor product $\cH \otimes_{alg} \cH$.
However, the maps $T_\xi$ are densely defined for $\xi$ lying in a dense subspace.

In the case where $V$ is the free fermion $\F^0$, however, $E(X;\F^0)$ agrees with the one-dimensional space $E(X)$ from Section \ref{secSegalCFTForDegenerateSurfaces} by Theorem \ref{thmFermionPantsBoundedness} and Proposition \ref{propDegenerateOperatorReparametrization}.
The free fermion will be our motivating example for defining what it means for a unitary vertex operator superalgebra to have a `good' theory of localized vertex operators.

\begin{Definition}
Let $V$ be a simple unitary vertex operator superalgebra, and let $\cH$ be its Hilbert space completion.
We say that $V$ has \emph{bounded localized vertex operators} if 
\begin{itemize}
\item Maps $T \in E(X;V)$ extend to bounded linear maps in $\cB(\cH \otimes \cH, \cH)$.
\item For intervals $I$, if we set $\cA_V(I) := (LV(I;V) \cup LV(I;V)^*)^{\prime\prime}$, then $\cA_V$ is a Fermi conformal net with conformal symmetry $U:\Diff_+^{NS}(S^1) \to \cPU(\cH_V)$ coming from the conformal vector $\nu$ of $V$.
\end{itemize}
\end{Definition}

Many of the required axioms of a Fermi conformal net are automatically satisfied once the maps $T \in E(X;V)$ are bounded, and so we give a set of sufficient conditions that one can check.
\begin{Proposition}\label{refSufficientForBLVO}
Let $V$ be a simple unitary vertex operator superalgebra, let $U:\Diff_+^{NS}(S^1) \to \cPU(\cH_V)$ be the associated projective representation of $\Diff_+^{NS}(S^1)$, and suppose that the following hold:
\begin{itemize}
\item Maps $T \in E(X;V)$ extend to bounded linear maps in $\cB(\cH \otimes \cH, \cH)$.
\item The algebras $\cA_V(I) = (LV(I;V) \cup LV(I;V)^*)^{\prime\prime}$ satisfy graded locality (i.e., when $I \cap J = \emptyset$, we have $[\cA_V(I), \cA_V(J)]_{\pm} = \{0\}$).
\item $U(\psi,\gamma)$ commutes elementwise with $\cA(I)$ whenever $(\psi,\gamma) \in \Diff_+(I^\prime)$.
\end{itemize}
Then $\cA_V$ is a Fermi conformal net with conformal symmetry $U$.
\end{Proposition}
\begin{proof}
The sets $LV(I;V)$ are $\Z/2\Z$-graded and satisfy $LV(I;V) \subset LV(J;V)$ when $I \subset J$, and the corresponding properties of $\cA_V(I)$ are immediate consequences.
Similarly $U(\psi,\gamma)LV(I;V)U(\psi,\gamma)^* = LV(\gamma(I);V)$, and diffeomorphism covariance of $\cA_V$ follows, given our assumption that $U(\psi,\gamma)$ commutes with $\cA_V(I)$ when $(\psi,\gamma) \in \Diff_+(I^\prime)$.
Since we have also assumed that $\cA_V$ satisfies graded locality, the only thing to check is the vacuum axiom.

Since $V$ is simple, $\Omega$ is the unique (up to scalar) vector fixed by $\Mob^{NS}$. 
Fix an interval $I$, and let $\cK = \overline{\cA(I)\Omega}$.
If $(\psi,\gamma) \in \Diff_+(I)$, then 
$$
U(\psi,\gamma)LV(I;V) = LV(I;V)U(\psi,\gamma) = LV(I;V).
$$
Hence $U(\psi,\gamma)\cK \subseteq \cK$.
By the Reeh-Schlieder property for the Virasoro nets, it follows that $U(\psi,\gamma)\Omega \in \cK$ for all $(\psi,\gamma) \in \Diff_+^{NS}(S^1)$.
Now from the definition of $LV(I;V)$, we can see that $\cK$ contains $Y(a,z)\Omega$ for all $a \in V$, for at least one $z \in \interior{\D}$.

Since $\cK \subseteq \overline{\cA_V(S^1)\Omega}$, we have $Y(a,z)\Omega \in \overline{\cA_V(S^1)\Omega}$ for the same $a$ and $z$ as above.
But $\overline{\cA_V(S^1)\Omega}$ is clearly unvariant under the rotation subgroup of $\Diff_+^{NS}$, and thus is $\tfrac12\Z$-graded.
Thus when $a$ is homogeneous, we must have $a \in \overline{\cA_V(S^1)\Omega}$, which establishes that $\overline{\cA_V(S^1)\Omega} = \cH$.
\end{proof}

\begin{Remark}
The first two conditions in the statement of Proposition \ref{refSufficientForBLVO} are analagous to the conditions required in \cite{CKLW18} to construct a conformal net;
the first is analogous to energy boundedness, and the second to strong locality.
We expect that it is not too difficult to show that the third condition holds automatically in the presence of the first two, but we will not discuss this question as the third condition is easily verified in all of our examples.
\end{Remark}

\begin{Remark}
The fact that we have defined $\cA_V(I)$ to be generated by $LV(I; V) \cup LV(I;V)^*$ instead of just $LV(I;V)$ is an artifact of the fact that we have only considered a special class of degenerate annuli and pairs of pants lying in one-parameter families (see also the discussion in Section \ref{secOutlookGeometry}).
If we were to instead define $LV(I;V)$ to be maps assigned to \emph{all} degenerate pairs of pants we would have $LV(I; V)= LV(I;V)^*$.
\end{Remark}

Our next project is to show that the property of having bounded localized vertex operators is well-behaved with respect to tensor products and taking unitary subalgebras.

\begin{Proposition}\label{propTensorProductGood}
Let $V_1$ and $V_2$ be simple unitary vertex operator superalgebras.
Then $V_1 \otimes V_2$ has bounded localized vertex operators if and only if $V_1$ and $V_2$ do. In this case, $\cA_{V_1 \otimes V_2} = \cA_{V_1} \otimes \cA_{V_2}$.
\end{Proposition}
\begin{proof}
Let $Y$ be the state-field correspondence for $V_1 \otimes V_2$, and let $Y^i$ be the state-field correspondence for $V_i$.
By definition, we have $Y(\xi_1 \otimes \xi_2, x) = Y^1(\xi_1, x)\Gamma^{p(\xi_2)} \otimes Y^2(\xi_2, x)$ for homogeneous $\xi_i \in V_i$.
If $X \in \cDP(I)$ and $T \in E(X; V_1 \otimes V_2)$, we have $T_{\xi_1 \otimes \xi_2} = T_{\xi_1} \grotimes T_{\xi_2}$.
Thus the boundedness of elements of $LV(I; V_1 \otimes V_2)$ is equivalent to the boundedness of elements of $LV(I; V_1) \grotimes LV(I; V_2)$, and we have $LV(I; V_1 \otimes V_2) \subseteq LV(I; V_1) \grotimes LV(I; V_2)$.

First consider when $V_1$ and $V_2$ have bounded localized vertex operators. 
Then $\cA_{V_1 \otimes V_2}$ is a diffeomorphism covariant subnet of $\cA_{V_1} \otimes \cA_{V_2}$, and to check equality it suffices to show that $\overline{\cA_{V_1 \otimes V_2} \Omega} = \cH_{V_1 \otimes V_2}$.
This can be done just as in the proof of Proposition \ref{refSufficientForBLVO}.

Now consider when $V_1 \otimes V_2$ has bounded localized vertex operators.
The inclusion $\cA_{V_1 \otimes V_2}(I) \subseteq \cA_{V_1}(I) \otimes \cA_{V_2}(I)$ is clear, but it requires a small argument to establish the reverse inclusion.
Let $\cH_i$ be the Hilbert space completion of $V_i$, and let $U_i$ be the projective representation of $\Diff_+^{NS}(S^1)$ on $\cH_i$ obtained by integrating the representation of the Virasoro algebra coming from $V_i$.
Let $\cK_i$ be the subspace of $\cH_i$ generated by $\Omega$ under $U_i$, and let 
$$
\cB_i(I) = \{U_i(\psi,\gamma) : (\psi, \gamma) \in \Diff_+(I)\}^{\prime\prime}.
$$
Then $\cB_i(I)$ and $\cB_i(J)$ commute when $I$ and $J$ are disjoint (see \cite[\S 3.2]{CKLW18}).

Let $\cC(I) \subseteq \cA_{V_1 \otimes V_2}(I)$ be the local algebra of the Virasoro subnet, given by
$$
\cC(I) = \{U_1(\psi, \gamma) \otimes U_2(\psi,\gamma) : (\psi,\gamma) \in \Diff_+(I)\}^{\prime \prime},
$$
and observe that $\cC(I) \subset \cB_1(I) \otimes \cB_2(I)$, and that $\cC(I)$ commutes with $\cB_1(J) \otimes \cB_2(J)$ when $I$ and $J$ are disjoint.

We now set out to verify that $\cB_1(I) \otimes \cB_2(I) \subseteq \cA_{V_1 \otimes V_2}(I)$.
Fix $X \in \cDP(I)$, let $T \in E(X; V_1 \otimes V_2)$, and let $T_1, T_2$ be such that $T_1 \otimes T_2 = T_{\Omega \otimes \Omega}$.
From the definition of $E(X; V_1 \otimes V_2)$, $T_1\Omega \otimes T_2\Omega$ lies in $\overline{\cC(S^1)(\Omega \otimes \Omega)}$, whose finite energy vectors are the subrepresentation of the Virasoro algebra $L_n^{V_1 \otimes V_2}$ generated by the vaccum $\Omega \otimes \Omega$.
Hence $T_1 \otimes T_2 \in \cC(I)$ by Proposition \ref{propSubnetTrivial}, and thus $T_1 \otimes T_2 \in \cB_1(I) \otimes \cB_2(I)$.
It is a standard, but non-trivial, fact about von Neumann algebras that we may now conclude $T_i \in \cB_i(I)$.

Now suppose that $Y \in \cDP(J)$ for some interval $J$ disjoint from $I$, and let $S \in E(Y; V_1 \otimes V_2)$.
Then writing $S_{\Omega \otimes \Omega} = S_1 \otimes S_2$, we have $S_i \in \cB_I(J)$, as above, and thus $[S_i, T_i] = 0$.
Now if we select $a,b \in V_1$, we have
$$
T_{a \otimes \Omega} = \tilde T_1 \otimes T_2, \qquad S_{b \otimes \Omega} = \tilde S_1 \otimes S_2
$$
for some operators $\tilde T_1 \in LV(X, V_1)$ and $\tilde S_1 \in LV(Y, V_1)$.
Since $V_1 \otimes V_2$ has bounded localized vertex operators, $T_{a \otimes \Omega}$ and $S_{b \otimes \Omega}$ supercommute.
But since $T_2$ and $S_2$ are even and commute, and their product is nonzero, $\tilde T_1$ and $\tilde T_2$ supercommute as well.
Since all elements of $LV(X,V_1)$ and $LV(Y,V_1)$ arise as above, and we may apply the same argument to the adjoints, we 
get that $\cA_{V_1}(I)$ and $\cA_{V_1}(J)$ supercommute elementwise.
Applying the same argument to the second tensor factor shows that $\cA_{V_2}(I)$ and $\cA_{V_2}(J)$ also supercommute, and we conclude that the algebras $\cA_{V_i}(I)$ are graded local.
The same argument can also be used to show that $\cA_{V_i}(I)$ commutes with $\cB_i(J)$ when $I$ and $J$ are disjoint, which completes the proof that both $V_i$ have bounded localized vertex operators, by Proposition \ref{refSufficientForBLVO}.
%In this case, it follows that $\cA_{V_1 \otimes V_2}(I) = \cA_{V_1}(I) \grotimes \cA_{V_2}(I)$.
%It is straightforward to verify that given positive energy representations $U_1$ and $U_2$, the tensor product $\cA_{V_1} \otimes \cA_{V_2}$ satisfies graded locality and diffeomorphism covariance if and only if both $\cA_{V_i}$ do.
\end{proof}

\begin{Theorem}\label{thmSubtheoryGood}
Let $V$ be a simple unitary vertex operator superalgebra with bounded localized vertex operators, and let $W$ be a unitary subalgebra. 
Then $W$ has bounded localized vertex operators.
\end{Theorem}
\begin{proof}
First consider when $W$ is a conformal subalgebra; that is, when the conformal vector $\nu^V$ of $V$ lies in $W$.
Let $e_W \in \cB(\cH_V)$ be the projection onto $\cH_W$, the closure of $W$, and let
$$
LV(I;V)_W = \{ T_\xi : X \in \cDP(I), T \in E(X; V), \xi \in \cH_W\}.
$$
Since $W$ is a conformal subalgebra, $e_W$ commutes with all unitaries $U(\psi, \gamma)$ and with $e^{-t L(\rho)}$.

Let $X \in \cDP(I)$ and let $T \in E(X;V)$.
Recall that $T_\xi$ is given by the formula
$$
T_\xi(\eta) = U(\psi_1,\gamma_1)Y(s^{L_0} U(\psi_2,\gamma_2)^*\xi, w)e^{-t L(\rho)}U(\psi_3,\gamma_3)^*\eta
$$
when $\xi \in U(\psi_2,\gamma_2)V$ and $\eta \in U(\psi_3,\gamma_3)V$,  for some $(\psi_j,\gamma_j) \in \Diff^{NS}_+(S^1)$.
By the super version of \cite[Lem. 5.28]{CKLW18}, we have 
$$
e_W U(\gamma_1)Y(s^{L_0} U(\gamma_2)^*\xi, w) e^{-t L(\rho)} U(\gamma_3)^* e_W \eta = U(\gamma_1)Y(s^{L_0} U(\gamma_2)^*e_W \xi, w) e^{-t L(\rho)} U(\gamma_3)^* e_W \eta
$$
for all such $\xi$ and $\eta$, and
$$
U(\gamma_1)Y(s^{L_0} U(\gamma_2)^*\xi, w) e^{-t L(\rho)} U(\gamma_3)^* e_W \eta = e_W U(\gamma_1)Y(s^{L_0} U(\gamma_2)^*e_W \xi, w) e^{-t L(\rho)} U(\gamma_3)^* \eta
$$
for $\xi \in U(\gamma_2)W$ and $\eta \in U(\gamma_3)W$.

Since $T \in E(X;V)$ is bounded by assumption, these relations extend to all of $\cH_V$ and $\cH_W$, and we get
\begin{equation}\label{eqnLVCompression}
e_W T_\xi e_W = T_{e_W \xi} e_W = e_W T_{e_W \xi}
\end{equation}
for all $\xi \in \cH_V$.
Thus $LV(I;W) = e_W LV(I;V) e_W = e_W LV(I;V)_W$, so $LV(I;W)$ consists of bounded operators.

Let $\cB(I) = (LV(I;V)_W \cup {LV(I;V)_W}^*)^{\prime \prime}$. 
It is clear that $\cB(I) \subset \cA_V(I)$ and that $\cB(J) \subset \cB(I)$ when $J \subset I$.
For any $(\psi,\gamma) \in \Diff^{NS}_+(S^1)$ we have $U(\psi,\gamma)LV(I;V)_W U(\psi,\gamma)^* = LV(\gamma(I),V)_W$, and thus $U(\psi,\gamma)\cB(I)U(\psi,\gamma)^* = \cB(\gamma(I))$.
Hence $\cB$ is a  covariant subnet of $\cA_V$.

Let 
$$
\cB(S^1) = \bigvee_{I \in \cI} \cB(I)
$$
be the von Neumann algebra generated by all $\cB$ local algebras assigned to intervals.
Let $\cH_\cB = \overline{\cB(S^1)\Omega}$, so that $\cB$ is a Fermi conformal net on $\cH_\cB$ by Theorem \ref{thmSubnetsAreNets}.
We will show that $\cH_B = \cH_W$ and that $\cB(I)e_W = \cA_W(I)$, which will establish that $\cA_W(I)$ is a Fermi conformal net with confomal symmetry $U e_W$.

Since elements of $LV(I;V)_W$ commute with $e_W$ by \eqref{eqnLVCompression}, we have $\cH_\cB \subset \cH_W$.
Since $e_W LV(I;V) e_W = e_W LV(I;V)_W$ and $e_W LV(I;V)^* e_W = e_W {LV(I;V)_W}^*$, we have $e_W \cA_V(I) e_W = e_W \cB(I)$ and thus 
$$
\cH_\cB = e_W \cH_\cB \supseteq \overline{e_W \cB(I)\Omega} = \overline{e_W \cA_V(I) \Omega} = \cH_W.
$$
Hence $\cH_\cB = \cH_W$, and thus we have a Fermi conformal net $e_W \cB(I)$ on $\cH_W$ with conformal symmetry $e_W U$.
Moreover, since $LV(I;W) = e_W LV(I;V)_W$ and $LV(I;W)^* =  e_W {LV(I;V)_W}^*$, we have $e_W \cB(I) = \cA_W(I)$, which completes the proof when the inclusion $W \subset V$ is conformal.

Now consider when the inclusion $W \subset V$ is not conformal.
Let $\tilde W = \{ \xi_{(-1)} \eta : \xi \in W, \, \eta \in W^c\} \subseteq V$. 
By Proposition \ref{propWandCommutantGenerateTensorProduct}, $\tilde W$ is a unitary conformal subalgebra of $V$, so by the above proof $\tilde W$ has bounded localized vertex operators. 
But by the same proposition, $\tilde W$ is unitarily equivalent to $W \otimes W^c$, so by Proposition \ref{propTensorProductGood}, $\tilde W$ has bounded localized vertex operators as well.
\end{proof}

Theorem \ref{thmFermionIsGood}, combined with Proposition \ref{propExpIsNS},  says that the free fermion vertex operator algebra $\F^0$ has bounded localized vertex operators.
We can use Proposition \ref{propTensorProductGood} and Theorem \ref{thmSubtheoryGood} to extend this to more examples.
\begin{Theorem}\label{thmFermiSubalgebrasAreGood}
Let $W$ be a unitary subalgebra of $(\F^0)^{\otimes N}$ for some $N \in \Z_{\ge 1}$.
Then $W$ has bounded localized vertex operators. 
\end{Theorem}
\begin{proof}
By Proposition \ref{propTensorProductGood}, $(\F^0)^{\otimes N}$ has bounded localized vertex operators, and so by Theorem \ref{thmSubtheoryGood} the same is true of any unitary subalgebra.
\end{proof}

We are led naturally to ask which unitary vertex operator algebras can arise as unitary subalgebras of $(\F^0)^{\otimes N}$.
We have nothing approaching an exhaustive answer, but this class includes many important examples.

\begin{Example}[The free boson]
The free boson arises as the charge zero component of $\F^0$, a result which comprises one half of the fermion-boson correspondence (see \cite[\S5.1-5.2]{Kac98}).
The free boson is a unitary subalgebra of $\F^0$ since it is conformal (and in particular, $L_1$-invariant), and $\theta$-invariant, as $\theta$ exchanges the charge $M$ and charge $-M$ subspaces of $\F^0$.
\end{Example}

\begin{Example}[Sublattices of $\Z^N$]\label{exLattice}
Given a positive definite integral lattice $\Lambda$, there is a corresponding simple vertex operator superalgebra $V_\Lambda$ (see \cite[\S5.5]{Kac98}) which has a natural unitary structure (\cite[Thm 4.12]{DongLin14} and \cite[Thm 2.9]{AiLin17}). 
As discussed in \cite[Ex. 5.5a]{Kac98}, $(\F^0)^{\otimes N}$ is the vertex operator superalgebra corresponding to the lattice $\Z^N$.
Given a sublattice $\Lambda \subset \Z^N$, one has an embedding of vertex operator superalgebras $V_\Lambda \subset (\F^0)^{\otimes N}$.
It is straightforward to check that if $\Lambda \subset \Lambda^\prime$, then $V_\Lambda$ is a unitary subalgebra of $V_{\Lambda^\prime}$ from explicit formulas for $\theta_L$ (see \cite[Lem. 2.8]{AiLin17}, where $\theta$ is called $\phi$) and for $L_1$ (see the proof of \cite[Prop. 5.5]{Kac98}).
\end{Example}

\begin{Example}[Many WZW models]\label{exWZW}
Let $G$ be compact, simple, simply connected Lie group, and let $\g$ be its complexified Lie algebra.
Since the weight 1 subspace of $(\F^{0})^{\otimes N}$ contains a copy of $\fru(N)$ (see \cite[\S7]{Wa98} for an explicit construction), given a unitary representation $\pi: G \to \C^N$, we obtain an embedding of the affine vertex algebra $V_{\g,\Delta_\pi} \hookrightarrow (\F^0)^{\otimes N}$ at some level $\Delta_\pi \in \Z_{> 0}$, called the Dynkin index of $\pi$.
It is clear from the explicit formula for the action of the matrix units $(E_{ij})_{(-1)}$ on $(\F^0)^{\otimes N}$ (see e.g. \cite[\S7]{Wa98}) that $V_{\g,\Delta_\pi}$ is invariant under $\theta_\F$.
Since $V_{\g,\Delta_\pi}$ is generated by vectors with weight $1$, it will automatically be invariant under $L_1$.
Thus $V_{\g,\Delta_\pi}$ is a unitary subalgebra of $(\F^0)^{\otimes N}$.

For $k \in \Z_{> 0}$, $V_{\g,k\Delta_\pi}$ is a unitary subalgebra of $V_{\g,\Delta_\pi}^{\otimes k}$, and thus every $V_{\g,k \Delta_\pi}$ has bounded localized vertex operators.
The smallest Dynkin indices $\Delta_\g = \min_\pi \Delta_\pi$ for each $\g$ are given in Figure \ref{figDynkinIndices} (see  \cite[Tbl. 5]{DynkinIndex} and \cite[Prop. 2.6]{LaszloSorger}).
For more details on this construction, see the discussion at the beginning of \cite[\S5.2]{Po03}.

\begin{figure}[h!tbp]
$$
\begin{array}{|r|c|c|c|c|c|c|c|c|c|}
\hline \g=& A_n & B_n& C_n& D_n& E_6 & E_7 & E_8 & F_4 & G_2\\
\hline \Delta_\g=&1 &2 &1&2&6&12&60&6&2\\
\hline
\end{array}
$$
\caption{Minimal Dynkin indices for simple Lie algebras}
\label{figDynkinIndices}
\end{figure}

Since the $D_n$ level 1 VOA comes from a sublattice of $\Z^n$, we have in fact shown that the $D_n$ VOAs have bounded localized vertex operators at all positive integer levels as a consequence of Example \ref{exLattice}, instead of just at even ones as suggested by Figure \ref{figDynkinIndices}.
Of course, the $A_n$ and $C_n$ VOAs also have bounded localized vertex operators at every level as a consequence of Figure \ref{figDynkinIndices}.
We expect that all affine VOAs have bounded localized vertex operators. 
\end{Example}

\begin{Example}[Many (super) Virasoro models]\label{exVirasoro}
If $c \in \Z_{\ge 1}$, then the Virasoro vertex operator algebra with central charge $c$ is a unitary subalgebra of $(\F^0)^{\otimes c}$, and thus has bounded localized vertex operators.
If $c$ lies in the discrete series, then the corresponding Virasoro VOA is realized as a subalgebra of $SU(2)_n \otimes SU(2)_1$ inside the unitary coset subalgebra ${SU(2)_{n+1}}^c$ (the Goddard-Kent-Olive construction \cite{GKO}).
Thus the discrete series of Virasoro VOAs have bounded localized vertex operators, since $SU(2)_n \otimes SU(2)_1 \subset (\F^0)^{\otimes 2n+2}$ is a unitary subalgebra.
We get the same when $c$ is the sum of an integer and values in the discrete series of unitary Virasoro representations, and when $c$ is the central charge of a coset of one of the other examples given (and so on).

Similarly, the discrete series of $(N=1)$ super Virasoro vertex operator algebras are realized in the coset of $SU(2)_{n+2} \subset SU(2)_{n} \otimes SU(2)_2$ (by \cite[\S3]{GKO}, see also \cite[\S6.4]{CaKaLo08}), and so have bounded localized vertex operators.
In \cite[\S 5]{CarpiHillierKawahigashiLongoXu15}, it is shown that the discrete series of $N=2$ super Virasoro VOAs can be embedded as unitary subalgebras of free fermions, and in a recent paper \cite{MasonTuiteYamskulna18}, the $N=4$ super conformal algebra with central charge $c=6$ is realized as a unitary conformal subalgebra of $(\F^0)^{\otimes 6}$.
\end{Example}

\begin{Remark}\label{rmkCKLWComparison}
The main results of \cite{CKLW18} should generalize to the case of super VOAs and Fermi conformal nets without any major modification, and using the ``super version'' of that paper, one can prove that the Fermi conformal nets constructed via Theorem \ref{thmFermiSubalgebrasAreGood} from unitary subalgebras $V \subset (\F^0)^{ \otimes N}$ coincide with the CKLW nets (that is, the nets constructed in \cite{CKLW18}).
The free fermion Fermi conformal net is, by definition, generated by smeared generating fields for the free fermion vertex operator superalgebra, and so the CKLW free fermion net agrees with the one constructed from $\F^0$ via Theorem \ref{thmFermionIsGood}.
By \cite[Cor. 8.2]{CKLW18} and Proposition \ref{propTensorProductGood}, the net constructed from $(\F^0)^{\otimes N}$ agrees with the CKLW net.
Now by the super version of  \cite[Thm. 7.1]{CKLW18}, unitary subalgebras of $(\F^0)^{\otimes N}$ are strongly local, and the corresponding CKLW nets agree with the ones constructed from bounded localized vertex operators by the super version of \cite[Thm. 7.4]{CKLW18}.
A direct proof that the even part of $(\F^0)^{\otimes N}$ is strongly local will also appear in \cite{CarpiWeinerXu}, which implies that any even unitary subalgebra of $(\F^0)^{\otimes N}$ is strongly local by the results of \cite{CKLW18}.

We expect that the above discussion should apply to any simple unitary vertex operator superalgebra with bounded localized vertex operators. That is, we expect that such vertex operator superalgebras are energy bounded and strongly local, and that the Fermi conformal net arising from the bounded localized vertex operators is isomorphic to the CKLW net.
\end{Remark}

\subsection{Further directions}\label{subsecOutlook}

The goal of this paper is to demonstrate that many Fermi conformal nets can be constructed geometrically from unitary vertex operator superalgebras via assigning values to some degenerate Riemann surfaces.
In the interest of (relative) brevity, we have not attempted to develop a general theory of degenerate Riemann surfaces, or bounded localized vetex operators.
In this section we will briefly discuss several directions for future research.

\subsubsection{Relaxing the semigroup condition for fermions}\label{secOutlookGeometry}

Let $U \subset \D$ be a Jordan domain with $C^\infty$ boundary, and let $\phi:\D \to \overline{U}$ be a Riemann map.
For our construction of Fermi conformal nets, it sufficed to assign bounded operators to degenerate annuli $\D \setminus U$ with the property that $\phi$ fit into a one-parameter semigroup fixing $0$.
This condition was essential to our proof, but it would be very surprising if it were anything other than a technical convenience.
In the free fermion example, we saw that the boudnedness of the operator assigned to the degenerate annulus is equivalent to being able to write $W_\phi$ as the sum of a contraction and a trace class operator.
This, in turn, is equivalent to a condition on the decay of the \emph{approximation numbers}\footnote{
One might also call these the \emph{singular values} of $W_\phi$, but this term is sometimes reserved for compact operators}
$$
a_n(W_\phi) = \inf \{ \norm{W_\phi - F} : \operatorname{rank}(F) < n\}.
$$
When $U \subset \D$ is a Jordan domain with $C^\infty$ boundary and $\overline{U} \cap S^1 \ne \emptyset$, we have $\lim_{n \to \infty} a_n(W_\phi) = 1$, and the boundedness of the operator assigned to the degenerate annulus is equivalent to the condition $\prod_{n=1}^\infty a_n(W_\phi) < \infty$.

The $\phi$ with this property on the approximation numbers (relaxing the requirement that $\overline{U} \cap S^1 \ne \emptyset$) form a semigroup, and it is quite large.
As a consequence of the results in this paper, it contains all one-parameter semigroups of $\phi$ with common fixed point lying in the open disk $\interior{\D}$.
At some point, we would like to show that this semigroup in fact contains all $\phi$ mapping onto Jordan domains with $C^\infty$ boundary by carefully analyzing the approximation numbers of $W_\phi$.

\subsubsection{A general theory of Segal CFT for degenerate Riemann surfaces}

Eventually, we would like to upgrade our construction of maps assigned to degenerate Riemann surfaces to a functorial field theory.
That is, one should be able to precisely describe a bordism category of degenerate Riemann surfaces, and construct examples of field theories using this bordism category as a source.
In the free fermion example, the maps that should be assigned to degenerate surfaces can be characterized via commutaiton relations with respect to a Hardy space, just as with the degenerate surfaces considered in this paper.

A related project is Henriques' partial construction of extended 2d functorial conformal field theories from Riemann surfaces with cusps \cite{Henriques14}.
Henriques uses a presentation of the category of complex bordisms which features a generator
$$
%\begin{tikzpicture}[baseline={([yshift=-.5ex]current bounding box.center)}]
%\filldraw[fill=red!10!blue!20!gray!30!white] (180:1cm) .. controls ++(0:.6cm) and ++(180:.4cm) .. (90:.5cm) .. controls ++ (0:.4cm) and ++(180:.6cm) .. (0:1cm) .. controls ++(180:.6cm) and ++(0:.4cm) .. (270:.5cm) .. controls ++(180:.4cm) and ++(0:.6cm) .. (180:1cm);
%\filldraw[fill=white] (0,0) circle (.2cm);
%\end{tikzpicture}
%\,\,  \qquad \mbox { and } \qquad \,\,
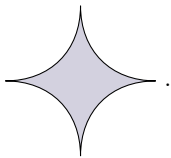
\begin{tikzpicture}[baseline={([yshift=-.5ex]current bounding box.center)}]
\filldraw[fill=red!10!blue!20!gray!30!white] (180:1cm) .. controls ++(0:.6cm) and ++(270:.6cm) .. (90:1cm) .. controls ++(270:.6cm) and ++(180:.6cm) .. (0:1cm) .. controls ++(180:.6cm) and ++(90:.6cm) ..  (270:1cm) .. controls ++(90:.6cm) and ++(0:.6cm) .. (180:1cm);
\end{tikzpicture}
\,\,.
$$
In the language of our paper, this generator corresponds to a degenerate Riemann surface
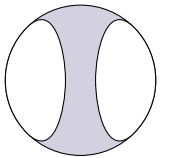
\begin{equation}\label{eqnMyVersionFourPointed}
%\begin{tikzpicture}[baseline={([yshift=-.5ex]current bounding box.center)}]
%	\coordinate (a) at (120:1cm);
%	\coordinate (b) at (240:1cm);
%	\coordinate (c) at (180:.25cm);
%% BIG DISK
%	\fill[fill=red!10!blue!20!gray!30!white] (0,0) circle (1cm);
%	\draw (0,0) circle (1cm);
%% CURVED BOUNDARY REGION
%	\fill[fill=white] (a)  .. controls ++(210:.6cm) and ++(90:.4cm) .. (c) .. controls ++(270:.4cm) and ++(150:.6cm) .. (b) -- ([shift=(240:1cm)]0,0) arc (240:480:1cm);
%	\draw ([shift=(240:1cm)]0,0) arc (240:480:1cm);
%	\draw (a) .. controls ++(210:.6cm) and ++(90:.4cm) .. (c);
%	\draw (b) .. controls ++(150:.6cm) and ++(270:.4cm) .. (c);
%% INNER DISK
%	\filldraw[fill=white] (180:.65cm) circle (.2cm); 
%% COORDINATE LABELS
%%	\node at (a) {(a)};
%%	\node at (b) {(b)};
%%	\node at (c) {(c)};
%\end{tikzpicture}
%\,\,  \qquad \mbox { and } \qquad \,\,
\begin{tikzpicture}[baseline={([yshift=-.5ex]current bounding box.center)}]
	\filldraw[fill=red!10!blue!20!gray!30!white] (0,0) circle (1cm);
	\filldraw[fill=white] (225:1cm) .. controls ++(-45:.5cm) and ++(270:.4cm) .. (180:.2cm) .. controls ++(90:.4cm) and ++(45:.5cm) .. (135:1cm) arc (135:225:1cm);
	\filldraw[fill=white] (45:1cm) .. controls ++(135:.5cm) and ++(90:.4cm) .. (0:.2cm) .. controls ++(270:.4cm) and ++(225:.5cm) .. (-45:1cm) arc (-45:45:1cm);
\end{tikzpicture}
\,\,.
\end{equation}

We did not discuss degenerate surfaces of this type, but the results of this paper allow one to assign bounded maps to such a degenerate surface in the free fermion example as long as the maps corresponding to the individual annuli
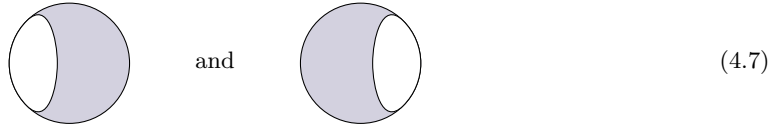
\begin{equation}\label{eqnFilledInLeftAndRight}
\begin{tikzpicture}[baseline={([yshift=-.5ex]current bounding box.center)},scale=.8]
	\filldraw[fill=red!10!blue!20!gray!30!white] (0,0) circle (1cm);
	\filldraw[fill=white] (225:1cm) .. controls ++(-45:.5cm) and ++(270:.4cm) .. (180:.2cm) .. controls ++(90:.4cm) and ++(45:.5cm) .. (135:1cm) arc (135:225:1cm);
	%\filldraw[fill=white] (45:1cm) .. controls ++(135:.5cm) and ++(90:.4cm) .. (0:.2cm) .. controls ++(270:.4cm) and ++(225:.5cm) .. (-45:1cm) arc (-45:45:1cm);
\end{tikzpicture}
\,\,  \qquad \mbox { and } \qquad \,\,
\begin{tikzpicture}[baseline={([yshift=-.5ex]current bounding box.center)},scale=.8]
	\filldraw[fill=red!10!blue!20!gray!30!white] (0,0) circle (1cm);
	%\filldraw[fill=white] (225:1cm) .. controls ++(-45:.5cm) and ++(270:.4cm) .. (180:.2cm) .. controls ++(90:.4cm) and ++(45:.5cm) .. (135:1cm) arc (135:225:1cm);
	\filldraw[fill=white] (45:1cm) .. controls ++(135:.5cm) and ++(90:.4cm) .. (0:.2cm) .. controls ++(270:.4cm) and ++(225:.5cm) .. (-45:1cm) arc (-45:45:1cm);
\end{tikzpicture}
\end{equation}
are bounded.

We briefly sketch a proof of this fact, which is similar to the proof of boundedness of operators assigned to degenerate pairs of pants in Theorem \ref{thmFermionPantsBoundedness}.

Given a degenerate Riemann surface such as the one in \eqref{eqnMyVersionFourPointed}, write it as $\D \setminus (\phi_1(\interior{\D}) \cup \phi_2(\interior{\D}))$ for Riemann maps $\phi_i$.
If both annuli $\D \setminus \phi_i(\interior{\D})$ have associated bounded operators, then it must be that $\prod_{n=1}^\infty a_n(W_{\phi_i}) < \infty$ for $i =1,2$. 
Equivalently, this means that each $W_{\phi_i}$ can be written as the sum of a contraction and a trace class operator.

Now if we set $W\xi = (W_{\phi_1}\xi, W_{\phi_2}\xi)$, we have 
$$
WW^* = \begin{pmatrix} W_{\phi_1}W_{\phi_1}^* & W_{\phi_1}W_{\phi_2}^*\\
W_{\phi_2}W_{\phi_1}^* & W_{\phi_2}W_{\phi_2}^*
\end{pmatrix}
.
$$
Since $\phi_1(S^1) \cap \phi_2(S^1) = \emptyset$, it is straightforward to check that the off-diagonal entries of $WW^*$ are trace class (in fact, they are integral operators with a smooth kernel).
On the other hand, $W_{\phi_i}W_{\phi_i}^*$ can be written as the sum of a contraction and a trace class, so the same is true of $WW^*$, and hence $W$.
Thus $W \oplus \tilde W$ defines an admissible operator in $\cB(H, H \oplus H)$, where $H = L^2(S^1)$ and admissibility is understood with repsect to the Hardy space projections $p$ and $p \oplus p$.

Arguing as in Section \ref{subsecCalculationSegalCFT}, one may show that the adjoint of the operator which should be assigned the the degenerate surface in \eqref{eqnMyVersionFourPointed} is the implementing operator associated to $(W \oplus \tilde W, \hat \Omega)$, for a vector $\hat \Omega$ which is assigned to a non degenerate Riemann surface by the free fermion Segal CFT.
Boundedness now follows as in Theorem \ref{thmFermionPantsBoundedness}.

\subsubsection{More examples and constructions}

While the class of vertex operator superalgebras which can be embedded unitarily in $(\F^0)^{\otimes N}$ is quite large, there are important examples for which we do not know of such an embedding.
Most notably, the lists of lattice, WZW and Virasoro models discussed in Examples \ref{exLattice}, \ref{exWZW} and \ref{exVirasoro} are incomplete.
Ideally, we would like a general argument for each of the three cases.

It would also be desirable to show that the property of having bounded localized vertex operators is preserved under additional operations, for example ``nice'' extensions.
In order to prove anything about localized vertex operators for extensions, we would require a broader notion of localized vertex operators which includes module and intertwining operators.

\subsubsection{Modules and intertwining operators}

In this paper we only considered operators assigned to degenerate Riemann surfaces in the vacuum sector, and we saw that the operators that should be assigned were related to vertex operators.
To assign operators to degenerate Riemann surfaces with boundary components labeled by sectors, we would need to generalize our results to intertwining operators.
Bounded localized intertwining operators will play an important role in relating the tensor product of VOA modules with the tensor product of representations of the associated conformal net, in the same way that Wassermann used the boundedness of certain smeared intertwining operators in his proof of the fusion rules for the $SU(N)_k$ conformal nets in \cite{Wa98}.
We begin the study of bounded localized intertwining operators in the sequel article \cite{Ten18ax}.

%%%%%%%%%%%%%%%%%%%%%%%%%%%%%%%%%%%%%%%%%%%%%%
%%%%%%%%%%%%%%%%%%%%%%%%%%%%%%%%%%%%%%%%%%%%%%
%%%%%%%%%%%%%%%%%%%%%%%%%%%%%%%%%%%%%%%%%%%%%%
%%%%%%%%%%%%%%%%%%%%%%%%%%%%%%%%%%%%%%%%%%%%%%
%%%%%%%%%%%%%%%%%%%%%%%%%%%%%%%%%%%%%%%%%%%%%%
%%%%%%%%%%%%%%%%%%%%%%%%%%%%%%%%%%%%%%%%%%%%%%

\section{Implementing operators}
\label{secImplementingOperators}

Consider the following general scenario.
Let $H$ and $K$ be separable Hilbert spaces, and let $p$ and $q$ be projections on $H$ and $K$, respectively. 
Assume that $pH$ and $(1-p)H$ are infinite dimensional.
With this data, we can form the Fock spaces $\F_{H,p}$ and $\F_{K,q}$, which carry representations of $\CAR(H)$ and $\CAR(K)$, respectively.

Fix an orthonormal basis $\{\xi_i\}_{i \in \Z}$ for $H$, and assume that $\xi_i \in pH$ when $i \ge 0$, $\xi_i \in (1-p)H$ when $i < 0$.
Such a basis for $H$ is said to be \emph{compatible with $p$}.
Recall that if $I = \{i_1, \ldots, i_n\} \subset \Z$ with
$
i_1 < i_2 < \cdots < i_n,
$
and if $\{\psi_i\} \subset H$ is a family of vectors indexed by $I$, then we write
\begin{equation}\label{eqnFermionProductNotation}
a(\psi_I) = a(\psi_{i_1}) \cdots a(\psi_{i_n}) \in \CAR(H).
\end{equation}

The Fock space $\F_{H,p}$ has an orthonormal basis $a(\xi_J)a(\xi_I)^*\Omega_p$, where $I$ runs over finite subsets of $\Z_{\ge 0}$ and $J$ runs over finite subsets of $\Z_{< 0}$.

\begin{Definition}\label{defImplementer}
Let $H, K, p, q,$ and  $\xi_i$ be as above.
Let $r \in \cB(H, K)$ and $\hOmega \in \F_{K, q}$.
Then the \emph{implementing operator associated} to $(r, \hOmega)$ is the densely defined linear map $R:\F_{H, p} \to \F_{K, q}$ given by
$$
Ra(\xi_J)a(\xi_I)^*\Omega_p = a(r \xi_J)a(r \xi_I)^*\hOmega.
$$
\end{Definition}

We now set out to establish a sufficient condition for an implementing operator to be bounded.

\begin{Definition}\label{defDiagonalExpectation}
Let $H$ and $K$ be Hilbert spaces, and let $p$ and $q$ be projections on $H$ and $K$, respectively. 
For $r \in \cB(H, K)$, define the \emph{diagonal expectation} of $r$ by $E(r) = qrp + (1-q)r(1-p)$.
The class of {\it admissible} operators $\cA(H, K)$ is defined to be those $r \in \cB(H,K)$ with $qr(1-p)$ trace class, and which have the property that there exist $a,x \in \cB(H,K)$ with $\norm{a} \le 1$ and $x$ trace class such that $E(r) = a + x$.
\end{Definition}
In other words, if we think of elements of $\cB(H,K)$ as $2 \times 2$ matrices with respect to the decompositions $pH \oplus (1-p)H$ and $qK \oplus (1-q)K$, then for $r \in \cB(H,K)$ to be admissible we require the top right entry of $r$ to be trace class, and the diagonal entries to have a decomposition as $(${\it contraction}$) + (${\it trace class}$)$.

Definition \ref{defDiagonalExpectation} depends on the projections $p$ and $q$, which we omit from the notation as they will remain fixed in our applications.

In a moment, we will give Theorem \ref{thmAdmissibleBoundedness}, the main result of Section \ref{secImplementingOperators} which partially characterizes boundedness of implementing maps in terms of admissibility. 
First, we need to briefly recall some facts about the representation theory of the $\CAR$ algebra (see Section \ref{subsecFermionicFockSpace}).

Let $q^\prime \in \cB(K)$ be a projection, and assume that $q^\prime - q$ is Hilbert-Schmidt. 
Then there is a unique-up-to-scalar vector $\hat \Omega_{q^\prime} \in \F_{K,q}$ such that 
\begin{equation}\label{eqnqVacuum}
a(f)\hOmega_{q^\prime} = a(g)^*\hOmega_{q^\prime} = 0
\end{equation}
for every $f \in q^\prime K$ and every $g \in (1-q^\prime)K$. 
When $q^\prime = q$, then $\hOmega_{q^\prime}$ is just the ordinary vacuum vector $\Omega_q \in \F_{K,q}$.

\begin{Theorem}\label{thmAdmissibleBoundedness}
Let $H$ and $K$ be separable Hilbert spaces, and let $p$ and $q$ be projections on $H$ and $K$, respectively, with $pH$ and $(1-p)H$ infinite dimensional. 
Let $\{\xi_i\}_{i \in \Z}$ be an orthonormal basis for $H$ which is compatible with $p$.
Let $r \in \cB(H,K)$, and assume that $qr(1-p)$ is trace class. 
Let $q^\prime$ be a projection on $K$ with $q^\prime - q$ trace class, and let $\hat \Omega_{q^\prime} \in \F_{K,q}$ be a non-zero vector satisfying \eqref{eqnqVacuum}.  
Then the implementing operator associated to $(r,\hOmega_{q^\prime})$ is bounded if and only if $r \in \cA(H,K)$. 
\end{Theorem}

We will prove Theorem \ref{thmAdmissibleBoundedness} with several lemmas giving operations under which the boundedness of the implementer for $(r, \hOmega)$ is preserved.

First, we check that the boundedness of the implementing operator is independent of the choice of basis used to define it.

\begin{Proposition}\label{propImplementerIndependentOfBasis}
Let $H,K,p,q$ be as in Theorem \ref{thmAdmissibleBoundedness}, and let $\hOmega \in \F_{K,q}$ and $r \in \cB(H,K)$ be arbitrary.
Then the boundedness of the implementing operators associated to $(r, \hOmega)$ is independent of the choice of basis $\xi_i$. 
When the implementing operators for two choices of bases are bounded, then their extensions to $\F_{H,p}$ coincide.
\end{Proposition}
\begin{proof}
Let $\xi_i^{(1)}$ and $\xi_i^{(2)}$ be two orthonormal bases for $H$, and densely define linear maps $R^{(1)}$ and $R^{(2)}$ by
$$
R^{(m)}a(\xi^{(m)}_J)a(\xi^{(m)}_I)^*\Omega_p = a(r \xi^{(m)}_J)a(r \xi^{(m)}_I)^* \hOmega.
$$
Assume that $R^{(1)}$ extends to a bounded map on all of $\F_{H,p}$. 
Fix finite subsets $I \subset \Z_{\ge 0}$ and $J \subset \Z_{<0}$, and write
$$
a(\xi^{(2)}_J)a(\xi^{(2)}_I)^*\Omega_p = \sum_{I^\prime,J^\prime} c_{I^\prime,J^\prime} a(\xi^{(1)}_{J^\prime})a(\xi^{(1)}_{I^\prime})^*\Omega_p
$$
where $I^\prime$ runs over finite subsets of $\Z_{\ge 0}$, $J^\prime$ runs over finite subsets of $\Z_{< 0}$, and $c_{I^\prime, J^\prime} \in \C$.
Then we have
$$
a(r\xi^{(2)}_J)a(r\xi^{(2)}_I)^*\hOmega = \sum_{I^\prime,J^\prime} c_{I^\prime,J^\prime} a(r\xi^{(1)}_{J^\prime})a(r\xi^{(1)}_{I^\prime})^*\hOmega
$$
for the same coefficients $c_{I^\prime,J^\prime}$.
We can now calculate
\begin{align*}
R^{(1)} a(\xi^{(2)}_J)a(\xi^{(2)}_I)^*\Omega_p &= R^{(1)}\sum_{I^\prime,J^\prime} c_{I^\prime,J^\prime} a(\xi^{(1)}_{J^\prime})a(\xi^{(1)}_{I^\prime})^*\Omega_p\\
&= \sum_{I^\prime,J^\prime} c_{I^\prime,J^\prime} a(r\xi^{(1)}_{J^\prime})a(r\xi^{(1)}_{I^\prime})^*\hOmega\\
&= a(r\xi^{(2)}_J)a(r\xi^{(2)}_I)^*\hOmega\\
&= R^{(2)} a(\xi^{(2)}_J)a(\xi^{(2)}_I)^*\Omega_p.
\end{align*}
Since $R^{(1)}$ and $R^{(2)}$ agree on a basis, $R^{(2)}$ is also bounded and $R^{(1)} = R^{(2)}$.
\end{proof}

\begin{Lemma}\label{lemOffDiagonalTraceClassPerturbation}
Let $H,K,p,q,\xi_i$ be as in Theorem \ref{thmAdmissibleBoundedness}, and let $r \in \cB(H,K)$ and $\hOmega \in \F_{K,q}$ be arbitrary. 
Let $x \in B(H, K)$ be a trace class operator with $xp = 0$. 
Then the implementer associated to $(r, \hOmega)$ is bounded if and only if the implementer associated to $(r+x, \hOmega)$ is.
\end{Lemma}
\begin{proof}
Let $R$ be the implementer assigned to $(r, \hOmega)$ , and let $T$ be the implementer assigned to $(r+x, \hOmega)$. 
Assume that $R$ defines a bounded operator, and we will prove that $T$ is bounded as well. 
By Proposition \ref{propImplementerIndependentOfBasis}, we can choose any orthonormal basis $\xi_i$ for $H$ to define $R$ and $T$ with respect to, as long as $\xi_i \in pH$ when $i \ge 0$ and $\xi_j \in (1-p)H$ when $j < 0$. 

Since $xp = 0$ and $x$ is compact, the singular value decomposition of $x$ yields an orthonormal basis $\{\xi_j\}_{j < 0}$ for $(1-p)H$, an orthonormal set $\{\eta_j\}_{j < 0} \subset K$, and scalars $\lambda_j \in \C$ with $x \xi_j = \lambda_j \eta_j$. Moreover, since $x$ is trace class we have $\sum \abs{\lambda_j} < \infty$. 
Extend $\xi_j$ to an orthonormal basis $\{\xi_j\}_{j \in \Z}$ for $H$.

For $L \subset \Z_{< 0}$ a finite subset, set $\lambda_L = \prod_{\ell \in L} \lambda_\ell$. 
Recall that if we have $L = \{\ell_1, \ldots, \ell_n\}$ with $\ell_1 < \cdots < \ell_n$, and if $\{\psi_\ell\}_{\ell \in L}$ is a family of vectors indexed by $L$, then we set
$$
a(\psi_L) = a(\psi_{\ell_1}) \cdots a(\psi_{\ell_n}).
$$

We will now show that
\begin{equation}\label{eqnPerturbedSumExpansion}
T = \sum_{L  \subset \Z_{<0}} \lambda_L a(\eta_L)R a(\xi_L)^*,
\end{equation}
where the sum runs over finite subsets $L$.
Since $\norm{a(\eta_\ell)^*} = \norm{a(\xi_\ell)} = 1$, 
$$
\sum_{L  \subset \Z_{<0}} \norm{\lambda_L a(\eta_L)R a(\xi_L)^*} \le \norm{R} \sum_L \abs{\lambda_L}  = \norm{R}\prod_{\ell \in \Z_{< 0}} (1 + \abs{\lambda_\ell}),
$$
and so the right-hand side of \eqref{eqnPerturbedSumExpansion} converges absolutely in norm.
Thus to verify \eqref{eqnPerturbedSumExpansion}, and in particular that $T$ is bounded, it suffices to check that both sides agree when applied to basis vectors $a(\xi_J)a(\xi_I)^*\Omega$, where $J \subset \Z_{< 0}$ and $I \subset \Z_{\ge 0}$ are finite sets.

For $J \subset \Z$ a finite subset, $\{\psi_j\}_{j \in J}$ a family of vectors and $L \subseteq J$, let $\epsilon_{L,J} \in \{ \pm 1\}$ be such that
$$
a(\psi_J) = \epsilon_{L,J} a(\psi_L)a(\psi_{J \setminus L}).
$$
Note that $\epsilon_{L,J}$ is independent of the $\psi_j$. 
With this notation, for $J \subset \Z_{<0}$ and $I \subset \Z_{\ge 0}$ finite subsets we have
\begin{align*}
Ta(\xi_J)a(\xi_I)^*\Omega_p &= a((r+x)\xi_J)a(r\xi_I)^*\hOmega\\
&= \sum_{L \subset J} \epsilon_{L,J} a(x \xi_L) a(r \xi_{J \setminus L}) a(r \xi_I)^* \hOmega\\
&= \sum_{L \subset J} \epsilon_{L,J} a(x \xi_L) R a(\xi_{J \setminus L}) a(\xi_I)^* \Omega_p\\
&= \sum_{L \subset J} \lambda_L a(\eta_L) Ra(\xi_L)^* a(\xi_J)a(\xi_I)^*\Omega_o\\
&=\left(\sum_{L \subset \Z_{< 0}} \lambda_L a(\eta_L) R a(\xi_L)^*\right)  a(\xi_J)a(\xi_I)^*\Omega_p.
\end{align*}
In the last equality, we used that $a(\xi_L)^*a(\xi_J)a(\xi_I)^*\Omega_p = 0$ when $L \subset \Z_{< 0}$ but $L \not \subset J$.
This establishes \eqref{eqnPerturbedSumExpansion} and completes the proof.
\end{proof}

The next lemma establishes Theorem \ref{thmAdmissibleBoundedness} in the case where the $\hat \Omega_{q^\prime} = \Omega_q$, the vacuum vector in $\F_{K, q}$.

\begin{Lemma}\label{lemAdmissibleBoundednessVacuum}
Let $H$, $K$, $p$, $q$, and $\xi_i$ be as in Theorem \ref{thmAdmissibleBoundedness}. 
Let $r \in \cB(H,K)$ and assume that $qr(1-p)$ is trace class. 
Then the implementer $R$ associated to $(r,\Omega_q)$ is bounded if and only if $r \in \cA(H, K)$.
\end{Lemma}
\begin{proof}
We must show that $R$ is bounded if and only if $qrp + (1-q)r(1-p) =: E(r)$ can be written as the sum of a contraction plus a trace class. 

Recall that $R$ is given by the formula
\begin{equation}\label{eqnImplementerAgainstVacuum}
Ra(\xi_J)a(\xi_I)^*\Omega_p = a(r \xi_J) a(r \xi_I)^* \Omega_q
\end{equation}
where as usual $I$ is a finite subset of $\Z_{\ge 0}$ and $J$ is a finite subset of $\Z_{< 0}$.
Since $a(g)^*\Omega_q = 0$ when $g \in (1-q)K$, one can see that $R$ is independent of $(1-q)rp$, and so we assume without loss of generality that $(1-q)rp = 0$. 

By Lemma \ref{lemOffDiagonalTraceClassPerturbation} and the assumption that $qr(1-p)$ is trace class, the boundedness of \eqref{eqnImplementerAgainstVacuum} is unchanged by subtracting off $qr(1-p)$. 
Thus it suffices to prove the lemma under the assumption that $r$ is block diagonal, i.e. that $r = E(r) = qrp + (1-q)r(1-p)$.

There is a natural unitary $U_H:\F_{H,p} \to \Lambda (1-p)H \otimes \Lambda \overline{pH}$  given by
$$
U_H a(\xi_J) a(\xi_I)^*\Omega_p = a(\xi_J)\Omega \otimes a(\xi_I)^*\Omega
$$
for $J \subset \Z_{< 0}$ and $I \subset \Z_{\ge 0}$ finite subsets.
Here, we are thinking of $\Lambda (1-p)H = \F_{(1-p)H,0}$ and $\Lambda \overline{pH} = \F_{pH, 1_{pH}}$ when we write the actions of $\CAR(pH)$ and $\CAR((1-p)H)$ on these spaces.
Thus $a(\xi_J)\Omega$ gives an orthonormal basis for $\Lambda (1-p)H$ indexed by finite subsets $J \subset \Z_{< 0}$, and $a(\xi_I)^*\Omega$ gives an orthonormal basis for $\Lambda \overline{pH}$ indexed by finite subsets $I \subset \Z_{\ge 0}$.

Let $U_K:\F_{K,q} \to  \Lambda (1-q)K \otimes \Lambda \overline{qK}$ be the unitary given by
$$
U_K a(\psi_J) a(\psi^\prime_I)^*\Omega_q = a(\psi_J)\Omega \otimes a(\psi^\prime_I)^*\Omega
$$
for all finite families of vectors $\{\psi^\prime_i\}_{i \in I} \subset qK$ and $\{\psi_j\}_{j \in J} \subset (1-q)K$.

Since $r$ is block diagonal, we have $U_KRU_H^* = R_- \otimes R_+$, where $R_-: \Lambda (1-p)H \to \Lambda (1-q)K$ and $R_+: \Lambda \overline{pH} \to \Lambda \overline{qK}$ are given by
$$
R_-a(\xi_J)\Omega = a((1-q)r(1-p)\xi_J)\Omega, \qquad R_+a(\xi_I)^*\Omega = a(qrp\xi_I)^*\Omega
$$
for finite subsets $J \subset \Z_{< 0}$ and $I \subset \Z_{\ge 0}$.

Thus to complete the proof we must show that $R_- \otimes R_+$ is bounded if and only if $r$ can be written as a sum of a contraction plus a trace class, or equivalently if $(1-q)r(1-p)$ and $qrp$ can both be written as the sum of a contraction plus a trace class.
We will prove that $R_-$ is bounded if and only if $(1-q)r(1-p)$ can be written as contraction plus trace class, and the corresponding statement for $R_+$ and $qrp$ is identical.

We begin with a small piece of notation.
If $L_1$ and $L_2$ are Hilbert spaces, and $t \in \cB(L_1,L_2)$, we will write $\Lambda(t):\Lambda(L_1) \to \Lambda(L_2)$ for the densely defined operator given on finite wedge products by
$$
\Lambda(t)(\psi_1 \wedge \cdots \wedge \psi_n) = t\psi_1 \wedge \cdots \wedge t\psi_n.
$$
Note that $\Lambda(t)$ is the restriction of 
$$
\bigoplus_{n=0}^\infty t^{\otimes n} \in \bigoplus_{n=0}^\infty \cB\big({L_1}^{\otimes n}, \, {L_2}^{\otimes n}\big)
$$
to an invariant subspace, so $\norm{\Lambda(t)} \le 1$ when $\norm{t} \le 1$.
However, $\Lambda(t)$ can be bounded even when $t$ is not a contraction.

To simplify notation, let $H^\prime = (1-p)H$, $K^\prime = (1-q)K$, and $s = (1-q)r(1-p) \in \cB(H^\prime,K^\prime)$.
In this notation,  $R_- = \Lambda(s)$, and we must show that $R_-$ is bounded if and only if $s$ can be written as the sum of a contraction and a trace class operator.

Assume first that $s = b + x$ with $\norm{b} \le 1$ and $x$ trace class.
Since $x$ is trace class, the singular value decomposition yields an orthonormal basis $\{\xi_j\}_{j \in \Z_{< 0}}$ for $H^\prime$, an orthonormal set $\{\eta_j\}_{j \in \Z_{< 0}}$, and scalars $(\lambda_j) \in \ell^1(\Z_{< 0})$ such that $x \xi_j = \lambda_j\eta_j$.
One can then verify, just as in the proof of Lemma \ref{lemOffDiagonalTraceClassPerturbation}, that
$$
R_- = \sum_{L \subset \Z_{< 0}} \lambda_L \,a(\eta_L)\Lambda(b)a(\xi_L)^*
$$
where the sum is indexed by finite subsets $L$ and $\lambda_L = \prod_{\ell \in L} \lambda_\ell$.
Hence 
$$
\norm{R_-} \le \sum_L \abs{\lambda_L} = \prod_{\ell < 0} (1 + \abs{\lambda_\ell}) < \infty
$$
and so $R_-$ is bounded.

Now assume that $R_-$ is bounded, and we will prove that $s$ can be written as the sum of a contraction and a trace class operator.
Let $s = u\abs{s}$ be the polar decomposition, and observe that $\Lambda(\abs{s}) = \Lambda(u^*)\Lambda(s)$, and thus $\Lambda(\abs{s})$ is bounded.
Note that $\Lambda(\abs{s}) = \abs{\Lambda(s)} \ge 0$.
Let $p_{\le 1}$ be the spectral projection for $\abs{s}$ corresponding to the interval $[0,1]$, and let $p_{> 1} = 1-p_{\le 1}$.
Observe that $b:=\abs{s}p_{\le 1} + p_{> 1}$ is a contraction, and let $x = \abs{s} - b$.
Then $x$ is supported on $p_{> 1}H^\prime$, and $x \ge 0$.

Let $\psi_1, \ldots, \psi_n \in p_{> 1}H^\prime$ be an arbitrary orthonormal family.
Then we have
\begin{align*}
\norm{\Lambda(\abs{s})} &\ge \ip{\Lambda(\abs{s}) \psi_1 \wedge \cdots \wedge \psi_n, \psi_1 \wedge \cdots \wedge \psi_n}\\
&= \ip{\Lambda(x + 1) \psi_1 \wedge \cdots \wedge \psi_n, \psi_1 \wedge \cdots \wedge \psi_n}\\
&= \sum_{L \subseteq \{1, \ldots, n\}} \det (\ip{x\psi_i,\psi_j})_{i,j \in L}\\
&\ge \sum_{j = 1}^n \ip{x\psi_j, \psi_j}.
\end{align*}
Hence $x$ is trace class, with $\norm{x}_1  \le \norm{\Lambda(\abs{s})}$. 
We have therefore produced a decomposition $\abs{s} = b + x$ with $\norm{b} \le 1$ and $x$ trace class.
It follows that $s = ub + ux$ is a decompostion of the same type, which was to be shown.
\end{proof}

\begin{Lemma}\label{lemChangeVacuumVector}
Let $H$, $K$, $p$, $q$ and $\xi_i$ be as in Theorem \ref{thmAdmissibleBoundedness}. 
Let $r \in \cB(H, K)$, and assume that $qr(1-p)$ is trace class. 
Let $q^\prime$ be a projection on $K$ with $q - q^\prime$ trace class.  
Then the implementer associated to $(r, \Omega_q)$ is bounded if and only if the implementer associated to $(r, \hat \Omega_{q^\prime})$ is.
\end{Lemma}
\begin{proof}
Let $R$ be the implementer associated to $(r, \hat \Omega_{q^\prime})$.

Let $u \in \cU_{res}(K,q)$ be a unitary with $q^\prime  = uqu^*$, and let $U \in \cU(\F_{K,q})$ be the image of $u$ under the basic representation (see Secion \ref{subsecFermionicFockSpace}). 
Then $U\Omega_q = \hat \Omega_{q^\prime}$ and $Ua(f)U^* = a(uf)$ for all $f \in K$.

Then we see that
\begin{align*}
Ra(\xi_J)a(\xi_I)^*\Omega_p &= a(r \xi_J) a(r \xi_I)^* \hat \Omega_{q^\prime}\\
&= a(r \xi_J) a(r \xi_I)^* U\Omega_q\\
&= U a(u^*r \xi_J) a(u^*r \xi_I)^*\Omega_q.
\end{align*}
Thus $R$ is bounded if and only if the implementer associated to $(u^*r, \Omega_q)$ is bounded.
Our problem is then reduced to showing that, under the assumption that $qr(1-p)$ is trace class, the implementer associated to $(r, \Omega_q)$ is bounded if and only if the implementer associated to $(u^*r, \Omega_q)$ is bounded, where $u \in \cU(K)$ has the property that $uqu^* - q$ is trace class.

By Lemma \ref{lemAdmissibleBoundednessVacuum}, it suffices to show that if $r \in \cA(H,K)$, then $u^*r \in \cA(H,K)$ as well. 
Assume that $r \in \cA(H,K)$. 
Then $qr(1-p)$ is trace class, and by assumption $[q,u^*] = u^*(uqu^* - q)$ is as well. 
Hence 
$$
qu^*r(1-p) = u^*qr(1-p) + [q,u^*]r(1-p)
$$
is trace class as well.

Similarly, we have
\begin{align*}
E(u^*r) &= qu^*rp + (1-q)u^*r(1-p)\\
 &= u^*\big(qrp + (1-q)r(1-p)\big) + [q, u^*]r(2p-1)\\
 &= u^*E(r) + [q,u^*]r(2p-1),
\end{align*}
and since since $E(r)$ can be written as the sum of a contraction and a trace class operator, so can $E(u^*r)$.
This establishes that $u^*r \in \cA(H,K)$, and completes the proof.
\end{proof}

We can now assemble the above lemmas to give  a short proof of Theorem \ref{thmAdmissibleBoundedness}.

\begin{proof}[Proof of Theorem \ref{thmAdmissibleBoundedness}:]
By Lemma \ref{lemChangeVacuumVector}, the implementer associated to $(r, \hat \Omega_{q^\prime})$ is bounded if and only if the implementer associated to $(r, \Omega_q)$ is bounded, and by Lemma \ref{lemAdmissibleBoundednessVacuum}, this implementer is bounded if and only if $r \in \cA(H,K)$.
\end{proof}

\bibliographystyle{alpha}
\bibliography{../../ffbib}

\newcommand{\etalchar}[1]{$^{#1}$}
\def\lfhook#1{\setbox0=\hbox{#1}{\ooalign{\hidewidth
  \lower1.5ex\hbox{'}\hidewidth\crcr\unhbox0}}}
\begin{thebibliography}{CKLW18}

\bibitem[AL17]{AiLin17}
Chunrui Ai and Xingjun Lin.
\newblock On the unitary structures of vertex operator superalgebras.
\newblock {\em Journal of Algebra}, 487:217--243, 2017.

\bibitem[BP78]{BerksonPorta78}
Earl Berkson and Horacio Porta.
\newblock Semigroups of analytic functions and composition operators.
\newblock {\em Michigan Math. J.}, 25(1):101--115, 1978.

\bibitem[CHK{\etalchar{+}}15]{CarpiHillierKawahigashiLongoXu15}
Sebastiano Carpi, Robin Hillier, Yasuyuki Kawahigashi, Roberto Longo, and Feng
  Xu.
\newblock {$N=2$} superconformal nets.
\newblock {\em Comm. Math. Phys.}, 336(3):1285--1328, 2015.

\bibitem[CKL08]{CaKaLo08}
Sebastiano Carpi, Yasuyuki Kawahigashi, and Roberto Longo.
\newblock Structure and classification of superconformal nets.
\newblock {\em Ann. Henri Poincar\'e}, 9(6):1069--1121, 2008.

\bibitem[CKLW18]{CKLW18}
Sebastiano Carpi, Yasuyuki Kawahigashi, Roberto Longo, and Mih{\'a}ly Weiner.
\newblock From vertex operator algebras to conformal nets and back.
\newblock {\em Mem. Amer. Math. Soc.}, 254(1213), 2018.

\bibitem[CW]{CarpiWeinerLocal}
Sebastiano Carpi and Mih{\'a}ly Weiner.
\newblock Local energy bounds and representations of conformal nets.
\newblock {\em In preparation}.

\bibitem[CWX]{CarpiWeinerXu}
Sebastiano Carpi, Mih{\'a}ly Weiner, and Feng Xu.
\newblock From vertex operator algebra modules to representations of conformal
  nets.
\newblock {\em In preparation}.

\bibitem[DL14]{DongLin14}
Chongying Dong and Xingjun Lin.
\newblock Unitary vertex operator algebras.
\newblock {\em J. Algebra}, 397:252--277, 2014.

\bibitem[Dur83]{Duren83}
Peter~L. Duren.
\newblock {\em Univalent functions}, volume 259 of {\em Grundlehren der
  Mathematischen Wissenschaften [Fundamental Principles of Mathematical
  Sciences]}.
\newblock Springer-Verlag, New York, 1983.

\bibitem[Dyn52]{DynkinIndex}
E.~B. Dynkin.
\newblock Semisimple subalgebras of semisimple {L}ie algebras.
\newblock {\em Mat. Sbornik N.S.}, 30(72):349--462 (3 plates), 1952.

\bibitem[FH05]{FewsterHollands}
Christopher~J. Fewster and Stefan Hollands.
\newblock Quantum energy inequalities in two-dimensional conformal field
  theory.
\newblock {\em Rev. Math. Phys.}, 17(5):577--612, 2005.

\bibitem[FLM88]{FLM88}
Igor Frenkel, James Lepowsky, and Arne Meurman.
\newblock {\em Vertex operator algebras and the {M}onster}, volume 134 of {\em
  Pure and Applied Mathematics}.
\newblock Academic Press, Inc., Boston, MA, 1988.

\bibitem[GKO86]{GKO}
P.~Goddard, A.~Kent, and D.~Olive.
\newblock Unitary representations of the {V}irasoro and super-{V}irasoro
  algebras.
\newblock {\em Comm. Math. Phys.}, 103(1):105--119, 1986.

\bibitem[Gol85]{Goldstein85}
Jerome~A. Goldstein.
\newblock {\em Semigroups of linear operators and applications}.
\newblock Oxford Mathematical Monographs. The Clarendon Press, Oxford
  University Press, New York, 1985.

\bibitem[GW85]{GoWa85}
Roe Goodman and Nolan~R. Wallach.
\newblock Projective unitary positive-energy representations of {${\rm
  Diff}(S\sp 1)$}.
\newblock {\em J. Funct. Anal.}, 63(3):299--321, 1985.

\bibitem[Hen14]{Henriques14}
Andr{\'e} Henriques.
\newblock Three-tier {CFT}s from {F}robenius algebras.
\newblock In {\em Topology and field theories}, volume 613 of {\em Contemp.
  Math.}, pages 1--40. Amer. Math. Soc., Providence, RI, 2014.

\bibitem[Kac98]{Kac98}
Victor Kac.
\newblock {\em Vertex algebras for beginners}, volume~10 of {\em University
  Lecture Series}.
\newblock American Mathematical Society, Providence, RI, second edition, 1998.

\bibitem[KL04]{KawahigashiLongo04}
Yasuyuki Kawahigashi and Roberto Longo.
\newblock Classification of local conformal nets. {C}ase {$c<1$}.
\newblock {\em Ann. of Math. (2)}, 160(2):493--522, 2004.

\bibitem[LS97]{LaszloSorger}
Yves Laszlo and Christoph Sorger.
\newblock The line bundles on the moduli of parabolic {$G$}-bundles over curves
  and their sections.
\newblock {\em Ann. Sci. \'Ecole Norm. Sup. (4)}, 30(4):499--525, 1997.

\bibitem[MTW18]{MorinelliTanimotoWeiner18}
Vincenzo Morinelli, Yoh Tanimoto, and Mih\'{a}ly Weiner.
\newblock Conformal covariance and the split property.
\newblock {\em Comm. Math. Phys.}, 357(1):379--406, 2018.

\bibitem[MTY18]{MasonTuiteYamskulna18}
Geoffrey Mason, Michael Tuite, and Gaywalee Yamskulna.
\newblock {$N=2$} and {$N=4$} subalgebras of super vertex operator algebras.
\newblock {\em J. Phys. A}, 51(6):064001, 22, 2018.

\bibitem[Pos03]{Po03}
Hessel~B. Posthuma.
\newblock {\em Quantization of Hamiltonian loop group actions}.
\newblock PhD thesis, University of Amsterdam, 2003.

\bibitem[PS86]{PrSe86}
Andrew Pressley and Graeme Segal.
\newblock {\em Loop groups}.
\newblock Oxford Mathematical Monographs. The Clarendon Press Oxford University
  Press, New York, 1986.
\newblock Oxford Science Publications.

\bibitem[Rud87]{BigRudin}
Walter Rudin.
\newblock {\em Real and complex analysis}.
\newblock McGraw-Hill Book Co., New York, third edition, 1987.

\bibitem[Sha93]{Shapiro93}
Joel~H. Shapiro.
\newblock {\em Composition operators and classical function theory}.
\newblock Universitext: Tracts in Mathematics. Springer-Verlag, New York, 1993.

\bibitem[ST74]{ShapiroTaylor73}
J.~H. Shapiro and P.~D. Taylor.
\newblock Compact, nuclear, and {H}ilbert-{S}chmidt composition operators on
  {$H\sp{2}$}.
\newblock {\em Indiana Univ. Math. J.}, 23:471--496, 1973/74.

\bibitem[Ten17]{Ten16}
James~E. Tener.
\newblock Construction of the unitary free fermion {S}egal {CFT}.
\newblock {\em Comm. Math. Phys.}, 355(2):463--518, 2017.

\bibitem[Ten18]{Ten18ax}
James~E. Tener.
\newblock Representation theory in chiral conformal field theory: from fields
  to observables.
\newblock {\em arXiv:1810.08168 [math-ph]}, 2018.

\bibitem[Thu74]{Thurston74}
William Thurston.
\newblock Foliations and groups of diffeomorphisms.
\newblock {\em Bull. Amer. Math. Soc.}, 80:304--307, 1974.

\bibitem[TL99]{TL99}
Valerio Toledano~Laredo.
\newblock Integrating unitary representations of infinite-dimensional {L}ie
  groups.
\newblock {\em J. Funct. Anal.}, 161(2):478--508, 1999.

\bibitem[Was98]{Wa98}
Antony Wassermann.
\newblock Operator algebras and conformal field theory. {III}. {F}usion of
  positive energy representations of {${\rm LSU}(N)$} using bounded operators.
\newblock {\em Invent. Math.}, 133(3):467--538, 1998.

\bibitem[Wei05]{Weiner05}
Mih{\'a}ly Weiner.
\newblock {\em Conformal covariance and related properties of chiral {QFT}}.
\newblock PhD thesis, Universit{\`a} di Roma ``Tor Vergata'', 2005.
\newblock arXiv:math/0703336.

\bibitem[Yam14]{Yamauchi2014}
Hiroshi Yamauchi.
\newblock Extended {G}riess algebras and {M}atsuo-{N}orton trace formulae.
\newblock In {\em Contributions in Mathematical and Computational Sciences},
  pages 75--107. Springer Science $+$ Business Media, 2014.

\end{thebibliography}

\Addresses

\end{document}